\documentclass[11pt]{article}
\usepackage{fullpage}
\usepackage{amsmath}
\usepackage{graphicx}
\usepackage{enumerate}
\usepackage{natbib}
\usepackage{url} % not crucial - just used below for the URL 
\usepackage{amsfonts}
\usepackage{upgreek}
\usepackage{mathtools}
\usepackage{bm}
\usepackage{hyperref}
\usepackage{color}
\usepackage{xcolor}
\usepackage{booktabs}
\usepackage{multirow}
\usepackage{makecell}
\usepackage[font=footnotesize]{caption}
\usepackage[utf8]{inputenc}
\usepackage{authblk}
\usepackage{amsthm,amssymb,enumitem,afterpage,setspace,array,verbatim,commath,mathrsfs,bm}

\numberwithin{equation}{section}
\newtheorem{theorem}{Theorem}

\DeclareMathOperator*{\argmin}{argmin}

\allowdisplaybreaks

\begin{document}

\onehalfspacing

\title{\bf Scalable Function-on-Scalar Quantile Regression \\ for Densely Sampled Functional Data}
\author[1]{Yusha Liu}
\author[2]{Meng Li}
\author[3]{Jeffrey S. Morris}
\affil[1]{\small Department of Biostatistics, University of North Carolina at Chapel Hill}
\affil[2]{\small Department of Statistics, Rice University}
\affil[3]{\small Division of Biostatistics, Department of Biostatistics, Epidemiology and Informatics, Perelman School of Medicine, University of Pennsylvania}
\date{}
\maketitle

\begin{abstract}
Functional quantile regression (FQR) is a useful alternative to mean regression for functional data as it provides a comprehensive understanding of how scalar predictors influence the conditional distribution of functional responses. In this article, we study the FQR model for densely sampled, high-dimensional functional data without relying on parametric error or independent stochastic process assumptions, with the focus on statistical inference under this challenging regime along with scalable implementation. This is achieved by a simple but powerful distributed strategy, in which we first perform separate quantile regression to compute $M$-estimators at each sampling location, and then carry out estimation and inference for the entire coefficient functions by properly exploiting the uncertainty quantification and dependence structure of $M$-estimators. We derive a uniform Bahadur representation and a strong Gaussian approximation result for the $M$-estimators on the discrete sampling grid, leading to dimension reduction and serving as the basis for inference. An interpolation-based estimator with minimax optimality is proposed, and large sample properties for point and simultaneous interval estimators are established. The obtained minimax optimal rate under the FQR model shows an interesting phase transition phenomenon that has been previously observed in functional mean regression. The proposed methods are illustrated via simulations and an application to a mass spectrometry proteomics dataset. 
\end{abstract}

\noindent%
{\it Keywords:}  Functional response regression; Functional data analysis; Quantile regression; Bahadur representation; Semiparametric model.
\vfill

\section{Introduction} \label{intro}
Function-on-scalar regression, which refers to the regression of functional responses on a set of scalar predictors, has been extensively studied in the functional data analysis literature; see \citet{ramsay2005functional, ramsay2007applied, morris2015functional} and \citet{wang2016functional} for a thorough review. High-dimensional functional data that are densely sampled at the same locations across subjects commonly arise in many fields. Examples include high-throughput genomics and epigenomics data (e.g., mutation status, copy number, methylation) over chromosomal locations, and neuroimaging data such as functional magnetic resonance imaging and electroencephalography data where brain activity is measured over time for multiple subjects. A function-on-scalar regression model can be formulated as
\begin{align} \label{eq:1}
Y_i(t) = X_i' \: \bm{\beta}(t) + \eta_i(t), \quad  i= 1, \dots, n,
\end{align}
where $t$ is the functional index, $Y_i(t)$ is a functional response on a compact support $\mathcal{T} \subset \mathbb{R}^1$, $\eta_i(t)$ is a residual process on $\mathcal{T}$, $X_i$ is a $d \times 1$ covariate vector in $\mathcal{X} \subset \mathbb{R}^{d}$, and $\bm{\beta}(t)=(\beta_1(t), \dots, \beta_d(t))'$ is a $d \times 1$ vector of regression coefficient functions that relate the covariates $X_i$ with the response $Y_i(t)$ at location $t$. We suppose a sample of $n$ curves $\bm{Y}(t)=(Y_1(t),\dots, Y_n(t))'$ are observed on a common grid $\bm{t}=(t_1, \dots, t_T)'$ in $\mathcal{T}$, where the number of observations $T$ per curve is allowed to grow with $n$, and $\bm{X}$ is an $n \times d$ design matrix.

Existing work on model \eqref{eq:1} has focused predominantly on functional mean regression, where $\eta_i(t)$ is assumed to be a zero-mean stochastic process, and the conditional mean of $Y_i(t)$ can be modeled as $X_i'\:\bm{\beta}(t)$ for each $t$. Quantile regression \citep{koenker1978regression} that studies the effect of covariates on a given quantile level $\tau \in (0,1)$ of a response variable can provide a much more comprehensive understanding of how covariates influence different aspects of the conditional distributions of the response, and has been widely used in various practical applications. In this paper, we study the function-on-scalar quantile regression model, which involves quantile regression of functional responses on scalar predictors that we henceforth refer to as \textit{functional quantile regression} (FQR). For a given quantile level $\tau$, we assume that the stochastic process $\eta_i(t)$ in model \eqref{eq:1} has a zero $\tau$th quantile for each $t$, so the conditional $\tau$th quantile of $Y_i(t)$ is equal to $X_i'\:\bm{\beta}_{\tau}(t)$. The within-function dependence structure is determined by the functional residual process $\eta_i(t)$. Primary interest is estimation of coefficient functions $\bm{\beta}_{\tau}(t)$ that characterize the effect of covariates $X$ on the $\tau$th quantile of the functional response $Y(t)$ at location $t$ as well as performing statistical inference via asymptotic simultaneous confidence bands, while accounting for the within-function dependence structure. 

There is comparatively little work on quantile regression for functional data. A related but distinct line of research is scalar-on-function quantile regression, which is concerned with the conditional quantile of a scalar response given functional covariates~\citep{Cardot+:05, Kato:12,chen+muller2012, li2021inference}; in contrast, the FQR considered in this work refers to function-on-scalar quantile regression, in which the stochastic within-function dependence of functional responses, when coupled with the semiparametric quantile regression problem, presents unique challenges. Among previous studies on FQR, \cite{wang2009quantile} introduced a partially linear varying coefficient model for quantile regression on sparse irregular longitudinal data, but the number of measurements $T$ per subject does not diverge with the number of subjects $n$, and simultaneous band construction was not addressed. Under a densely sampling design where $T$ grows with $n$, \cite{liu2019function} proposed a Bayesian approach and used the posterior samples for estimation and inference. However, they assumed an asymmetric Laplace working likelihood for $\eta(t)$ at each $t$, and for model tractability they did not model the within-function correlations across $t$. More recently, \cite{zhang2021high} studied FQR with a focus on estimation of coefficient functions rather than inference. Our estimation and inference strategy, along with model assumptions, are drastically different both in principle and implementation; however, some overlap might be expected as both works are on FQR. 

In the present paper, we focus on both estimation and inference of $\bm{\beta}_{\tau}(t)$ in model \eqref{eq:1} where $T$ grows with $n$, and do not rely on parametric or independent assumptions on the residual process $\eta(t)$, endowing the model in \cite{liu2019function} with substantially increased flexibility. 

We make the following contributions in this work. From the methodological perspective, we propose a novel distributed strategy in which we first perform pointwise quantile regression separately at each sampling location $t_l \:(l=1, \dots, T)$ to obtain the $ M$-estimator $\hat{\bm{\beta}}_{\tau}(t_l)$ that minimizes the check loss function at each $t_l$, then utilize these $M$-estimators and their uncertainty quantifications to carry out estimation and inference for the entire coefficient functions.  As one concrete example based on this general distributed strategy, we introduce an interpolation-based approach where we interpolate $\hat{\bm{\beta}}_{\tau}(t_l)$ between $t_l$'s to estimate $\bm{\beta}_{\tau}(t)$ for any $t \in \mathcal{T}$.

From the theoretical perspective, we present a uniform Bahadur representation for $\hat{\bm{\beta}}_{\tau}(t_l)$ across the sampling grid $\bm{t}$, where we allow the sampling frequency $T$ to grow exponentially fast with the sample size $n$ by appealing to Vapnik--Chervonenkis (VC) theory.  Based on this uniform Bahadur representation, we next derive a strong Gaussian approximation result for the asymptotic joint distribution of $\hat{\bm{\beta}}_{\tau}(\bm{t})$, which builds a theoretical foundation for our proposed distributed strategy. Importantly, we do not make any parametric assumptions on the residual process $\eta(t)$ for this Gaussian approximation result to hold. Instead, we merely require a mild condition on its zero-crossing behavior in addition to several standard assumptions in the quantile regression literature; see Assumption (A6) in Section \ref{gaussian} for more details. We provide rigorous justification for the interpolation-based approach, showing its rate optimality for estimation and constructing asymptotically valid simultaneous confidence bands under a dense design. 

From the computational perspective, the use of this distributed strategy makes our proposed approach easy to implement and computationally scalable to high dimensional settings ($T \gg n$), which can be further accelerated by utilizing parallel computing, while capable of accounting for intrafunctional correlations in the functional responses. We demonstrate the scalability of our proposed approach by applying it to a mass spectrometry proteomics dataset sampled on $T=3279$ spectral locations, much larger than $T$ that can be handled by \cite{liu2019function}. 

The rest of this paper is organized as follows. We present the uniform Bahadur representation in Section \ref{bahadur}, and study its asymptotic behavior in Section \ref{gaussian}. In Section \ref{interpolation}, we introduce the interpolation-based estimator for the entire coefficient functions $\bm{\beta}_{\tau}(t)$ and develop asymptotic theory for this estimator. We assess the finite sample performance of this interpolation-based approach via a simulation study in Section \ref{simulations}, show an application to a mass spectrometry proteomics dataset in Section \ref{data}, and conclude the paper in Section \ref{discussion}. The supplementary materials include additional technical results and all the proofs. 

\textbf{Model and Notation.} For a given quantile $\tau \in (0,1)$, let $Q(x; t, \tau) = x'\bm{\beta}_{\tau}(t) $ denote the $\tau$th quantile of the functional response $Y(t)$ conditional on the covariates $X=x$ at location $t \in \mathcal{T}$. Let $\left\{(X_i, \; Y_i(t_l)_{l=1}^{T})\right\}_{i=1}^{n}$ be the i.i.d. samples in $\mathcal{X} \times \mathbb{R}^{T}$. Denote the empirical measure of $\left(X_i, \: Y_i(t_l)_{l=1}^{T}\right)$ by $\mathbb{P}_n$ with the corresponding expectation $\mathbb{E}_n$, and the true underlying measure by $P$ with the corresponding expectation $\mathbb{E}$. Denote by $\lVert\mathbf{b}\rVert$ the $L^2$-norm of a vector $\mathbf{b}$. For a square matrix $A$, $\lambda_{\min}(A)$ and $\lambda_{\max}(A)$ are respectively its smallest and largest eigenvalues, and $\lVert A \rVert$ is its operator norm. Let $\mathcal{S}^{m-1} \coloneqq \left\{\mathbf{u}\in \mathbb{R}^m: \: \lVert\mathbf{u}\rVert = 1\right\}$. Define $\rho_{\tau}(u) \coloneqq \left(\tau - \mathbf{1}(u \leq 0)\right) u$, where $\mathbf{1}(\cdot)$ is the indicator function, and $\psi(Y,X;\bm{\beta},\tau) \coloneqq X\left(\mathbf{1}\{Y \leq X'\bm{\beta}\} - \tau\right)$. For a given location $t$, we denote the $M$-estimator of $\bm{\beta}_{\tau}(t)$ by 
\begin{align}\label{eq:2}
\hat{\bm{\beta}}_{\tau, \: n}(t) \coloneqq \argmin_{\bm{\beta}\in \mathbb{R}^d} \sum_{i=1}^{n} \rho_{\tau}(Y_i(t) - X_i'\bm{\beta}). 
\end{align}

\section{Uniform Bahadur Representation} \label{bahadur}
As our first main result, we derive a uniform Bahadur representation for the $M$-estimator $\hat{\bm{\beta}}_{\tau,\: n}(t)$ defined in equation \eqref{eq:2} on the discrete sampling grid $\bm{t}=(t_1, \dots, t_T)'$. For notational simplicity, we suppress the subscript $n$ in $\hat{\bm{\beta}}_{\tau, \: n}(t)$, with the understanding that we consider an estimator based on $n$ curves. 

Throughout the paper, we assume that the following assumptions hold. 

\begin{itemize}
	\item [(A1)] There exist constants $\xi > 0$ and $M > 0$ such that $\lVert X\rVert \leq \xi$ almost surely, and $1/M \leq \lambda_{\min}(\mathbb{E}[XX']) \leq \lambda_{\max}(\mathbb{E}[XX']) \leq M$.
	\item [(A2)] The conditional distribution $F_{Y(t)|X}(y|x)$ is twice differentiable with respect to $y$ for each $t$ and $x$. Denote the derivatives by $f_{Y(t)|X}(y|x)=\!\frac{\partial}{\partial y} F_{Y(t)|X}(y|x)$ and $f'_{Y(t)|X}(y|x)=\!\frac{\partial}{\partial y} f_{Y(t)|X}(y|x)$. Assume that $\overline{f} \coloneqq \sup_{y,x,t \in \mathcal{T}}|f_{Y(t)|X}(y|x)| < \infty$ and $\overline{f'} \coloneqq \sup_{y,x,t \in \mathcal{T}}|f'_{Y(t)|X}(y|x)| < \infty$.
	\item [(A3)] There exists $f_{\min} > 0$ such that $\inf_{t \in \mathcal{T}} \inf_{x} f_{Y(t)|X}(Q(x; \; t, \tau)|x) \geq f_{\min}$.
\end{itemize}

\textit{Remark 1.} Assumption (A1) is a mild condition on the covariate. At any given $t$, Assumptions (A2) and (A3) are standard assumptions on the conditional density $f_{Y(t)|X}(y|x)$ in the quantile regression literature. In our context of FQR, we additionally require that these conditions hold uniformly in $t \in \mathcal{T}$. Letting $J_{\tau}(t) \coloneqq \mathbb{E}[XX'f_{Y(t)|X}(Q(X; \; t, \tau)|X)] = \mathbb{E}[XX'f_{Y(t)|X}(X'\bm{\beta}_{\tau}(t)|X)]$ for each $t$, Assumptions (A1) and (A3) imply that the smallest eigenvalues of $J_{\tau}(t)$ are bounded away from zero uniformly in $t$. 

\begin{theorem}[Uniform Bahadur Representation]
\label{thm1}
Suppose Assumptions (A1)-(A3) hold and additionally $\log T \:\! \log n = o(n^{1/3})$. Then for $t \in \bm{t}$, 
\begin{align}\label{eq:6}
\hat{\bm{\beta}}_{\tau}(t)-\bm{\beta}_{\tau}(t) = -\frac{1}{n} J_{\tau}(t)^{-1} \sum_{i=1}^{n}\psi(Y_i(t), X_i; \bm{\beta}_{\tau}(t), \tau) + r_{n}(t,\tau), 
\end{align}
where $\sup_{t \in \bm{t}}\lVert r_{n}(t,\tau) \rVert = o_p(n^{-1/2})$. \end{theorem}

\textit{Remark 2.} There is a rich literature \citep{he1996general, arcones1996bahadur} on Bahadur representation for $M$-estimators under a scalar response setting. Theorem \ref{thm1} can effectively be viewed as a nontrivial extension of these results to the functional response setting, where we bound the remainder term $r_{n}(t,\tau)$ at $o_p(n^{-1/2})$ \textit{uniformly} over $t$. 

\textit{Remark 3.} The condition $\log T \:\! \log n = o(n^{1/3})$ is very mild as it essentially allows $T$ to grow exponentially fast with $n$. We achieve this flexibility using some arguments based on VC theory. Note that Theorem \ref{thm1} may not be obtained from a straightforward modification of uniform Bahadur representation results over quantile levels $\tau$ available in the literature, such as Theorem 5.1 in \citet{chao2017quantile} and Theorem 2 in \citet{belloni2019conditional}, whose proof techniques were utilized by our proof of Theorem \ref{thm1}. This is primarily because the VC index of the class of functions defined by (S.3.1) over $\bm{t}$ in the functional response setting depends on $\mathbf{card}(\bm{t}) = T$ while its analogue defined over $\tau$ is $O(1)$.  See Lemma 3 and its proof for details.

\section{Strong Gaussian Approximation to M-estimators}\label{gaussian}
The uniform Bahadur representation provided in Theorem \ref{thm1} enables us to study the asymptotic joint distribution of any given linear combination of $\hat{\bm{\beta}}_{\tau,\: n}(t)$ on the discrete sampling grid $\bm{t}$, i.e., $\bm{a}'\hat{\bm{\beta}}_{\tau,\: n}(t)$ where $\bm{a} \in \mathcal{S}^{d-1}$. For notational simplicity, we denote $\bm{a}'\hat{\bm{\beta}}_{\tau,\: n}(t)$ by $\hat{\mu}_n(t)$ and $\bm{a}'\bm{\beta}_{\tau}(t)$ by $\mu(t)$ throughout the rest of the paper, with the understanding that we consider a given quantile level $\tau$ and a given linear combination $\bm{a}$.  

As the second main result, we present a strong Gaussian approximation to the $M$-estimators $\hat{\mu}_n(\bm{t})$. The following additional assumptions are needed.

\begin{itemize}
	\item [(A4)] The coefficient function $\bm{\beta}_{\tau}(t)$ is differentiable with respect to $t$, and $\sup_{t \in \mathcal{T}} \lVert \frac{d}{dt} \bm{\beta}_{\tau}(t) \rVert < \infty$.
	\item [(A5)] The conditional density $f_{Y(t)|X}(y|x)$ is differentiable with respect to $t$ for each $y$ and $x$, and $\sup_{y,x,t \in \mathcal{T}}\left|\frac{\partial}{\partial t} f_{Y(t)|X}(y|x) \right| < \infty$.
	\item [(A6)] Conditional on $\forall\: X \in \mathcal{X}, \eta(t) \:|\: X $ has almost surely continuous sample paths in $\mathcal{T}$, and for $\forall \: t < s \in \mathcal{T}, X \in \mathcal{X}$, there exists a constant $c_0$ independent of $t, s$ and $X$ such that $P\left(\eta(v) \:|\: X \; \text{shows change of sign at least once in} \; v \in [t,s] \right) \leq c_0 |t-s|$.
\end{itemize}

\textit{Remark 4.} Assumption (A4) is about the differentiability of the coefficient function $\bm{\beta}_{\tau}(t)$ with respect to $t$ and uniform boundedness of its first derivative. Assumption (A5) requires that for each $y \in \mathbb{R}$ and $x \in \mathcal{X}$, the conditional density $f_{Y(t)|X}(y|x)$ is differentiable in $t \in \mathcal{T}$, and this derivative is uniformly bounded over $x, y$ and $t$. Assumption (A6) regularizes the residual process $\eta(t)$ using its zero-crossing behavior, which does not require specifying the distribution of the stochastic process $\eta(t)$. Assumption (A6) holds if $\eta(t)$ is a Gaussian process which possesses almost surely continuous sample paths in $\mathcal{T}$ and certain additional properties; see Lemma 5 in the supplement and discussion therein for more details. 

\begin{theorem}[Gaussian Coupling for \textit{M}-estimator]
\label{thm2} For a given linear combination $\bm{a} \in \mathcal{S}^{d-1}$ and any $t \in \mathcal{T}$, let
\begin{align}\label{eq:10}
\mathbb{G}_n(t) \coloneqq \frac{1}{\sqrt{n}} \bm{a}' J_{\tau}(t)^{-1} \sum_{i=1}^{n} X_i\left(\mathbf{1}\{Y_i(t) \leq X_i'\bm{\beta}_{\tau}(t)\} - \tau\right).
\end{align}
Under Assumptions (A1)-(A6), if we additionally assume that $\log T \:\! \log n = o(n^{1/3})$ and $\delta_{T} \coloneqq \max_{1 \leq l \leq T-1} \left|t_{l+1} - t_l \right| = o(1)$, then we have
\begin{align}\label{eq:16}
\sqrt{n} \left(\hat{\mu}_n(\bm{t}) - \mu(\bm{t}) \right) =  \tilde{G}_n(\bm{t}) + \tilde{r}_n(\bm{t}),
\end{align}
where $\tilde{G}_n(\cdot)$ is a process on $\mathcal{T}$ that, conditional on $\left(X_i\right)_{i=1}^{n}$, is zero-mean Gaussian with almost surely continuous sample paths and the covariance function
\begin{align}\label{eq:15}
\mathbb{E}\left[\tilde{G}_n(t) \tilde{G}_n(s) \mid (X_i)_{i=1}^n \right] = \mathbb{E}\left[\vphantom{\tilde{G}_n(t)} \mathbb{G}_n(t) \mathbb{G}_n(s) \mid (X_i)_{i=1}^n \right], \quad \forall \; t,\: s \in \mathcal{T}, 
\end{align} 
and the sup norm of the residual term $\tilde{r}_n(\bm{t})$ is bounded by $o_p(1)$.
\end{theorem}

Theorem \ref{thm2} shows that $\sqrt{n} \left(\hat{\mu}_n(\bm{t}) - \mu(\bm{t}) \right)$ can be strongly approximated by zero-mean Gaussian, which has important theoretical and practical implications.  More specifically, Theorem \ref{thm2} implies that, rather than directly working with the $n \times T$ matrix $\bm{Y}$ consisting of observed functional responses, we can instead work with the $T \times 1$ vector $\hat{\mu}_n(\bm{t})$ that is asymptotically zero-mean Gaussian with the $T \times T$ covariance matrix 
\begin{align}\label{eq:17}
\Sigma_{\bm{t}} \coloneqq  \text{Cov}\left[\tilde{G}_n(\bm{t}) \mid (X_i)_{i=1}^n \right], 
\end{align}
after centering by $\mu(\bm{t})$ and rescaling by $\sqrt{n}$. This data reduction is computationally appealing especially for large sample size $n$. In addition, Theorem \ref{thm2} effectively transforms the originally complicated FQR problem, which is semiparametric in nature, into a much more manageable Gaussian mean regression problem with a particular covariance structure for the residual errors, for which many modeling approaches are available in the literature, such as the commonly used kernel or spline smoothing and some nonparametric Bayesian methods.

We next estimate the entire coefficient function $\mu(t)$ for $t \in \mathcal{T}$ based on the $M$-estimator $\hat{\mu}_n(t)$ for $t \in \bm{t}$, which can be achieved using various approaches. In particular, we propose an approach based on interpolation in Section \ref{interpolation}, which is shown to be rate optimal for estimating $\mu(t)$.

\section{Asymptotic Properties of Interpolation-based Estimator} \label{interpolation}

\subsection{Linear Interpolation-based Estimator} \label{linear}
We first consider a linear interpolation-based estimator $\hat{\mu}_n(t)^{LI}$, which is defined as
\begin{align}\label{eq:7}
\hat{\mu}_n(t)^{LI} \coloneqq \frac{t_{l+1}- t}{t_{l+1}-t_l} \hat{\mu}_n(t_l) +  \frac{t- t_l}{t_{l+1}-t_l} \hat{\mu}_n(t_{l+1}), \quad \forall \; t \in [t_l, t_{l+1}], \;\; l=1, 2, \dots, T-1. 
\end{align}
If $t \in \bm{t}$, then it is apparent that $\hat{\mu}_n(t)^{LI} = \hat{\mu}_n(t)$. Let $\tilde{\mu}$ be the linear interpolation of $\{ \mu(t_l): \; 1 \leq l \leq T\}$, that is, 
\begin{align}\label{eq:8} 
\tilde{\mu}(t) \coloneqq \frac{t_{l+1}- t}{t_{l+1}-t_l} \mu(t_l) + \frac{t - t_l}{t_{l+1}-t_l} \mu(t_{l+1}),  \quad \forall \; t \in [t_l, t_{l+1}], \;\; l=1, 2, \dots, T-1. 
\end{align} 

The following theorem shows that the process $\sqrt{n}\left( \hat{\mu}_n(\cdot)^{LI} - \tilde{\mu}(\cdot) \right)$ converges weakly to a centered Gaussian process in $l^{\infty}(\mathcal{T})$.

\begin{theorem}[Weak Convergence]
\label{thm3}
Under the conditions assumed for Theorem \ref{thm2},
\begin{align}\label{eq:9}
\hat{\mu}_n(t)^{LI} - \tilde{\mu}(t) = -\frac{1}{\sqrt{n}} \mathbb{G}_n(t) + o_p(\frac{1}{\sqrt{n}}), 
\end{align}
where the remainder term $o_p(n^{-1/2})$ is uniform in $t \in \mathcal{T}$. In addition,
\begin{align}\label{eq:11}
\sqrt{n} \left( \hat{\mu}_n(\cdot)^{LI} - \tilde{\mu}(\cdot) \right) \rightsquigarrow \mathbb{G}(\cdot) \;\; \text{in} \;\; l^{\infty}(\mathcal{T}),
\end{align}
where $\mathbb{G}(\cdot)$ is a centered Gaussian process on $\mathcal{T}$ with the covariance function $H_{\tau}$ given by
\begin{align}\label{eq:12}
H_{\tau}(t, s; \: \bm{a}) \coloneqq \bm{a}'J_{\tau}(t)^{-1} \mathbb{E}\left[\psi(Y(t), X; \: \bm{\beta}_{\tau}(t), \tau) \cdot \vphantom{\left(\lambda\right)^2} \psi(Y(s), X;\: \bm{\beta}_{\tau}(s), \tau)'\right] J_{\tau}(s)^{-1} \bm{a},  
\end{align}
for any $t, s \in \mathcal{T}$. In particular, there exists a version of $\mathbb{G}$ with almost surely continuous sample paths.
\end{theorem}

\textit{Remark 5.} The major challenge in proving Theorem \ref{thm3} is to show the asymptotic tightness of the process $\mathbb{G}_n(t)$ in $l^{\infty}(\mathcal{T})$ (see Section 1.5 in \citet{van1996weak} and the proof to Lemma 8 in the supplement for more details),  i.e.,  for any $c > 0$,
\begin{align}\label{eq:13}
\lim_{\delta \downarrow 0} \limsup_{n \rightarrow \infty} P\left(\sup_{t,s \in \mathcal{T}, \; |t-s| \leq \delta} \left| \mathbb{G}_n(t) - \mathbb{G}_n(s) \right| > c \right) = 0. 
\end{align}
The two mild assumptions (A5) and (A6) are the only conditions made on the residual process $\eta(t)$ to conclude the asymptotic tightness of $\mathbb{G}_n(t)$. Note that unlike \cite{liu2019function}, we do not need to make any distributional assumptions on $\eta(t)$ or its dependence structure across $t$. 

We then present a strong approximation to the process $\sqrt{n}\left( \hat{\mu}_n(\cdot)^{LI} - \tilde{\mu}(\cdot) \right)$ by a sequence of Gaussian processes, by extending Theorem \ref{thm2} to the continuum $\mathcal{T}$. 

\begin{theorem}[Gaussian Coupling for Linear Interpolation-based Estimator]
\label{thm4}
Under the conditions assumed for Theorem \ref{thm2}, 
\begin{align}\label{eq:14}
\sqrt{n} \left( \hat{\mu}_n(t)^{LI} - \tilde{\mu}(t) \right) = \tilde{G}_n(t) + \tilde{r}_n(t), \quad t \in \mathcal{T}, 
\end{align}
where $\sup_{t \in \mathcal{T}} \left|\tilde{r}_n(t)\right| = o_p(1)$. If we additionally assume that $\delta_{T} = o(n^{-1/2})$, then we have
\begin{align}\label{eq:19}
\sqrt{n} \left( \hat{\mu}_n(t)^{LI} - \mu(t) \right) = \tilde{G}_n(t) + \tilde{r}_n(t), \quad t \in \mathcal{T}, 
\end{align}
where $\sup_{t \in \mathcal{T}} \left|\tilde{r}_n(t)\right| = o_p(1)$.
\end{theorem}

The Gaussian coupling result \eqref{eq:19} in Theorem \ref{thm4} allows us to construct a $1-\alpha$ simultaneous confidence band for the functional parameter $\mu$. More specifically, let 
\begin{align} \label{eq:22}
\sigma_n(t) \coloneqq  \left(\mathbb{E}\left[\tilde{G}^2_n(t) \:|\: (X_i)_{i=1}^n \right] \right)^{1/2} = \left( \tau (1-\tau) \bm{a}'J_{\tau}(t)^{-1} \:\! \mathbb{E}_n \left[X_i X_i' \right]  \:\!  J_{\tau}(t)^{-1} \bm{a} \right)^{1/2}, 
\end{align}
and $C_n(\alpha)$ is defined such that
\begin{align}\label{eq:20}
P\left( \sup_{t \in \mathcal{T}} \left|\sigma^{-1}_n(t) \tilde{G}_n(t) \right| \leq C_n(\alpha) \right) = 1 - \alpha,
\end{align}
then a $1-\alpha$ simultaneous confidence band for $\mu(t)$ is given by
\begin{align}\label{eq:18}
\left(\hat{\mu}_n(t)^{LI} - \frac{1}{\sqrt{n}} C_n(\alpha) \sigma_n(t), \;\;  \hat{\mu}_n(t)^{LI} + \frac{1}{\sqrt{n}} C_n(\alpha) \sigma_n(t) \right). 
\end{align}

In practice, both $\sigma_n(t)$ and $C_n(\alpha)$ are unknown and need to be estimated.  The expression of $\sigma_n(t)$ involves $J_{\tau}(t)$, which we estimate using the Powell sandwich method \citep{powell1991estimation} for $t \in \bm{t}$, followed by linear interpolation of Powell's estimator for $t \in \mathcal{T}$, as described in Section S.1. We adopt the weighted bootstrap method proposed in \citet{belloni2019conditional} to estimate $C_n(\alpha)$, which is also elaborated in Section S.1. Substituting the estimator $\hat{\sigma}_n(t)$ and $C^{b}_n(\alpha)$ into \eqref{eq:18} yields $1-\alpha$ simultaneous confidence band for the functional parameter $\mu$. This is given by Theorem \ref{thm6}, which is the third main result.

\begin{theorem}[Simultaneous Confidence Band]
\label{thm6}
Suppose Assumptions (A1)-(A6) hold. If we additionally assume that $\log T \:\! \log n = o(n^{1/3})$ and $\delta_{T} = o(n^{-1/2})$, then a $ 1-\alpha$ simultaneous confidence band for $\mu(t)$ is given by
\begin{align}\label{eq:33}
\left(\hat{\mu}_n(t)^{LI} - \frac{1}{\sqrt{n}} C^{b}_n(\alpha) \hat{\sigma}_n(t), \;\;  \hat{\mu}_n(t)^{LI} + \frac{1}{\sqrt{n}} C^{b}_n(\alpha) \hat{\sigma}_n(t) \right). 
\end{align}
\end{theorem}

Note that Theorem \ref{thm6} requires that the functional data are sampled on a sufficiently dense grid such that $\delta_{T} = o(n^{-1/2})$, which is equivalent to $T \gg n^{1/2}$ if the sampling locations $\bm{t}$ are equally spaced. This additional assumption bounds the bias associated with the linear interpolation-based estimator $\hat{\mu}_n(t)^{LI}$ at $o(n^{-1/2})$, eliminating the need to estimate the bias term and simplifying the construction of the simultaneous confidence band. To perform functional inference on $\mu(t)$ while adjusting for multiple testing over $t$, we can invert simultaneous confidence bands to construct simultaneous band scores (\textit{SimBaS}) $P_{\mu}(t)$ for each $t \in \mathcal{T}$, which is defined as the minimum $\alpha$ such that the $1-\alpha$ simultaneous confidence band of $\mu(t)$ excludes 0 at $t$ \citep{meyer2015bayesian}. $P_{\mu}(t)$ can be interpreted as the multiplicity-adjusted $p$-value for testing $\mu(t)=0$ at a given $t$ that adjusts across all $t \in \mathcal{T}$ based on the experimentwise error rate.

\subsection{Minimax Rate and Spline Interpolation-based Estimator} \label{splines}
We next consider an estimator for $\mu(t)$ based on spline interpolation which generalizes the linear interpolation-based estimator introduced in Section \ref{linear}, and show that this estimator is rate optimal for estimating $\mu(t)$.  For a general order $r \geq 1$,  the estimator based on $r$-th order spline interpolation is denoted by $\hat{\mu}_n(t)^{r-SI}$ and defined as the solution to
\begin{align}\label{eq:21}
\min_{g \in \mathcal{W}_2^{r}} \int_{\mathcal{T}}\left[g^{(r)}(t)\right]^2\:dt, \quad\quad \text{subject to} \;\; g(t_l) = \hat{\mu}_n(t_l), \quad l=1, \dots, T, 
\end{align} 
where $\mathcal{W}_2^{r}$ denotes the $r$-th order Sobolev-Hilbert space on $\mathcal{T}$, that is,
\begin{align}
\mathcal{W}_2^{r} \coloneqq \{g: \mathcal{T} \rightarrow \mathbf{R} \mid g, g^{(1)}, \dots, g^{(r-1)} \: \text{are absolutely continuous and} \: g^{(r)} \in \mathcal{L}_2(\mathcal{T})\}. 
\end{align}  
When $r=1$, the solution to equation \eqref{eq:21} is exactly $\hat{\mu}_n(t)^{LI}$ defined in equation \eqref{eq:7}. 

Theorem \ref{thm8} derives minimax lower bounds for estimating the coefficient function $\mu \in \mathcal{W}_2^{r}$.

\begin{theorem}[Minimax lower bound]
\label{thm8}
For any estimate $\check{\mu}$ of the coefficient function $\mu$ based on the observed data $\left\{(X_i, \; Y_i(t_l)_{l=1}^{T})\right\}_{i=1}^{n}$, 
\begin{align}\label{eq:36}
\lim_{d \rightarrow 0} \limsup\limits_{n \rightarrow \infty} \sup\limits_{\mu \in \mathcal{W}_2^{r}} P\left(\lVert \check{\mu} - \mu \rVert_{\mathcal{L}_2}^2 > d \left(T^{-2r} + n^{-1}\right) \right) = 1. 
\end{align}
\end{theorem}

Utilizing Theorem \ref{thm2} and some classical results about spline interpolation \citep{devore1993constructive}, we can calculate the rate of convergence for $\hat{\mu}_n(t)^{r-SI}$.

\begin{theorem}[Minimax upper bound] 
\label{thm5}
Suppose Assumptions (A1)-(A6) hold. If we additionally assume that $\log T \:\! \log n = o(n^{1/3})$, $\delta_{T} \leq C_0 \:\! T^{-1}$ for some constant $C_0 > 0$, then 
\begin{align}\label{eq:23}
\lim_{D \rightarrow \infty} \limsup\limits_{n \rightarrow \infty} \sup\limits_{\mu \in \mathcal{W}_2^{r}} P\left(\lVert \hat{\mu}_n^{r-SI} - \mu \rVert_{\mathcal{L}_2}^2 > D \left(T^{-2r} + n^{-1}\right) \right) = 0. 
\end{align}
\end{theorem}

Theorem \ref{thm5} in conjunction with Theorem \ref{thm8} shows that $\hat{\mu}_n^{r-SI}$ achieves the optimal rate for estimating the coefficient function $\mu$ over $\mathcal{W}_2^{r}$, which is our fourth main result. Notably, this convergence rate is identical to the optimal rate established by \citet{cai2011optimal} for estimating the mean function based on discretely sampled functional data under the common sampling design. Theorem \ref{thm5} is reminiscent of an interesting phase transition phenomenon observed by \citet{cai2011optimal} in their function-on-scalar mean regression setting. In particular, a phase transition in the convergence rate of $\hat{\mu}_n^{r-SI}$ occurs when $T$ is of the order $n^{1/2r}$. When the functional data are observed on a relatively dense grid ($T \gg n^{1/2r}$), the sampling frequency $T$ does not have an effect on the rate of convergence, which is of the order $1/n$ and only determined by the sample size $n$. For a sparser grid $T = O(n^{1/2r})$, the rate of convergence is of the order $T^{-2r}$ and only determined by the sampling frequency $T$.

\section{Simulation Studies} \label{simulations}
In this section, we present results from simulation studies to assess the finite sample performance of the linear interpolation-based approach defined by equation \eqref{eq:7}.

\textbf{Simulation design.} We simulate functional data to mimic the motivating mass spectrometry data described in Section \ref{data}, where we approximate the shapes of mass spectrometry peaks with Gaussian densities~\citep{zhang2009review}. Data are generated according to the following model:
\begin{align}\label{eq:24}
\begin{split}
y_i(t) = \; & \sum\limits_{k=1}^{7}c_{i,k} \varphi\left(t\mid\mu_{k},\sigma_{k}\right) +  \epsilon_i(t), \\
c_{i,k} = \; & \mathbf{1}\{x_{i1}=-1\} f_{1,k} + \mathbf{1}\{x_{i1}=1\} f_{2,k} + x_{i2},  
\end{split}
\end{align}
In \eqref{eq:24}, $\varphi\left(t\mid\mu_{k},\sigma_{k}\right)$ denotes the probability density function of a normal distribution with mean $\mu_{k}$ and standard deviation $\sigma_{k}$, which corresponds to the peak $k$ centered at $\mu_k$. $c_{i,k}$ represents the height of the peak $k$ in subject $i$, which is determined by the covariate vector $x_i = (1, x_{i1}, x_{i2})'$, where $x_{i1}$ is a binary variable taking values from $\{-1,1\}$ with equal probability, and $x_{i2}$ is a standard normal variable independent of $x_{i1}$. For each peak $k$, the values of $\mu_{k}$ and $\sigma_{k}$ are shown in Table \ref{table:1} along with the distributions of $f_{1,k}$ and $f_{2,k}$. The i.i.d. noise term $\epsilon_i(t)$ is an $\text{AR}(1)$ process with lag $1$ autocorrelation $\rho=0.5$ and a marginal standard normal distribution.  Under model \eqref{eq:24}, the $\tau$th quantile of $Y(t)$ conditional on $x_1$ and $x_2$ is $\beta_0^{\tau}(t) + \beta_1^{\tau}(t) x_1 + \beta_2^{\tau}(t) x_2$ for any $\tau \in (0,1)$, where $\beta_1^{\tau}(t)$ quantifies the difference in the $\tau${th} quantile of $Y(t)$ between the two groups indexed by $x_1$ while conditioning on $x_2$, and $\beta_2^{\tau}(t)$ quantifies the change in the $\tau${th} quantile of $Y(t)$ if the continuous covariate $x_2$ increases by one unit while conditioning on $x_1$.  We considered two different sample sizes $n = 400, 8000$, and four different sampling frequencies $T = 64, 128, 256, 512$ which are equally spaced on the interval $\mathcal{T} = [0, 15.33]$. For each combination of $(n, T)$, we simulated $100$ replicate datasets. The true functional coefficients $\beta_1^{\tau}(t)$ and $\beta_2^{\tau}(t)$, which are constant across quantiles, are shown in Figure~\ref{fig:figure1}. 

\begin{table}[h!]
\caption{Parameter specifications of the data generating models in simulations.}
\label{table:1}
\setlength{\tabcolsep}{36pt}
\begin{tabular*}{\linewidth}{ @{} *{5}c @{}}
\toprule
basis index $k$ & $\mu_{k}$ & $\sigma_{k}$ & $f_{1,k}$ & $f_{2,k}$ \\[3pt]
\midrule
1 & 1 & 0.15 & $N(0,1)$ & $N(2,1)$ \\[2pt]
2 & 3 & 0.15 & $N(0,1)$ & $N(0,1)$ \\[2pt]
3 & 5 & 0.18 & $N(2,1)$ & $N(0,1)$ \\[2pt]
4 & 7 & 0.18 & $N(0,1)$ & $N(0,1)$ \\[2pt]
5 & 9 & 0.18 & $N(0,1)$ & $N(0,1)$ \\[2pt]
6 & 11 & 0.2 & $N(0,1)$ & $N(2,1)$ \\[2pt]
7 & 13 & 0.2 & $N(0,1)$ & $N(0,1)$ \\
\bottomrule
\end{tabular*}
\end{table}

\begin{figure}[h]
    \centering
    \includegraphics[width=\textwidth]{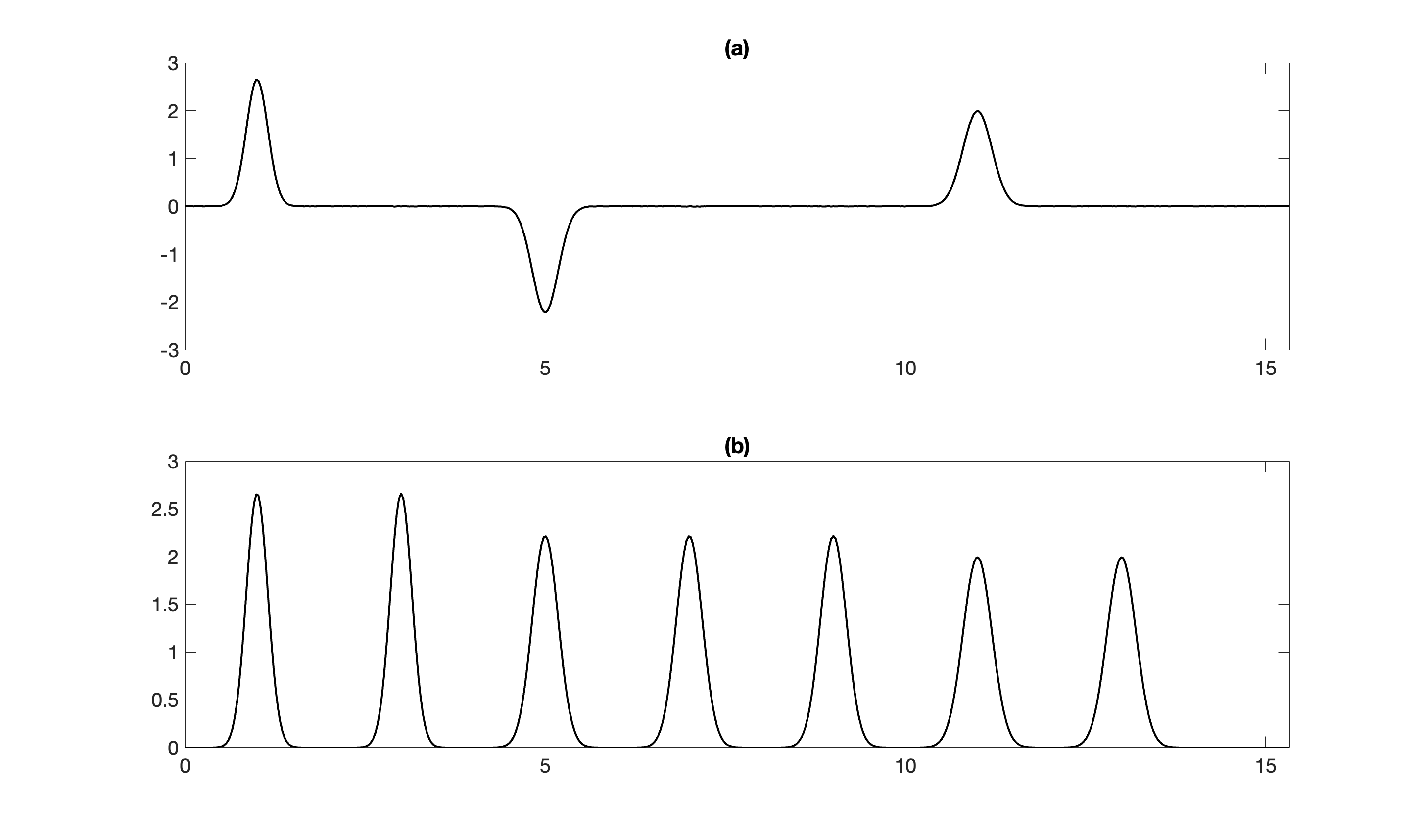}
    \caption{True functional coefficients $\beta_1^{\tau}(t)$ and $\beta_2^{\tau}(t)$, which are constant across quantiles, are shown respectively in (a) and (b).}  
    \label{fig:figure1}
\end{figure}

\textbf{Simulation results.} We apply the linear interpolation-based approach to the simulated datasets to perform FQR at $\tau = 0.5, 0.8, 0.9$. We evaluate estimation performance using the integrated mean squared error (IMSE) defined by $\int_{\mathcal{T}} \left(\hat{\mu}(t) - \mu(t) \right)^2 \text{d}t$ where $\hat{\mu}(t)$ is an estimate for $\mu(t)$, and inferential performance using the coverage probability of the 95\% simultaneous confidence band covering the true functional parameter. Simulation results are summarized in Figure~\ref{fig:figure2}. 

\begin{figure}[h]
    \centering
    \includegraphics[width=0.82\textwidth]{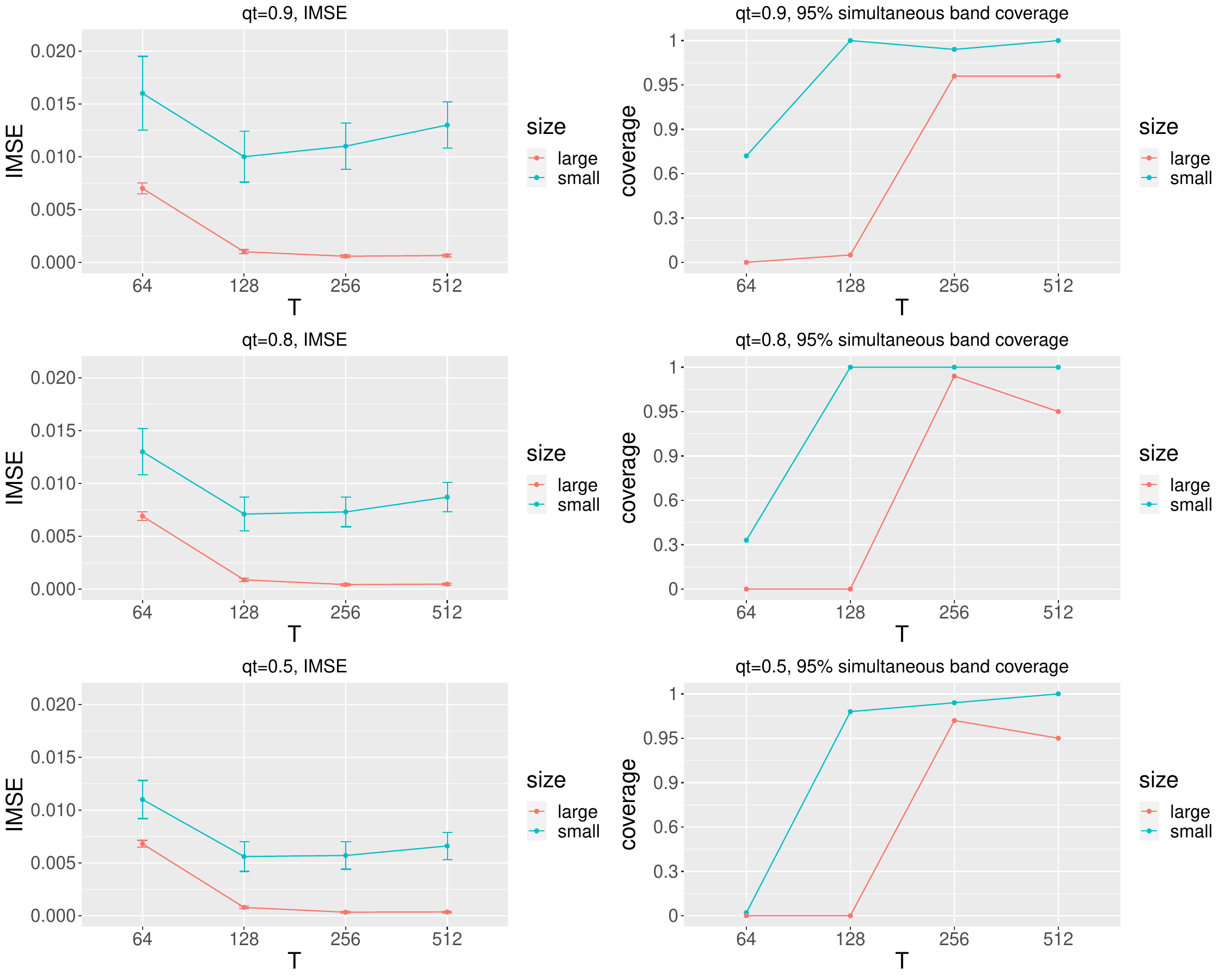}
    \includegraphics[width=0.82\textwidth]{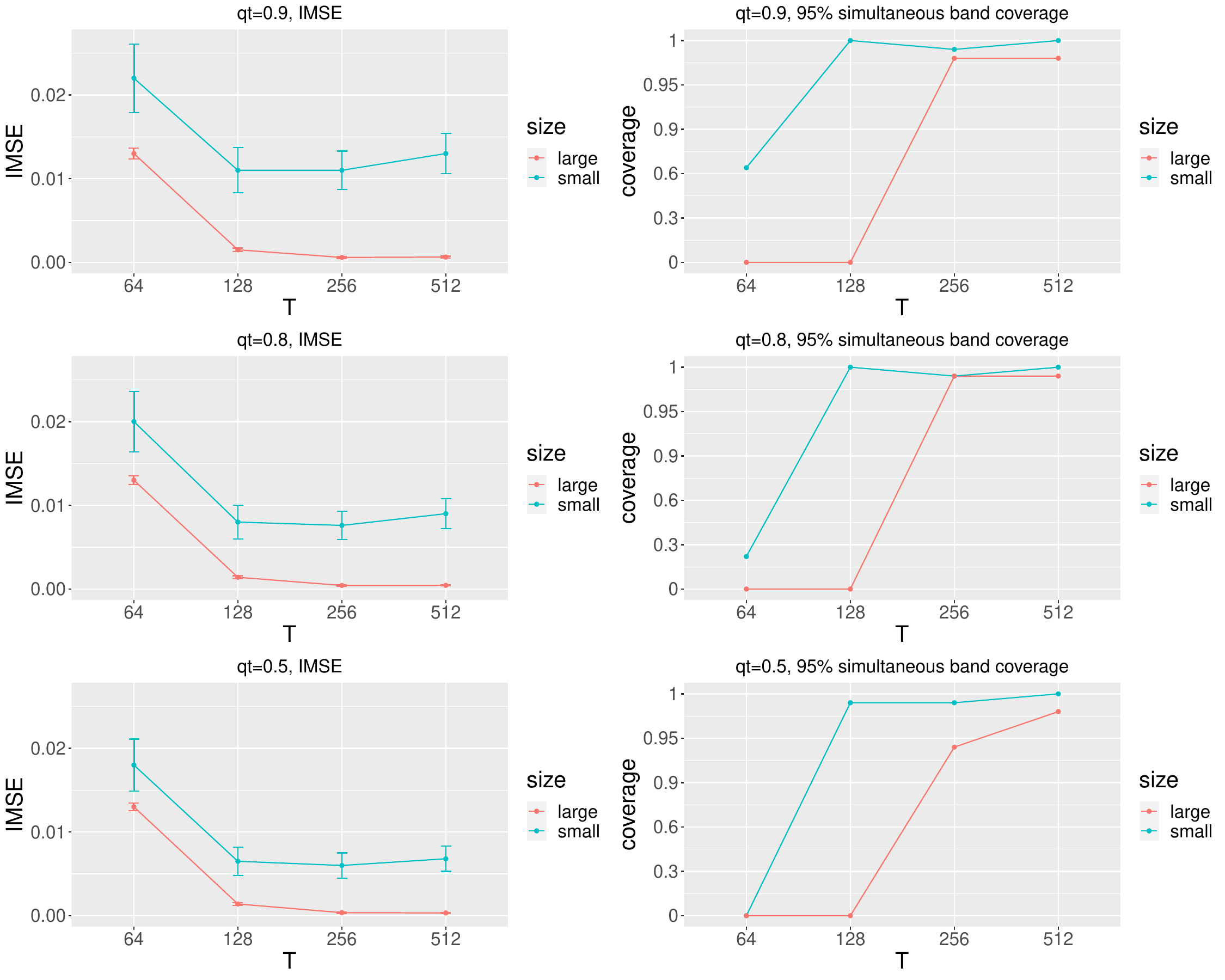}
    \vspace*{-3mm}
    \caption{Simulation results of the linear interpolation-based approach are shown for $\beta_1^{\tau}(t)$ in the top three rows, and $\beta_2^{\tau}(t)$ in the bottom three rows. Each summary measure is averaged across 100 replicate datasets. The error bars for IMSE denote the standard errors over the 100 replicates.}  
    \label{fig:figure2}
\end{figure}

In terms of estimation performance, for both $\beta_1^{\tau}(t)$ and $\beta_2^{\tau}(t)$, we can clearly observe a ``turning point'' in the estimation accuracy curve -- the IMSE drops drastically as $T$ increases before this turning point and then reaches a plateau or shows very mild variations afterwards. This applies to each quantile level $\tau$ and sample size $n$. For example, at $n = 8000$, the IMSE for estimating $\beta_1^{0.9}(t)$ drops from 0.007 at $T=64$ to 0.001 at $T=128$, then stabilizes as $T$ further increases. This observation agrees well with our theoretical results in Section \ref{splines} that suggest a phase transition in the rate of convergence for functional coefficient estimation, which is of the order $T^{-2}$ when $T = O(n^{1/2})$ and is of the order $1/n$ and unaffected by $T$ when $T \gg n^{1/2}$ in the case of linear interpolation. As expected, for each functional coefficient, the estimation accuracy improves as $n$ increases from 400 to 8000 for the same quantile level and $T$, and also improves as the quantile level gets closer to the median for the same $T$ and $n$. 

In terms of inferential performance, the coverage probability of the 95\% simultaneous confidence band also exhibits a turning point as a function of $T$, in the sense that the coverage probability sharply rises with increasing $T$ before this point and then stabilizes around or above the nominal level beyond this point. This applies to each functional parameter, quantile level and sample size. For example, at $n = 8000$, the 95\% simultaneous band coverage for $\beta_1^{0.9}(t)$ rises from 0.05 at $T=128$ to 0.96 at $T=256$ and remains stable afterwards. This observation is consistent with Theorem \ref{thm6}, which requires $T \gg n^{1/2}$ for the simultaneous confidence band to be asymptotically valid. Interestingly, when $T$ is large enough (e.g., $T=512$), the coverage probability is always closer to the 95\% nominal level for $n=8000$ than $n=400$, for each functional parameter and quantile level. This suggests that the simultaneous confidence band produced by the weighted bootstrap method tends to be conservative when the sample size is relatively small, but generally achieves nominal coverage for sufficiently large sample sizes, which empirically validates our theoretical results.   

Simulations were run on a $64$-bit operating system with 2 processors and an RAM of $24$GB. The running time averaged across 100 replicates and 3 quantile levels is presented in Table \ref{table:2} for each combination of $(n, T)$, suggesting that the running time is linear in $T$. For the Bayesian approach proposed by \cite{liu2019function}, it takes an average of 75 minutes to perform FQR on a dataset with $n=400$ and $T=300$, and the running time is quadratic in $T$. This leads to an estimated running time of more than 36 hours on a dataset with $n=8000$ and $T=512$ for the Bayesian approach; for this reason, we did not implement it as a comparator in the simulation study. For our proposed distributed strategy, we remark that parallel computing can be used across locations $t$ and across bootstrap iterations for the weighted bootstrap procedure for further acceleration, so it is scalable to much larger datasets than considered here. 

\begin{table}[h!]
\caption{Average computing time (in minutes) over 100 replicate datasets and 3 quantile levels for each combination of $(n, T)$.}
\label{table:2}
\setlength{\tabcolsep}{32pt}
\begin{tabular*}{\linewidth}{ @{} *{5}c @{}}
\toprule
& $T=64$ & $T=128$ & $T=256$ & $T=512$ \\[3pt]
\midrule
$n=400$ & 3 & 6 & 12 & 24 \\[2pt]
$n=8000$ & 30 & 60 & 120 & 240 \\
\bottomrule
\end{tabular*}
\end{table}

\section{Real Data Application} \label{data}
Mass spectrometry is a common analytical technique to simultaneously measure the expressions of many proteins in a biological sample, and produces a mass spectrum which is a highly spiky function quantifying the relative abundance of proteins at the mass-to-charge ratio $t$ in the given sample by the spectral intensity $y(t)$. We analyzed a mass spectrometry dataset \citep{koomen2005plasma} generated based on blood serum samples from 139 pancreatic cancer patients and 117 normal controls to identify proteins which have differential abundance between cancer and normal samples and might serve as potential proteomic biomarkers of pancreatic cancer.  

Nearly all existing statistical methods developed for differential expression detection in mass spectrometry data perform mean regression to compare mean protein expression levels across groups. However, given that cancer is characterized by inter-patient heterogeneity, proteomic biomarkers might be aberrantly expressed in only a small proportion of the cancer cohort compared to the normal cohort. In such cases, cancer-normal differences may be difficult to detect when comparing the means, and might be more easily detected by quantile-based methods that can look in the tails of the distributions. 

To obtain a more complete picture of protein expression differences between the cancer and normal cohorts in this dataset, we performed FQR at $\tau=0.5, 0.8, 0.9$ to identify spectral regions that significantly differ between the two cohorts at each quantile level. Since spectral intensities can span orders of magnitude across the entire range of mass-to-charge ratio, we took $\log_2$ transformation on the functional responses so that an absolute difference of one on the transformed scale corresponds to a two-fold change on the original scale. Each mass spectrum $y_i(t)$ is observed on the same set of $T=3,279$ spectral locations between $8,000$ and $16,000$ Daltons, and is regressed upon the covariate vector $x_i=(1, x_{i1})'$ where $x_{i1}$ denotes cancer $(=1)$ or normal $(=-1)$ status. We also performed functional mean regression by fitting the wavelet-based functional mixed model (WFMM, \cite{morris2006wavelet}) to this dataset and compared to FQR model fitting results, to illustrate the potential for FQR to provide new biological insights in mass spectrometry data analysis. More specifically, we fit the model $y_i(t) = x_i' \: \bm{\beta}^{\text{mean}}(t) + \epsilon_i(t)$, where $\epsilon_i(t)$ is assumed to be a zero-mean Gaussian process. The functional parameters of interest are $\beta_1^{\tau}(t)$ and $\beta_1^{\text{mean}}(t)$, which respectively model the $\log_2$ fold change of protein expression levels in the $\tau$th quantile and the mean between the cancer and normal cohorts at the spectral location $t$, and are referred to as the cancer main effect functions. 

Although this mass spectrometry dataset had been studied in \citet{liu2019function}, in the current work we focused on analyzing a non-overlapping and much broader spectral region ([$8000$D, $16000$D] versus [$5000$D, $8000$D]) where the number of observations $T$ per curve is roughly twice as many as that in \citet{liu2019function} ($3,279$ versus $1,659$). Thanks to the use of the distributed strategy, under the same computer setting as in Section \ref{simulations}, it took just 2 hours for the linear interpolation-based approach to perform FQR for each quantile level, which is computationally much more efficient than the Bayesian approach that would need at least 18 hours to perform the same task.

Figure \ref{fig:figure3} displays the estimate of $\beta_1^{\tau}(t)$ for $\tau=0.5, 0.8, 0.9$ and $\beta_1^{\text{mean}}(t)$, along with 95\% pointwise and simultaneous bands. Notably, the estimates of $\beta_1^{0.8}(t)$ and $\beta_1^{0.9}(t)$ are clearly greater in magnitude than those of $\beta_1^{0.5}(t)$ and $\beta_1^{\text{mean}}(t)$ in the spectral region [$11500$D, $12500$D]. The proteins corresponding to such locations are differentially expressed between the cancer and normal cohorts most prominently in the upper quantiles, indicating that they might be over-expressed in only a subset of cancer patients and reflect key aspects of the underlying biology of these patients' tumors. These proteins might serve as potential biomarkers of pancreatic cancer and warrant further investigation. In particular, the number of locations where the estimated cancer main effect function exceeds 0.5 (corresponding to 2-fold change between the two cohorts) in this region is respectively 302, 183, 1 and 79 for $\beta_1^{0.9}(t), \beta_1^{0.8}(t), \beta_1^{0.5}(t), \beta_1^{\text{mean}}(t)$; interestingly, this number is much larger for $\beta_1^{\text{mean}}(t)$ than $\beta_1^{0.5}(t)$, suggesting that functional median regression provides more robust estimation than functional mean regression when we are interested in comparing the average protein expression levels between the two cohorts.   

\begin{figure}[h]
    \centering
    \includegraphics[width=\textwidth]{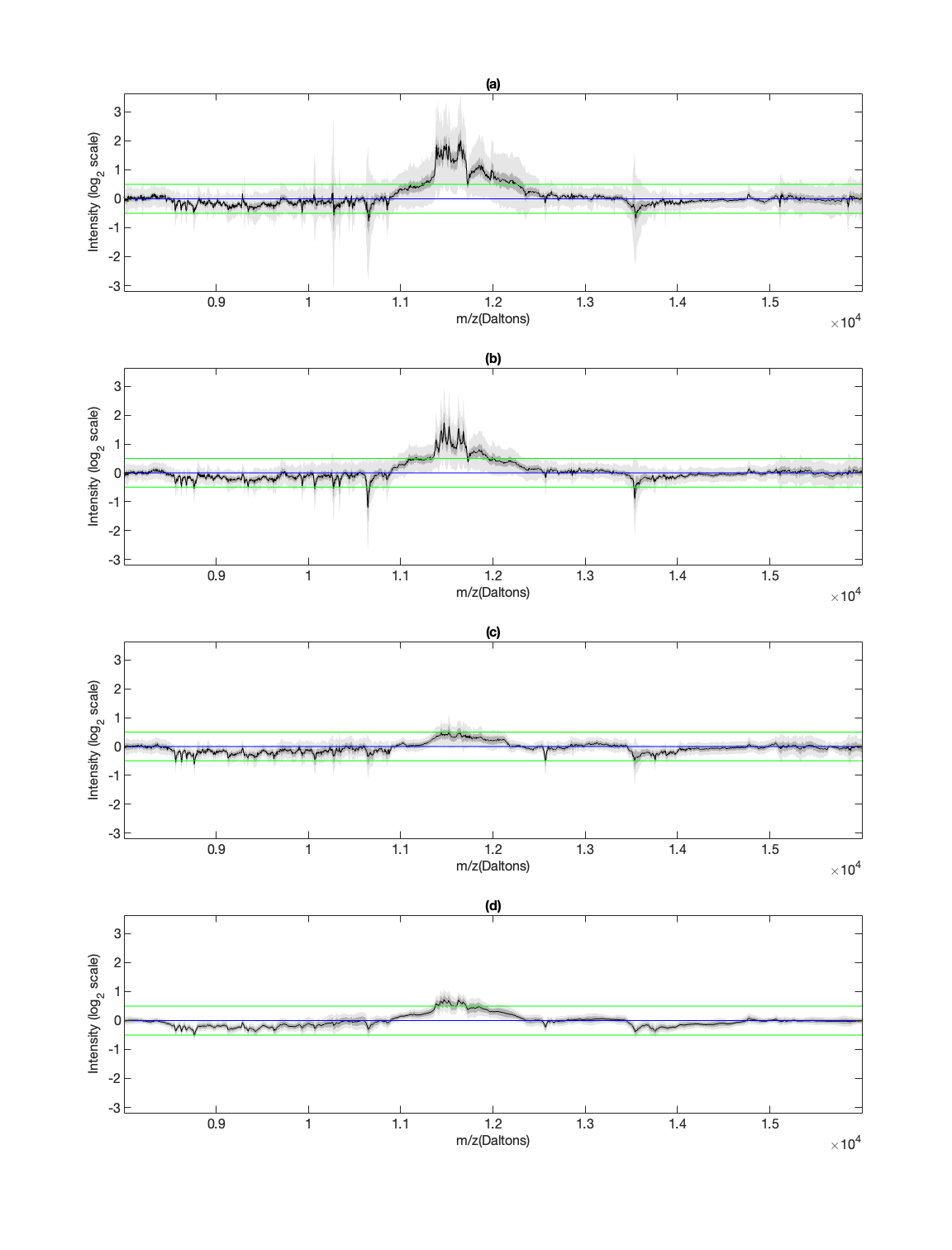}
    \vspace*{-25mm}
    \caption{Estimated cancer main effect functions $\beta_1^{\tau}(t)$ at $\tau = 0.9$ in (a),  $\tau = 0.8$ in (b) and $\tau = 0.5$ in (c) by the linear interpolation-based approach, and $\beta_1^{\text{mean}}(t)$ in (d) by the WFMM approach. In each subfigure, the functional coefficient estimate is shown on $\log_2$ scale, along with the 95\% pointwise (simultaneous) confidence band shown in dark (light) gray. The green horizontal lines correspond to $2$-fold change in the protein expression levels between the cancer and normal groups.}
    \label{fig:figure3}
\end{figure}

\section{Discussion} \label{discussion}
We have introduced a scalable distributed strategy to perform function-on-scalar quantile regression, where we first run separate quantile regression at each sampling location to compute the $M$-estimators, then make use of these $M$-estimators and their uncertainty estimates to do estimation and inference for the entire coefficient functions. As one concrete example, we propose to estimate the entire coefficient functions by linearly interpolating the $M$-estimators, which is shown to be rate optimal, and construct asymptotically valid simultaneous confidence bands. 

Theorem \ref{thm2} presents a strong Gaussian approximation result for the asymptotic joint distribution of $M$-estimators on the discrete sampling grid, which builds a theoretical foundation upon which various approaches to modeling the coefficient function other than the interpolation-based strategy can be developed. While the interpolation-based estimator possesses minimax optimality, in practice, more accurate estimation could be achieved by alternative modeling approaches. Also, our numerical studies reveal that the simultaneous confidence bands constructed using the weighted bootstrap method achieve nominal coverage when the sample size is sufficiently large, but tend to be overly conservative for smaller sample sizes, suggesting that the finite sample performance of simultaneous inference could be further improved. 

We briefly discuss a promising alternative strategy that we are currently exploring. Based on equation \eqref{eq:16}, $\sqrt{n} \left(\hat{\mu}_n(\bm{t}) - \mu(\bm{t}) \right) \! \mid \! \mu(\bm{t}) \; \sim \;  \text{MVN}(\bm{0}, \Sigma_{\bm{t}})$
up to some negligible error that is uniformly bounded by $o_p(1)$,  where $\Sigma_{\bm{t}}$ is the $T \times T$ covariance matrix defined in equation \eqref{eq:17}. This asymptotically valid Gaussian likelihood for $\hat{\mu}_n(\bm{t})$ conditioning on $\mu$ motivates us to adopt a Bayesian framework to model the coefficient function $\mu$ by placing a prior on it.  A Bayesian framework naturally incorporates our prior knowledge about characteristics of the functional responses through the choice of an appropriate prior for $\mu$, and yields uncertainty quantification based on posterior distributions in addition to point estimates. For example, for relatively smooth and regular $\mu$, we can use a Gaussian process prior with a squared exponential kernel; for spiky and spatially heterogeneous $\mu$, we can represent $\mu$ with wavelet basis functions and place a shrinkage prior on the wavelet basis coefficients. This Bayesian approach can properly take into account the covariance structure of the $M$-estimators when smoothing coefficient functions, in the sense that the heteroscedastic uncertainty across $t$ in the functional response is learned in the pointwise quantile regression step and then used to inform the functional coefficient regularization. Our preliminary numerical studies suggest that the Bayesian approach results in more adaptive smoothing and improved small sample properties than the interpolation approach. A theoretical investigation of the Bayesian approach would be beneficial but also challenging, which we reserve for future work. Future research topics also include explorations of alternative modeling approaches to FQR.

\section*{Acknowledgements}
This research was supported by the grants R01-CA178744 and R01-CA244845 from the National Cancer Institute, 1550088 from the National Science Foundation, 1R24MH117529 of the BRAIN Initiative of the United States National Institutes of Health, and an ORAU Ralph E. Powe Junior Faculty Enhancement Award.

\clearpage

\bibliographystyle{apalike}
\bibliography{references}

\end{document}

\maketitle

\setlength\abovedisplayskip{5pt}
\setlength\belowdisplayskip{-5pt}
\setlength\abovedisplayshortskip{3pt}
\setlength\belowdisplayshortskip{-5pt}

The supplement is organized as follows. In Section \ref{scb}, we describe how to estimate the quantities $\sigma_n(t)$ and $C_n(\alpha)$ which are needed to construct the simultaneous confidence band for the functional parameter $\mu$ using Theorem 5. In Section \ref{indep}, we present additional simulation results. In Section \ref{lemmas}, we state the lemmas that are needed to prove the theorems in the paper. We give proofs to the theorems in Section \ref{theorems} and proofs to the lemmas and propositions in Section \ref{lemma_proof}.

\section{Construction of Simultaneous Confidence Band for $\mu(t)$} \label{scb}
\subsection{Estimation of $\sigma_n(t)$}
The expression of $\sigma_n(t)$, as defined in (4.9),  involves $J_{\tau}(t)$ which we estimate using the Powell sandwich method \citep{powell1991estimation} for $t \in \bm{t}$, followed by linear interpolation of Powell's estimator for $t \in \mathcal{T}$.
\begin{align} \label{eq:27}
\hat{J}_{\tau}(t_l) \coloneqq  \frac{1}{2 h_n} \mathbb{E}_n \left[\mathbf{1}\{ | Y_i(t_l) \leq X_i'\hat{\bm{\beta}}_{\tau}(t_l) | \leq h_n \} X_i X_i' \right],  \quad  l = 1, 2, \dots, T, \\[-5pt]
\end{align}
where $h_n$ is some bandwidth parameter, and 
\begin{align} \label{eq:29}
\hat{J}_{\tau}(t) \coloneqq & \; \frac{t_{l+1}- t}{t_{l+1}-t_l} \hat{J}_{\tau}(t_l) +  \frac{t- t_l}{t_{l+1}-t_l} \hat{J}_{\tau}(t_{l+1}), \quad \forall \; t \in [t_l, t_{l+1}], \;\; l=1, 2, \dots, T-1. \\[-5pt]
\end{align}
Under the conditions $h_n = o(1)$ and $\log T \log n = o(n h_n)$, we show in Lemma \ref{lemma10} that $\hat{J}_{\tau}(t)$ is a consistent estimator of $J_{\tau}(t)$ uniformly over $t \in \mathcal{T}$.  Substituting $\hat{J}_{\tau}(t)$ into $J_{\tau}(t)$,  we estimate $\sigma_n(t)$ by $\hat{\sigma}_n(t) \coloneqq  \left( \tau (1-\tau) \bm{a}' \hat{J}_{\tau}(t)^{-1} \:\! \mathbb{E}_n \left[X_i X_i' \right]  \:\!  \hat{J}_{\tau}(t)^{-1} \bm{a} \right)^{1/2}$.

\subsection{Estimation of $C_n(\alpha)$}
To estimate $C_n(\alpha)$ defined in (4.10) for $\alpha \in (0, 1)$, we adopt the weighted bootstrap method proposed in \citet{belloni2019conditional}. In particular, we first draw i.i.d positive weights $\omega_1, \omega_2, \dots, \omega_n$ independently of the data, and solve the quantile regression problem based on the weighted data $\left\{(\omega_i X_i, \; \omega_i Y_i(t_l)_{l=1}^{T})\right\}_{i=1}^{n}$ for $t \in \bm{t}$.
\begin{align} \label{eq:30}
\hat{\bm{\beta}}^{b}_{\tau}(t) \coloneqq \argmin_{\bm{\beta}\in \mathbb{R}^d} \sum_{i=1}^{n} \rho_{\tau}( \omega_i Y_i(t) - \omega_i X_i'\bm{\beta}) = \argmin_{\bm{\beta}\in \mathbb{R}^d} \sum_{i=1}^{n} \omega_i \rho_{\tau}(Y_i(t) - X_i'\bm{\beta}). \\[-5pt]
\end{align}
Then we define the interpolation-based estimator $\hat{\mu}_n(t)^{b-LI}$ on $t \in \mathcal{T}$ for a given linear combination $\bm{a} \in \mathcal{S}^{d-1}$, \\[-5pt]
\begin{align} \label{eq:31}
\!\!\!\!\!\!\!\! \hat{\mu}_n(t)^{b-LI} \coloneqq \frac{t_{l+1}- t}{t_{l+1}-t_l} \bm{a}' \hat{\bm{\beta}}^{b}_{\tau}(t_l) +  \frac{t- t_l}{t_{l+1}-t_l} \bm{a}' \hat{\bm{\beta}}^{b}_{\tau}(t_{l+1}), \; \forall \; t \in [t_l, t_{l+1}], \; l=1, \dots, T-1. \\[-5pt]
\end{align}

Under certain conditions on the distributions for weights $\left(\omega_i \right)_{i=1}^n$, we show in Proposition \ref{prop1} that the process $\sqrt{n} \left( \hat{\mu}_n(\cdot)^{b-LI} - \hat{\mu}_n(\cdot)^{LI} \right)$ is coupled with a sequence of centered Gaussian processes $\tilde{G}^{b}_n(\cdot)$ that have the same covariance function as $\tilde{G}_n(\cdot)$ for each $n$, conditional on $\left(X_i\right)_{i=1}^{n}$.  

\begin{proposition}[Gaussian Coupling for Weighted Data]
\label{prop1}
Suppose $\omega_1, \dots, \omega_n$ are i.i.d draws from a distribution for the random variable $\omega$ that satisfies $\omega > 0, \; \mathbb{E} \left[\omega \right] = 1, \; \mathbb{E} \left[\omega^2 \right] = 2, \; \mathbb{E} \left[\omega^4 \right] \lesssim 1, \; \max_{1 \leq i \leq n} \omega_i \lesssim_{P} \log n$.  Under the conditions assumed for Theorem 2, 
\begin{align}\label{eq:32}
\sqrt{n} \left( \hat{\mu}_n(t)^{b-LI} - \hat{\mu}_n(t)^{LI} \right) = \tilde{G}^{b}_n(t) + \tilde{r}^{b}_n(t), \quad t \in \mathcal{T}, \\[-5pt]
\end{align}
where $\tilde{G}^{b}_n(\cdot)$ is a process on $\mathcal{T}$ that, conditional on $\left(X_i\right)_{i=1}^{n}$, is zero-mean Gaussian with almost surely continuous sample paths and the covariance function given by (3.3), and $\sup_{t \in \mathcal{T}} \left|\tilde{r}^{b}_n(t)\right| = o_p(1)$.
\end{proposition}

Given Proposition \ref{prop1}, we can estimate $C_n(\alpha)$ by $C^{b}_n(\alpha)$, which is defined such that
\begin{align}\label{eq:34}
P\left( \sup_{t \in \mathcal{T}} \left| \sqrt{n} \: \frac{ \hat{\mu}_n(t)^{b-LI} - \hat{\mu}_n(t)^{LI}}{\hat{\sigma}_n(t)} \right| \leq C^{b}_n(\alpha) \right) = 1 - \alpha. 
\end{align}

\bigskip

\section{Additional Details about the Simulation Study} \label{indep}
In this section, we present additional simulation results to assess the finite sample performance of the linear interpolation-based approach. 

Functional data were simulated according to model (5.1) in the main paper. The parameter specifications of the data generating model are essentially the same as those presented in the main paper, except that the i.i.d. noise term $\epsilon_i(t)$ is independent across $t$ with a marginal standard normal distribution to mimic the scenario where the measurement errors of functional observations are uncorrelated. Under this scenario, the residual process $\eta_i(t)$ in model (1.1) no longer possesses almost surely continuous sample paths, thus Assumption (A6) that regularizes the residual process using its zero-crossing behavior is violated. 

Similar to the simulation study in the main paper, we also considered two different sample sizes $n = 400, 8000$, and four different sampling frequencies $T = 64, 128, 256, 512$ which are equally spaced on $\mathcal{T}$. For each combination of $(n, T)$, we simulated $100$ replicate datasets and performed FQR at $\tau = 0.5, 0.8, 0.9$. For functional coefficients $\beta_1^{\tau}(t)$ and $\beta_2^{\tau}(t)$, simulation results averaged across 100 replicates are displayed in Figure~\ref{fig:figureS1}, which are qualitatively similar to the simulation results presented in Figure 2 in the main paper, where the Gaussian noise term $\epsilon_i(t)$ is an $\text{AR}(1)$ process with lag $1$ autocorrelation $\rho=0.5$. This suggests that the linear interpolation-based approach can still achieve satisfactory estimation and inferential performance in the presence of uncorrelated measurement errors. 

\begin{figure}[h]
    \centering
    \includegraphics[width=0.85\textwidth]{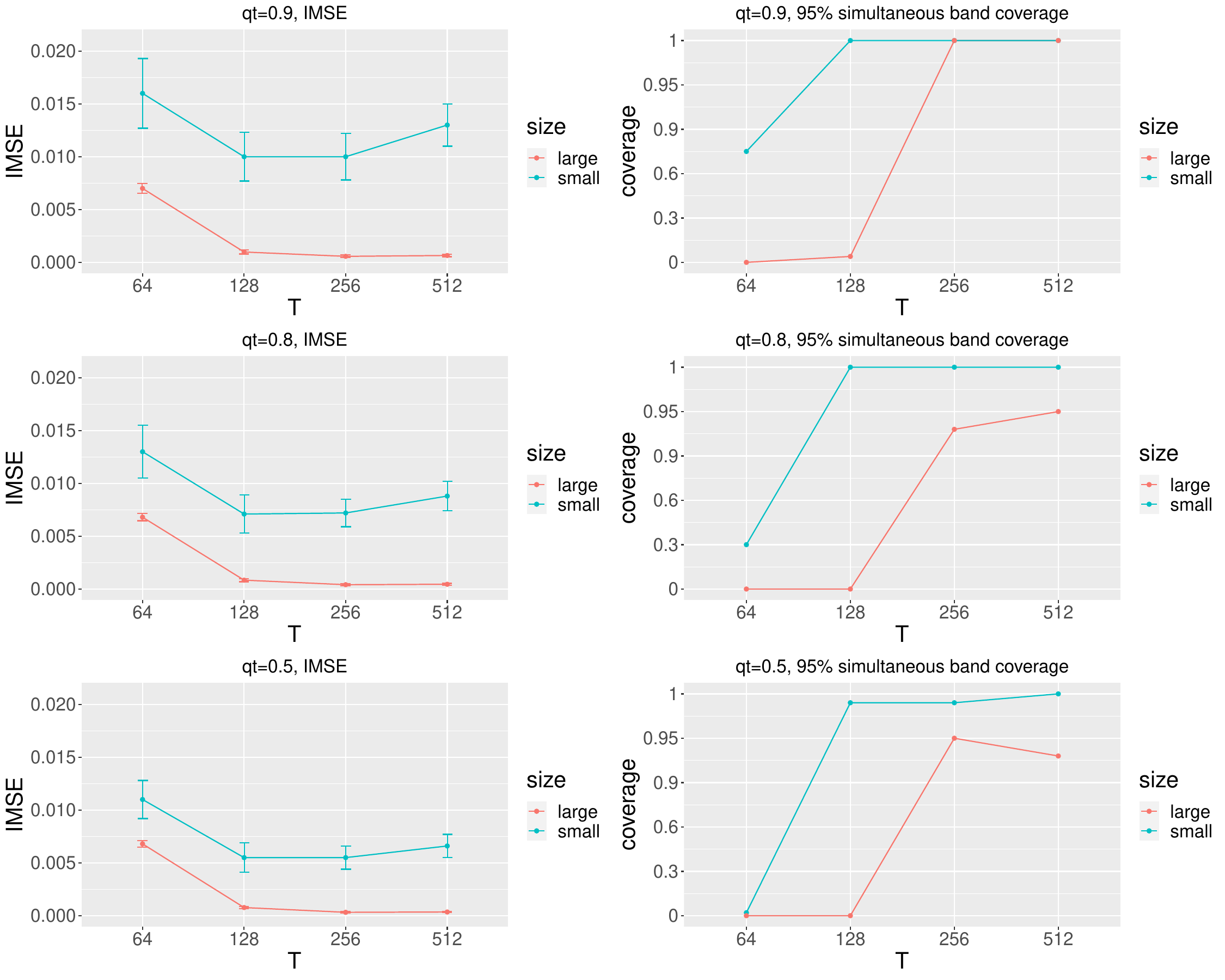}
    \includegraphics[width=0.85\textwidth]{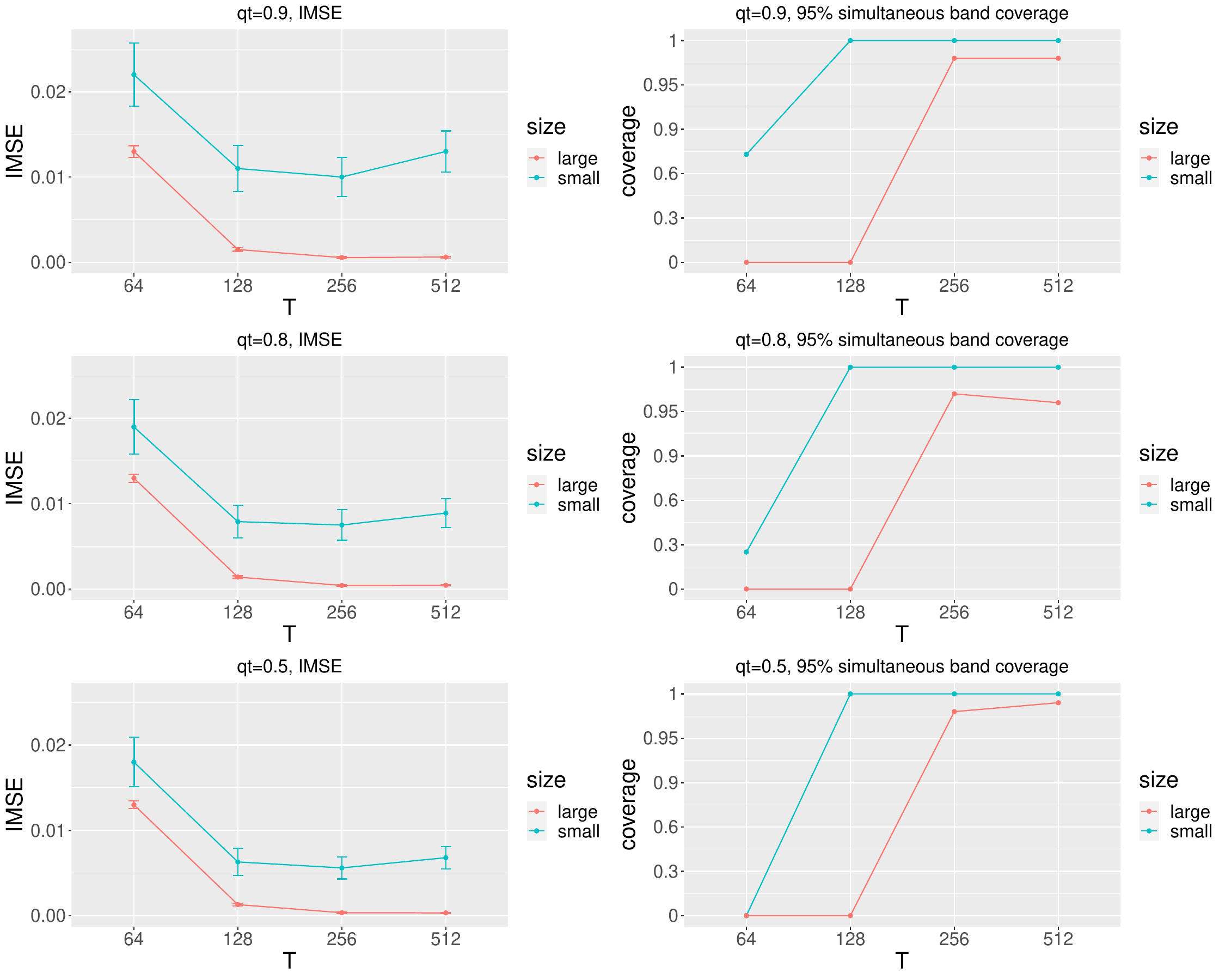}
    \caption{Simulation results of the linear interpolation-based approach are shown for $\beta_1^{\tau}(t)$ in the top three rows, and $\beta_2^{\tau}(t)$ in the bottom three rows. Each summary measure is averaged across 100 replicate datasets. The error bars for IMSE denote the standard errors over the 100 replicates.}  
    \label{fig:figureS1}
\end{figure}

\clearpage

\section{Statement of Lemmas} \label{lemmas}
Before stating the lemmas that will be needed to prove the theorems in the paper, we first define two classes of functions:
\begin{align} \label{eq:51}
\mathcal{G}_1 \coloneqq & \; \left\{ (X, \bm{Y}) \mapsto (\bm{u}'X)(\bm{1}\{Y(t) \leq X'\bm{\beta}\} - \tau) \; | \; \bm{u} \in \mathcal{S}^{d-1}, \; \bm{\beta} \in \mathbb{R}^{d}, \;  1 \leq t \leq T \right\}, \\
\mathcal{G}_2(\delta) \coloneqq & \; \left\{(X, \bm{Y}) \mapsto (\bm{u}'X)(\bm{1}\{Y(t) \leq X'\bm{\beta}_1\} - \bm{1}\{Y(t) \leq X'\bm{\beta}_2\}) \mid  \right. \\
\mathrel{\phantom{=}} & \; \left. \bm{u} \in \mathcal{S}^{d-1}, \; \bm{\beta}_1 \in \mathbb{R}^{d}, \; \bm{\beta}_2 \in \mathbb{R}^{d}, \; \lVert \bm{\beta}_1-\bm{\beta}_2 \rVert \leq \delta,  \;  1 \leq t \leq T \right\},\\[-5pt]
\end{align}
where $\bm{Y}$ denotes the $T \times 1$ vector with the $t^{th}$ element equal to $Y(t)$.

\bigskip

\begin{lemma} 
\label{lemma1}
Under Assumptions (A1)-(A3), for any $\delta > 0$, 
\begin{align}
\sup_{t \in \bm{t}} \sup_{\lVert \bm{\beta}-\bm{\beta}_{\tau}(t) \rVert \leq \delta} \left\lVert \vartheta(\bm{\beta}; t, \tau) - \vartheta(\bm{\beta}_{\tau}(t); t, \tau) - J_{\tau}(t)(\bm{\beta}-\bm{\beta}_{\tau}(t)) \right\rVert \leq \lambda_{\max}(\mathbb{E}[XX'])\overline{f'}\delta^2 \xi, \\[-5pt]
\end{align}
where $\vartheta(\bm{\beta}; t,\tau) \coloneqq \mathbb{E}\left[\psi(Y(t), X; \bm{\beta}, \tau)\right]$.
\end{lemma}

\begin{lemma}
\label{lemma2}
Let $s_{n,1} \coloneqq \lVert \mathbb{P}_n - P \rVert _{\mathcal{G}_1}$. Under Assumptions (A1)-(A3), for any $\upsilon > 1$, 
\begin{align}
\left\{ \sup_{t \in \bm{t}} \lVert \hat{\bm{\beta}}_{\tau}(t) - \bm{\beta}_{\tau}(t) \rVert \leq \frac{2\upsilon s_{n,1}}{\inf_{t \in \bm{t}}\lambda_{\min}(J_{\tau}(t))} \right\} \supseteq \left\{s_{n,1} < \frac{\inf_{t \in \bm{t}} \lambda^2_{\min}(J_{\tau}(t))}{4 \upsilon \xi \overline{f'} \lambda_{\max}(\mathbb{E}[XX'])} \right\}.  \\[-5pt]
\end{align}
\end{lemma}

\begin{lemma}
\label{lemma3}
For any $\epsilon > 0$, 
\begin{align}
N(\epsilon;\; \mathcal{G}_1, L^2(\mathbb{P}_n)) \leq \left(\frac{A \lVert F_1 \rVert_{L^2(\mathbb{P}_n)}}{\epsilon}\right)^{\nu_1(T)}, \;\; N(\epsilon;\; \mathcal{G}_2(\delta), L^2(\mathbb{P}_n)) \leq \left(\frac{A \lVert F_2 \rVert_{L^2(\mathbb{P}_n)}}{\epsilon}\right)^{\nu_2(T)}, \\[-5pt]
\end{align}
where the covering number $N(\epsilon;\; \mathcal{G}, L_p)$ is the minimal number of balls of radius $\epsilon$ (under $L_p$ norm) that is needed to cover $\mathcal{G}$, $A$ is some constant, $F_1$ and $F_2$ are respectively envelope functions of $\mathcal{G}_1$ and $\mathcal{G}_2(\delta)$, and $\nu_1(T)=O(\log T)$, $\nu_2(T)=O(\log T)$. 
\end{lemma}

\begin{lemma}
\label{lemma4} 
Consider the classes of functions $\mathcal{G}_1$ and $\mathcal{G}_2(\delta)$ defined in \eqref{eq:51}. Under Assumptions (A1)-(A3), for some constant $C$ independent of $n$ and all $\kappa_n > 0$,
\begin{align}\label{eq:53}
P\left(\lVert \mathbb{P}_n - P \rVert_{\mathcal{G}_1} \geq C\left[\left(\frac{\log T}{n}\right)^{1/2} \!\!\! + \frac{\log T}{n} + \left(\frac{\kappa_n}{n}\right)^{1/2} + \frac{\kappa_n}{n}\right]\right) \leq e^{-\kappa_n}. \\[-5pt]
\end{align}
For any $\delta_n \downarrow 0 $ satisfying $\delta_n \gg n^{-1}$, \eqref{eq:54} holds for sufficiently large $n$ and arbitrary $\kappa_n > 0$,
\begin{align}\label{eq:54}
P\left(\lVert \mathbb{P}_n - P \rVert_{\mathcal{G}_2(\delta_n)} \geq C\zeta_n(\delta_n, \kappa_n) \right) \leq e^{-\kappa_n}, \\[-5pt]
\end{align}
where \\[-5pt]
\begin{align}\label{eq:55}
\zeta_n(\delta_n, \kappa_n) \coloneqq \delta_n^{1/2} \left(\frac{\log T}{n}\log n\right)^{1/2} \! + \frac{\log T}{n} \log n + \delta_n^{1/2}\left(\frac{\kappa_n}{n}\right)^{1/2} \! + \frac{\kappa_n}{n}. \\[-5pt]
\end{align}
\end{lemma}

\begin{lemma}
\label{lemma5} 
Let $\eta(t)$ be a Gaussian process with almost-surely continuous sample paths in $\mathcal{T}$, and for any $v, s \in \mathcal{T}$, $v < s$, let $\mathcal{C}_{0}(v, s)$ denote the number of zero crossings by $\eta(t)$ in $v \leq t \leq s$. 
\begin{itemize}
	\item [(i)] ((10.3.1) of Cram{\'e}r and Leadbetter (1967)). Assume that $\eta(t)$ is a zero-mean stationary Gaussian process, then 
\begin{align}
\mathbb{E}[\mathcal{C}_{0}(v, s)] = \frac{s-v}{\pi}\sqrt{\frac{\lambda_2}{\lambda_0}}, \\[-5pt]
\end{align} 
where $\lambda_{2k} \:(k \geq 0)$ denotes the $2k$th moment of the spectral function $F$, i.e., 
\begin{align}
\lambda_{2k} = \int_{0}^{\infty} \lambda^{2k} d F(\lambda), \quad\quad k=0,1,2,\dots \\[-5pt]
\end{align}
	\item [(ii)] ((13.2.1) of Cram{\'e}r and Leadbetter (1967)). Assume that $\eta(t)$ is a Gaussian process with the mean function $m(t) \coloneqq \mathbb{E}[\eta(t)]$ and the covariance function $r(t,s) \coloneqq \mathbb{E}[(\eta(t)-m(t))(\eta(s)-m(s))]$, and both $m(t)$ and $r(t,s)$ are continuous. For convenience we shall write $\sigma^2(t)=r(t,t)$ for the variance at location $t$. Additionally, suppose that $m(t)$ has the continuous derivative $m'(t)$ for each $t$ and that $r(t,s)$ has a second mixed partial derivative $r_{11}(t,s)$, which is continuous at all diagonal points $(t,t)$. Suppose also that the joint normal distribution for $\eta(t)$ and the derivative $\eta'(t)$ is nonsingular for each $t$, and $\sigma(t) > 0$. Then
\begin{align}\label{eq:56}
\mathbb{E}[\mathcal{C}_{0}(v, s)] =  \int_{v}^{s} \gamma \sigma^{-1} (1-\mu^2)^{1/2} \phi \left(\frac{m}{\sigma}\right)\{2\phi(\omega) + \omega(2\Phi(\omega)-1)\} dt, \\[-5pt]
\end{align}
where \\[-5pt]
\begin{align} 
\gamma^2(t) \coloneqq & \; \text{Var}[\eta'(t)] = r_{11}(t,t) = \left[\frac{\partial^2 r(t,s)}{\partial t \: \partial s}\right]_{t=s}, \nonumber\\
\mu(t) \coloneqq & \; \frac{\text{Cov}[\eta(t), \eta'(t)]}{\gamma(t)\sigma(t)} = \frac{r_{01}(t,t)}{\gamma(t)\sigma(t)} = \frac{\left[\frac{\partial r(t,s)}{\partial s}\right]_{t=s}}{\gamma(t)\sigma(t)}, \nonumber\\
\omega(t) \coloneqq & \; \frac{m'(t)-\gamma(t)\mu(t)m(t)/\sigma(t)}{\gamma(t)(1-\mu^2(t))^{1/2}}. \nonumber\\[-5pt]
\end{align}
\end{itemize}
\end{lemma}

\textit{Remark.} If $\eta(t)$ is a zero-mean stationary Gaussian process with a finite second spectral moment $\lambda_2$, Lemma \ref{lemma5} (i) combined with Markov's inequality suggests that Assumption (A6) is satisfied with $c_0 = \frac{1}{\pi}\sqrt{\frac{\lambda_2}{\lambda_0}}$. More generally, if $\eta(t)$ is a possibly non-stationary Gaussian process as assumed in Lemma \ref{lemma5} (ii) and the integrand in equation \eqref{eq:56} is bounded above uniformly over $t \in \mathcal{T}$ by a constant $C_0$, an application of Markov's inequality leads to Assumption (A6) with $c_0 = C_0$.

\begin{lemma}
\label{lemma6} 
Under Assumptions (A1)-(A5), there exists a constant $C'$ such that uniformly over $t, s \in \mathcal{T}$, we have
\begin{align}\label{eq:57}
\lVert J_{\tau}(t)^{-1} - J_{\tau}(s)^{-1} \rVert \leq C' \left|t-s \right|. \\[-5pt]
\end{align}
\end{lemma}

\begin{lemma}
\label{lemma7} 
(Lemma A.1 of \citet{kley2016quantile}) Let $(\mathcal{T}, d)$ be an arbitrary metric space, and $D(\epsilon, d)$ be the packing number of this metric space. Assume that $\{\mathbb{G}_t: t \in \mathcal{T}\}$ is a separable stochastic process with $\lVert \mathbb{G}_s - \mathbb{G}_t \rVert_{\Psi} \leq C d(s,t)$ for all $s, t$ satisfying $d(s,t) \geq \bar{\omega}/2 \geq 0$,  where $\lVert Z \rVert_{\Psi} \coloneqq \inf \{C > 0: \mathbb{E}[\Psi(|Z|/C)] \leq 1\}$ is the \textit{Orlicz norm} of a real-valued random variable $Z$ (see Chapter 2.2 in \citet{van1996weak}) for a non-decreasing, convex function $\Psi: \mathbb{R}^{+} \rightarrow \mathbb{R}^{+}$ with $\Psi(0)=0$. Then, for any $\delta > 0, \omega \geq \bar{\omega}$, there exists a random variable $S_1(\omega)$ and a constant $K < \infty$ such that 
\begin{align}
\sup_{d(s,t)\leq \delta} \left| \mathbb{G}_s - \mathbb{G}_t \right| \leq S_1(\omega) \: + 2 \sup_{d(s,t)\leq \bar{\omega}, \: t \in \tilde{T}} \left| \mathbb{G}_s - \mathbb{G}_t \right|, \\[-5pt]
\end{align}
and \\[-5pt]
\begin{align}
\lVert S_1(\omega) \rVert_{\Psi} \leq K \left[\int_{\bar{\omega}/2}^{\omega} \Psi^{-1}\left(D(\epsilon,d)\right) \text{d}\epsilon + (\delta + 2\bar{\omega}) \Psi^{-1}\left(D^2(\omega,d)\right) \right],  \\[-5pt]
\end{align}
where the set $\tilde{T}$ contains at most $D(\bar{\omega},d)$ points. 
\end{lemma}

\begin{lemma}
\label{lemma8}
For $t\in \mathcal{T}$, let $\mathbb{G}_n(t)$ be defined as (3.1) for some given $\bm{a} \in \mathcal{S}^{d-1}$, and let $\mathbb{G}(\cdot)$ be a centered Gaussian process on $\mathcal{T}$ with the covariance function $H_{\tau}$ defined in (4.5). Suppose Assumptions (A1)-(A6) hold, then
\begin{align}\label{eq:59}
\mathbb{G}_n(\cdot) \rightsquigarrow \mathbb{G}(\cdot) \;\; \text{in} \;\; l^{\infty}(\mathcal{T}).  \\[-5pt]
\end{align}
In particular, there exists a version of $\mathbb{G}$ with almost surely continuous sample paths.
\end{lemma}

\begin{lemma}
\label{lemma9}
Under Assumptions (A1)-(A3), for the process $\tilde{G}_n(\cdot)$ defined in Theorem 2, there exist constants $C_2 > C_1 > 0$ such that 
\begin{align}\label{eq:60}
\begin{split}
& \; \quad C_1 \leq  \inf_{t \in \mathcal{T}} \mathbb{E}\left[\tilde{G}_n(t)^2 \:|\: (X_i)_{i=1}^n \right] \quad \text{with probability approaching one}, \\
\text{and} \; & \; \sup_{t \in \mathcal{T}} \mathbb{E}\left[\tilde{G}_n(t)^2 \:|\: (X_i)_{i=1}^n \right] \leq  C_2, \quad \text{a.s.} \\[-5pt]
\end{split}
\end{align}
\end{lemma}

\begin{lemma}
\label{lemma10}
Under Assumptions (A1)-(A5), for $\hat{J}_{\tau}(t)$ defined in \eqref{eq:29},  we have
\begin{align} \label{eq:35}
\sup_{t \in \mathcal{T}} \lVert \hat{J}_{\tau}(t) - J_{\tau}(t) \rVert \lesssim_{P} \left(\frac{\log T \log n}{n h_n} \right)^{1/2} + h_n + \delta_T, \\[-5pt]
\end{align}
if $h_n = o(1)$, $\delta_T = o(1)$ and $\log T \log n = o(n h_n)$.
\end{lemma}

\section{Proofs of Theorems} \label{theorems}
Throughout the following proofs, $C, C_1, C_2$, etc. will denote constants that do not depend on $n$ but may have different values in different parts of the proofs.

\begin{proof}[Proof of Theorem 1]
Our proof follows the arguments used in the proof of Theorem 5.1 in \citet{chao2017quantile}. With some rearrangement of terms, for each $t \in \bm{t}$, we have
\begin{align}
& \; \mathbb{P}_n\left[\psi(Y(t), X; \hat{\bm{\beta}}_{\tau}(t), \tau)\right] \\
= & \; n^{-1/2} \mathbb{G}_n\left[\psi(Y(t), X; \hat{\bm{\beta}}_{\tau}(t), \tau)\right] + \mathbb{E}\left[\psi(Y(t), X; \hat{\bm{\beta}}_{\tau}(t), \tau)\right] \\
= & \; n^{-1/2} \mathbb{G}_n\left[\psi(Y(t), X; \hat{\bm{\beta}}_{\tau}(t), \tau)\right] + \vartheta(\hat{\bm{\beta}}_{\tau}(t); t,\tau) \\
= & \; J_{\tau}(t)(\hat{\bm{\beta}}_{\tau}(t)-\bm{\beta}_{\tau}(t)) + n^{-1/2} \mathbb{G}_n\left[\psi(Y(t), X; \bm{\beta}_{\tau}(t), \tau)\right] + \\
  & \; \left\{ \vartheta(\hat{\bm{\beta}}_{\tau}(t); t,\tau) - \vartheta(\bm{\beta}_{\tau}(t); t,\tau) - J_{\tau}(t)(\hat{\bm{\beta}}_{\tau}(t)-\bm{\beta}_{\tau}(t)) \right\} + \\
  & \: \left\{ n^{-1/2} \mathbb{G}_n[\psi(Y(t), X; \hat{\bm{\beta}}_{\tau}(t), \tau)] - n^{-1/2} \mathbb{G}_n\left[\psi(Y(t), X; \bm{\beta}_{\tau}(t), \tau)\right] \right\}. \\[-5pt]
\end{align}
\begin{align}
\begin{split}
\text{Let} \;\; r_{n,1}(t,\tau) \coloneqq & \; J_{\tau}(t)^{-1}\mathbb{P}_n\left[\psi(Y(t), X; \hat{\bm{\beta}}_{\tau}(t), \tau)\right], \\
r_{n,2}(t,\tau) \coloneqq & \; - J_{\tau}(t)^{-1}\left\{ \vartheta(\hat{\bm{\beta}}_{\tau}(t); t,\tau) - \vartheta(\bm{\beta}_{\tau}(t); t,\tau) - J_{\tau}(t)(\hat{\bm{\beta}}_{\tau}(t)-\bm{\beta}_{\tau}(t)) \right\}, \\
r_{n,3}(t,\tau) \coloneqq & \; - n^{-1/2} J_{\tau}(t)^{-1}\left\{\mathbb{G}_n[\psi(Y(t), X; \hat{\bm{\beta}}_{\tau}(t), \tau)] - \mathbb{G}_n\left[\psi(Y(t), X; \bm{\beta}_{\tau}(t), \tau)\right] \right\}, \\[-5pt]
\end{split}
\end{align}
we then have 
\begin{align}
\begin{split}
\hat{\bm{\beta}}_{\tau}(t)-\bm{\beta}_{\tau}(t) = & \; - n^{-1/2} J_{\tau}(t)^{-1} \mathbb{G}_n\left[\psi(Y(t), X; \bm{\beta}_{\tau}(t), \tau)\right] + r_{n,1}(t,\tau) + r_{n,2}(t,\tau) + r_{n,3}(t,\tau) \\
= & \; -\frac{1}{n} J_{\tau}(t)^{-1} \sum_{i=1}^{n}\psi(Y_i(t), X_i; \bm{\beta}_{\tau}(t), \tau) + r_{n,1}(t,\tau) + r_{n,2}(t,\tau) + r_{n,3}(t,\tau). \\[-5pt]
\end{split}
\end{align}

Suppose Assumptions (A1)-(A3) hold. Under the more general condition $\log T \:\! \log n = o(n)$, we now bound the remainder terms $r_{n,j}(t,\tau) \; (j=1,2,3)$ as \eqref{eq:303}, \eqref{eq:304} and \eqref{eq:305}, from which Theorem 1 directly follows by further assuming $\log T \:\! \log n = o(n^{1/3})$.
\begin{align} \label{eq:303}
\sup_{t \in \bm{t}}\lVert r_{n,1}(t,\tau) \rVert \leq &\;  \frac{1}{\inf_{t \in \bm{t}}\lambda_{\min}(J_{\tau}(t))} \frac{\xi d}{n}  \quad \text{a.s.} \\[-5pt]
\end{align}

\noindent For any $\kappa_n = o(n)$, sufficiently large $n$, and a constant $C$ independent of $n$, 
\begin{align} \label{eq:304}
\!\!\!\! P\left(\sup_{t \in \bm{t}} \lVert r_{n,2}(t,\tau) \rVert \leq C\left[\left(\frac{\log T}{n}\right)^{1/2} \!\!\! + \frac{\log T}{n} + \left(\frac{\kappa_n}{n}\right)^{1/2} + \frac{\kappa_n}{n}\right]^2\right) \geq  1 - e^{-\kappa_n}.\\[-5pt]
\end{align}
\begin{align} \label{eq:305}
P\left(\sup_{t \in \bm{t}} \lVert r_{n,3}(t,\tau) \rVert \leq C\left[\left(\frac{\log T}{n} \log n \right)^{1/2} + \left(\frac{\kappa_n}{n} \right)^{1/2} \right]^{3/2}\right) \geq 1 - 2 e^{-\kappa_n}. \\[-5pt]
\end{align}

\bigskip

\textbf{Bound on the first residual term.}  Using standard arguments on duality theory for convex optimization, which are detailed in Lemma 34 in \citet{belloni2019conditional}, we obtain
\begin{align}
\sup_{t \in \bm{t}} \left\lVert \mathbb{P}_n\left[\psi(Y_i(t),X_i; \hat{\bm{\beta}}_{\tau}(t),\tau) \right] \right\rVert \leq \frac{d}{n} \max_{1\leq i \leq n} \lVert X_i \rVert = \frac{\xi d}{n}. \\[-5pt]
\end{align}

\noindent Therefore, 
\begin{align}
\sup_{t \in \bm{t}} \lVert r_{n,1}(t,\tau) \rVert  = & \; \sup_{t \in \bm{t}} \left\lVert J_{\tau}(t)^{-1}\mathbb{P}_n \! \left[\psi(Y(t), X; \hat{\bm{\beta}}_{\tau}(t), \tau)\right] \right\rVert  \\
\leq & \; \sup_{t \in \bm{t}}\lambda_{\max}(J_{\tau}(t)^{-1}) \sup_{t \in \bm{t}} \left\lVert \mathbb{P}_n \! \left[\psi(Y_i(t),X_i; \hat{\bm{\beta}}_{\tau}(t),\tau) \right] \right\rVert = \frac{1}{\inf_{t \in \bm{t}}\lambda_{\min}(J_{\tau}(t))} \frac{\xi d}{n}, \\[-5pt]
\end{align}
which is the inequality \eqref{eq:303}.

\bigskip

\textbf{Bound on the second and third residual terms.} To bound $r_{n,2}$ as given in inequality \eqref{eq:304}, observe that by Lemma \ref{lemma2} with $\upsilon = 2$, we have
\begin{align}
\Omega_{1,n} \coloneqq \left\{ \sup_{t \in \bm{t}} \lVert \hat{\bm{\beta}}_{\tau}(t) - \bm{\beta}_{\tau}(t) \rVert \leq \frac{4 s_{n,1}}{\inf_{t \in \bm{t}}\lambda_{\min}(J_{\tau}(t))} \right\} \supseteq \left\{s_{n,1} < \frac{\inf_{t \in \bm{t}} \lambda^2_{\min}(J_{\tau}(t))}{8\xi \overline{f'} \lambda_{\max}(\mathbb{E}[XX'])} \right\} \eqqcolon \Omega_{2,n}. \\[-5pt]
\end{align}

\noindent where $s_{n,1} \coloneqq \lVert \mathbb{P}_n - P \rVert _{\mathcal{G}_1}$. Define the event
\begin{align}
\Omega_{3,n} \coloneqq \left\{s_{n,1} \leq C\left[\left(\frac{\log T}{n}\right)^{1/2} \!\!\! + \frac{\log T}{n} + \left(\frac{\kappa_n}{n}\right)^{1/2} + \frac{\kappa_n}{n}\right] \right\}, \\[-5pt]
\end{align}

\noindent We have $P(\Omega_{3,n}) \geq 1 - e^{-\kappa_n}$ from Lemma \ref{lemma4}. Under the assumptions $\log T \log n = o(n)$ and $\kappa_n = o(n)$, for sufficiently large $n$, 
\begin{align} \label{eq:61}
C\left[\left(\frac{\log T}{n}\right)^{1/2} \!\!\! + \frac{\log T}{n} + \left(\frac{\kappa_n}{n}\right)^{1/2} + \frac{\kappa_n}{n}\right] \leq \frac{\inf_{t \in \bm{t}} \lambda^2_{\min}(J_{\tau}(t))}{8\xi \overline{f'} \lambda_{\max}(\mathbb{E}[XX'])}.\\[-5pt]
\end{align}

\noindent Therefore, for all $n$ for which \eqref{eq:61} holds, $\Omega_{3,n} \subseteq \Omega_{2,n} \subseteq \Omega_{1,n}$. Given this, for a constant $C_2$ that does not depend on $n$, we have
\begin{align}
\Omega_{3,n} \subseteq \left\{ \sup_{t \in \bm{t}} \lVert \hat{\bm{\beta}}_{\tau}(t) - \bm{\beta}_{\tau}(t) \rVert \leq C_2 \left[\left(\frac{\log T}{n}\right)^{1/2} \!\!\! + \frac{\log T}{n} + \left(\frac{\kappa_n}{n}\right)^{1/2} + \frac{\kappa_n}{n}\right] \right\} \eqqcolon \Omega_{4,n}. \\[-5pt]
\end{align}

\noindent In particular, for all $n$ for which \eqref{eq:61} holds, $P\left(\Omega_{4,n}\right) \geq 1 - e^{-\kappa_n}$. On $\Omega_{4,n}$, by Lemma \ref{lemma1}, for $\forall \; t \in \bm{t}$,
\begin{align}
& \; \left\lVert \vartheta(\hat{\bm{\beta}}_{\tau}(t); t, \tau) - \vartheta(\bm{\beta}_{\tau}(t); t, \tau) - J_{\tau}(t)(\hat{\bm{\beta}}_{\tau}(t)-\bm{\beta}_{\tau}(t)) \right\rVert \\
\leq & \;  \lambda_{\max}(\mathbb{E}[XX'])\overline{f'} \: \xi \:C_2^2 \left[\left(\frac{\log T}{n}\right)^{1/2} \!\!\! + \frac{\log T}{n} + \left(\frac{\kappa_n}{n}\right)^{1/2} + \frac{\kappa_n}{n}\right]^2. \\[-5pt]
\end{align}

\noindent Therefore, 
\begin{align}
 & \; \left\lVert J_{\tau}(t)^{-1} \left\{\vartheta(\hat{\bm{\beta}}_{\tau}(t); t, \tau) - \vartheta(\bm{\beta}_{\tau}(t); t, \tau) - J_{\tau}(t)(\hat{\bm{\beta}}_{\tau}(t)-\bm{\beta}_{\tau}(t)) \right\} \right\rVert \\
\leq & \; \lambda_{\max}(J_{\tau}(t)^{-1}) \left\lVert \vartheta(\hat{\bm{\beta}}_{\tau}(t); t, \tau) - \vartheta(\bm{\beta}_{\tau}(t); t, \tau) - J_{\tau}(t)(\hat{\bm{\beta}}_{\tau}(t)-\bm{\beta}_{\tau}(t)) \right\rVert \\
\leq & \; \frac{1}{\inf_{t \in \bm{t}}\lambda_{\min}(J_{\tau}(t))} \lambda_{\max}(\mathbb{E}[XX'])\overline{f'} \: \xi \:C_2^2 \left[\left(\frac{\log T}{n}\right)^{1/2} \!\!\! + \frac{\log T}{n} + \left(\frac{\kappa_n}{n}\right)^{1/2} + \frac{\kappa_n}{n}\right]^2. \\[-5pt]
\end{align}

\noindent We then have
\begin{align}
\Omega_{4,n} \subseteq \left\{\sup_{t \in \bm{t}} \lVert r_{n,2}(t,\tau) \rVert \leq C_3 \left[\left(\frac{\log T}{n}\right)^{1/2} \!\!\! + \frac{\log T}{n} + \left(\frac{\kappa_n}{n}\right)^{1/2} + \frac{\kappa_n}{n}\right]^2 \right\}, \\[-5pt] 
\end{align}
for some constant $C_3$ that is independent of $n$, and this gives inequality \eqref{eq:304}.

\bigskip

To bound $r_{n,3}$ as given in inequality \eqref{eq:305}, first observe that for any $\delta > 0$, on the set $\left\{\sup_{t \in \bm{t}} \lVert \hat{\bm{\beta}}_{\tau}(t) - \bm{\beta}_{\tau}(t) \rVert \leq \delta \right\}$, we have the following inequality
\begin{align} \label{eq:62}
\sup_{t \in \bm{t}}\lVert r_{n,3}(t,\tau) \rVert \leq \frac{1}{\inf_{t \in \bm{t}} \lambda_{\min}(J_{\tau}(t))} \lVert \mathbb{P}_n - P \rVert _{\mathcal{G}_2(\delta)}. \\[-5pt]
\end{align}

\noindent To see this, note that $\forall \: t \in \bm{t}$,
\begin{align}
& \; \lVert r_{n,3}(t,\tau) \rVert \\
\leq & \; \lVert J_{\tau}(t)^{-1} \rVert \: \left\lVert n^{-1/2}\left( \mathbb{G}_n[\psi(Y(t), X; \hat{\bm{\beta}}_{\tau}(t), \tau)] - \mathbb{G}_n\left[\psi(Y(t), X; \bm{\beta}_{\tau}(t), \tau)\right] \right) \right\rVert \\
= & \; \lVert J_{\tau}(t)^{-1} \rVert \sup_{\bm{u} \in \mathcal{S}^{d-1}} \left\{ \mathbb{P}_n \left[\bm{u}'X \left(\mathbf{1}\left\{Y(t) \leq X'\hat{\bm{\beta}}_{\tau}(t) \right\} - \tau \right) \right] \right. \\
\mathrel{\phantom{=}} & \left. - \; \mathbb{E} \left[\bm{u}'X\left(\mathbf{1}\left\{Y(t) \leq X'\hat{\bm{\beta}}_{\tau}(t) \right\} - \tau \right) \right] - \mathbb{P}_n \left[\bm{u}'X \left(\mathbf{1}\left\{Y(t) \leq X'\bm{\beta}_{\tau}(t) \vphantom{\left(\lambda\right)^2} \right\} - \tau \right) \right] \right. \\
\mathrel{\phantom{=}} & \left. + \; \mathbb{E} \left[\bm{u}'X \left(\mathbf{1}\left\{Y(t) \leq X'\bm{\beta}_{\tau}(t) \vphantom{\mathbb{P}_n \left[\bm{u}'X \left(\mathbf{1}\left\{Y(t) \leq X'\hat{\bm{\beta}}_{\tau}(t) \right\} - \tau \right) \right]} \right\} - \tau \right) \right] \right\} \\
= & \; \lVert J_{\tau}(t)^{-1} \rVert \sup_{\bm{u} \in \mathcal{S}^{d-1}} \left\{ \mathbb{P}_n \left[\bm{u}'X \left(\mathbf{1}\left\{Y(t) \leq X'\hat{\bm{\beta}}_{\tau}(t) \right\} - \mathbf{1}\left\{Y(t) \leq X'\bm{\beta}_{\tau}(t) \vphantom{\left(\lambda\right)^2}\right\} \right) \right] \right. \\
\mathrel{\phantom{=}} & \left.\kern-\nulldelimiterspace - \; \mathbb{E} \left[\bm{u}'X \left(\mathbf{1}\left\{Y(t) \leq X'\hat{\bm{\beta}}_{\tau}(t) \right\} - \mathbf{1}\left\{Y(t) \leq X'\bm{\beta}_{\tau}(t) \vphantom{\left(\lambda\right)^2} \right\} \right) \right]\right\} \\
\leq & \; \lVert J_{\tau}(t)^{-1} \rVert \lVert \mathbb{P}_n - P \rVert _{\mathcal{G}_2(\delta)} = \lambda_{\max}(J_{\tau}(t)^{-1}) \lVert \mathbb{P}_n - P \rVert _{\mathcal{G}_2(\delta)} \leq \frac{1}{\inf_{t \in \bm{t}} \lambda_{\min}(J_{\tau}(t))} \lVert \mathbb{P}_n - P \rVert _{\mathcal{G}_2(\delta)}. \\[-5pt] 
\end{align}

\noindent Therefore, for any $\delta, a > 0$, 
\begin{align} \label{eq:63}
P(\sup_{t \in \bm{t}} \lVert r_{n,3}(t,\tau) \rVert \!\geq a) \leq P(\sup_{t \in \bm{t}} \lVert \hat{\bm{\beta}}_{\tau}(t) - \bm{\beta}_{\tau}(t) \rVert \!\geq \delta) \!+\! P(\frac{\lVert \mathbb{P}_n - P \rVert _{\mathcal{G}_2(\delta)}}{\inf_{t \in \bm{t}} \lambda_{\min}(J_{\tau}(t))} \!\geq a).  \\[-5pt]
\end{align}
 
\noindent Now let $\delta = \delta_n \coloneqq C \left(\left(\frac{\log T}{n}\log n \right)^{1/2} + \left(\frac{\kappa_n}{n}\right)^{1/2}\right)$ for some constant $C$, and
\begin{align} \label{eq:64}
a \coloneqq & C \zeta_n(\delta_n, \kappa_n) \\
= & C \left[ C^{1/2} \left(\left(\frac{\log T}{n}\log n \right)^{1/2} \!\!\! +  \! \left(\frac{\kappa_n}{n}\right)^{1/2}\right)^{1/2} \left(\left(\frac{\log T}{n}\log n \right)^{1/2} \!\!\! + \! \left(\frac{\kappa_n}{n}\right)^{1/2}\right) + \frac{\log T}{n}\log n + \frac{\kappa_n}{n} \right], \\[-5pt]
\end{align}
where $\zeta_n$ is defined in \eqref{eq:55} in Lemma \ref{lemma4}. Given  $\log T \log n = o(n)$ and $\kappa_n = o(n)$, the last two terms in the aforementioned expression are negligible compared to the first term for large $n$. Therefore, for some sufficiently large constant $C_1$, 
\begin{align}
a = C \zeta_n(\delta_n, \kappa_n) \leq C_1 \left[\left(\frac{\log T}{n} \log n \right)^{1/2} + \left(\frac{\kappa_n}{n} \right)^{1/2} \right]^{3/2}. \\[-5pt]
\end{align}

\noindent Given $\delta_n \gg n^{-1}$, apply inequality \eqref{eq:54} in Lemma \ref{lemma4} to obtain
\begin{align} \label{eq:65}
\begin{split}
& \; P\left(\frac{\lVert \mathbb{P}_n - P \rVert _{\mathcal{G}_2(\delta_n)}}{\inf_{t \in \bm{t}} \lambda_{\min}(J_{\tau}(t))} \geq C_1 \left[\left(\frac{\log T}{n} \log n \right)^{1/2} + \left(\frac{\kappa_n}{n} \right)^{1/2} \right]^{3/2}\right)\\
\leq & \; P\left(\frac{\lVert \mathbb{P}_n - P \rVert _{\mathcal{G}_2(\delta_n)}}{\inf_{t \in \bm{t}} \lambda_{\min}(J_{\tau}(t))} \geq  C \zeta_n(\delta_n, \kappa_n) \right) \leq e^{-\kappa_n}. \\[-5pt]
\end{split}
\end{align}

\noindent Recall that $P\left(\sup_{t \in \bm{t}} \lVert \hat{\bm{\beta}}_{\tau}(t) - \bm{\beta}_{\tau}(t) \rVert \geq C_2 \left[\left(\frac{\log T}{n}\right)^{1/2} \! + \frac{\log T}{n} + \left(\frac{\kappa_n}{n}\right)^{1/2} \! + \frac{\kappa_n}{n}\right] \right) \leq e^{-\kappa_n}$ for some constant $C_2$. Given $C_2 \left[\left(\frac{\log T}{n}\right)^{1/2} \! + \frac{\log T}{n} + \left(\frac{\kappa_n}{n}\right)^{1/2} \! + \frac{\kappa_n}{n}\right] \leq \delta_n$ for large $n$, we have 
\begin{align} \label{eq:66}
P\left(\sup_{t \in \bm{t}} \lVert \hat{\bm{\beta}}_{\tau}(t) - \bm{\beta}_{\tau}(t) \rVert \geq \delta_n \right) \leq e^{-\kappa_n}, \\[-5pt]
\end{align}
for large $n$. Given \eqref{eq:63}, inequality \eqref{eq:305} follows directly from \eqref{eq:65} and \eqref{eq:66}. 

\bigskip

Now further assume that $\log T \log n = o(n^{1/3})$ and $\kappa_n = o(n^{1/3})$, then for any $\epsilon > 0$, with sufficiently large $n$, \eqref{eq:304} and \eqref{eq:305} lead to $P(\sqrt{n}\sup_{t \in \bm{t}} \lVert r_{n,j}(t,\tau) \rVert \leq \epsilon) \geq 1 - 2 e^{-\kappa_n}$ for $j = 2, 3$. Combined with \eqref{eq:303}, Theorem 1 immediately follows.
\end{proof}

\bigskip

\begin{proof}[Proof of Theorem 2]
Theorem 2 directly follows from Theorem 4 by observing that $\hat{\mu}_n(t) =\hat{\mu}_n(t)^{LI}$ and $\mu(t) = \tilde{\mu}(t)$ for any $t \in \bm{t}$.
\end{proof}

\bigskip

\begin{proof}[Proof of Theorem 3]
With $\mathbb{G}_n(t)$ defined in (3.1), the following holds for any $t \in [t_l, t_{l+1}]$,
\begin{align}\label{eq:58}
& \: \hat{\mu}_n(t)^{LI} + \frac{1}{\sqrt{n}} \mathbb{G}_n(t) \\ 
= & \; \frac{t_{l+1}- t}{t_{l+1}-t_l} \hat{\mu}_n(t_l) +  \frac{t- t_l}{t_{l+1}-t_l} \hat{\mu}_n(t_{l+1}) + \frac{1}{\sqrt{n}} \mathbb{G}_n(t)  \\
= & \;  \frac{t_{l+1}- t}{t_{l+1}-t_l} \left(\mu(t_l) - \frac{1}{\sqrt{n}} \mathbb{G}_n(t_l) + \bm{a}' r_n(t_l, \tau) \right) \\
\mathrel{\phantom{=}} & + \frac{t- t_l}{t_{l+1}-t_l} \left(\mu(t_{l+1}) - \frac{1}{\sqrt{n}} \mathbb{G}_n(t_{l+1}) + \bm{a}' r_n(t_{l+1}, \tau) \right) +  \frac{1}{\sqrt{n}} \mathbb{G}_n(t). \quad (\text{by Theorem 1}) \\[-5pt]
\end{align}

\noindent Therefore, \\[-5pt]
\begin{align}\label{eq:67}
\begin{split}
& \; \left|\hat{\mu}_n(t)^{LI} - \tilde{\mu}(t) + \frac{1}{\sqrt{n}} \mathbb{G}_n(t) \right| \\
= & \; \left|\frac{1}{\sqrt{n}} \mathbb{G}_n(t) - \frac{t_{l+1}- t}{t_{l+1}-t_l} \frac{1}{\sqrt{n}} \mathbb{G}_n(t_l) - \frac{t- t_l}{t_{l+1}-t_l} \frac{1}{\sqrt{n}} \mathbb{G}_n(t_{l+1}) \right. \\
\mathrel{\phantom{=}} & \left.\kern-\nulldelimiterspace + \: \frac{t_{l+1}- t}{t_{l+1}-t_l} \bm{a}'r_n(t_l, \tau) + \frac{t- t_l}{t_{l+1}-t_l} \bm{a}' r_n(t_{l+1}, \tau) \right| \\
\leq & \; \frac{1}{\sqrt{n}} \sup_{v,\: s \in \mathcal{T}, \: |v-s| \leq \delta_{T(n)}} \left|\mathbb{G}_n(v) - \mathbb{G}_n(s) \right| \: + \:  \sup_{t \in \bm{t}} \: \lVert r_n(t,\tau) \rVert. \\[-5pt]
\end{split}
\end{align}

\noindent In \eqref{eq:67}, we explicitly write $T$ as a function of $n$ in $\delta_{T(n)}$ to emphasize its dependence on $n$. We have shown that $\sup_{t \in \bm{t}} \: \lVert r_n(t,\tau) \rVert = o_p(n^{-1/2})$ in Theorem 1 for $\log T \log n = o(n^{1/3})$. If we can show that 
\begin{align}\label{eq:68}
\sup_{v,\: s \: \in \mathcal{T}, \: |v-s| \leq \delta_{T(n)}} \left|\mathbb{G}_n(v) - \mathbb{G}_n(s) \right| = o_p(1), \\[-5pt]
\end{align}
for $\delta_{T(n)} = o(1)$, then (4.3) immediately follows from \eqref{eq:67}. Therefore, it remains to show \eqref{eq:68}, i.e., we need to prove that for $\forall \: c > 0$, 
\begin{align}\label{eq:69}
\limsup_{n \rightarrow \infty} \: P \left(\sup_{v,\: s \: \in \mathcal{T}, \: |v-s| \leq \delta_{T(n)}} \left|\mathbb{G}_n(v) - \mathbb{G}_n(s) \right| > c \right) = 0. \\[-5pt]
\end{align}

We now prove \eqref{eq:69} by contradiction. Define $\zeta \left(n, \: \delta \right) \coloneqq \sup_{v,\: s \: \in \mathcal{T}, \: |v-s| \leq \delta } \left|\mathbb{G}_n(v) - \mathbb{G}_n(s) \right|$. If \eqref{eq:69} does not hold, i.e., $\limsup_{n \rightarrow \infty} \: P \left( \zeta \left(n, \: \delta_{T(n)} \right)  > c \right) = \epsilon$ 
for some $\epsilon > 0$, then there exists a subsequence $\{n_k\}_{k \geq 1}$ such that 
\begin{align}\label{eq:72}
\lim_{k \rightarrow \infty} \: P \left( \zeta \left(n_k,\: \delta_{T(n_k)} \right) > c \right) = \epsilon. \\[-5pt]
\end{align}

\noindent Under Assumptions (A1)-(A6), we have proven in Lemma \ref{lemma8} that
\begin{align}
\lim_{\delta \downarrow 0} \limsup_{n \rightarrow \infty} P\left( \zeta \left(n, \: \delta \right) > c \right) = 0.  \quad (\text{asymptotic equicontinuity}) \\[-5pt]
\end{align}

\noindent Therefore, there exists $\delta_0 > 0$ such that
\begin{align} \label{eq:73} 
\limsup_{n \rightarrow \infty} P\left( \zeta \left(n, \: \delta_0 \right) > c \right) < \epsilon/2. \\[-5pt]
\end{align}

\noindent By assumption, $\lim_{k \rightarrow \infty} \delta_{T(n_k)} = 0$, so we can take sufficiently large $K_0$ such that $\delta_{T(n_{K_0})} < \delta_0$, and for $\forall \; k \geq K_0$,
\begin{align}\label{eq:74}
P \left( \zeta \left(n_k,\: \delta_0 \right) > c \right) \geq P \left( \zeta \left(n_k,\: \delta_{T(n_k)} \right) > c \right)  > \epsilon/2, \quad (\text{by} \;\; \eqref{eq:72}) \\[-5pt]
\end{align}
which contradicts with \eqref{eq:73}. Therefore, \eqref{eq:68} holds and (4.3) follows from \eqref{eq:67}. (4.4) is then a direct consequence of (4.3) and Lemma \ref{lemma8}. 
\end{proof}

\bigskip

\begin{proof}[Proof of Theorem 4]
Provided we can show that for any given sequence of non-stochastic, bounded vectors $(X_i)_{i=1}^n$ in $\mathbb{R}^d$, there exists a sequence of zero-mean Gaussian processes $(\tilde{G}_n)_{n \geq 1}$ such that 
\begin{itemize}
\item[(i)] the sample paths of $\tilde{G}_n$ are a.s. continuous and the covariance functions of $\tilde{G}_n$ coincide with those of $\mathbb{G}_n$ for each $n$, i.e., $\mathbb{E}\left[\tilde{G}_n(t) \tilde{G}_n(s)\right] = \mathbb{E}\left[\vphantom{\tilde{G}_n(t)} \mathbb{G}_n(t) \mathbb{G}_n(s) \right]$ for all $t,\: s \in \mathcal{T}$; 
\item[(ii)] $\tilde{G}_n$ closely approximates $\mathbb{G}_n$ in sup norm, i.e., $\sup_{t \in \mathcal{T}} \left|\tilde{G}_n(t) - \mathbb{G}_n(t) \right| = o_p(1)$,
\end{itemize}
then conditioning on $\left(X_i\right)_{i=1}^{n}$, (4.7) in Theorem 4 immediately follows from (4.3) in Theorem 3. If we additionally assume that $\delta_{T} = o(n^{-1/2})$, (4.8) then follows from (4.7) by observing
\begin{align}
\left|\tilde{\mu}(t) - \mu(t) \right|  = \left|\frac{t_{l+1}- t}{t_{l+1}-t_l} \mu(t_l) + \frac{t - t_l}{t_{l+1}-t_l} \mu(t_{l+1}) - \mu(t) \right| \leq \delta_T \overline{\mu'} = o(1/\sqrt{n}). \\[-5pt]
\end{align}

We next prove (i) and (ii) following arguments used in the proof of Lemma 14 in \citet{belloni2019conditional}. More specifically, we first define a sequence of projections $\pi_j: \mathcal{T} \rightarrow \mathcal{T}, \; j= 0, 1, 2, \dots, \infty$ by $\pi_j(t) = l/2^j$ if $t \in \left((l-1)/2^j, l/2^j\right]$, for $l=1, \dots, 2^j$. Given a process $G$ in $l^{\infty}(\mathcal{T})$, the sample paths of its projection $G \circ \pi_j$ are by definition step functions with at most $2^j$ steps. Therefore, we can identify the process $G \circ \pi_j$ with a random vector $G \circ \pi_j$ in $\mathbb{R}^{2^j}$. Similarly, we can also identify a random vector $W$ in $\mathbb{R}^{2^j}$ with a process $W$ in $l^{\infty}(\mathcal{T})$ whose sample paths are step functions with at most $2^j$ steps. The proof of (i) and (ii) then consists of the following 4 steps, for some $j=j_n \rightarrow \infty$:
\begin{itemize}
\item[(1)] $\tilde{r}_{n,1} = \sup_{t \in \mathcal{T}} \left|\mathbb{G}_n(t) - \mathbb{G}_n \circ \pi_j(t) \right| = o_p(1)$;
\item[(2)] there exists $\mathcal{N}_{nj} \overset{d}{=} \text{MVN}\left(0, \text{Cov}\left[\mathbb{G}_n \circ \pi_j \right]\right)$ such that $\tilde{r}_{n,2} = \lVert \mathcal{N}_{nj} - \mathbb{G}_n \circ \pi_j \rVert = o_p(1)$;
\item[(3)] there exists a Gaussian process $\tilde{G}_n$ satisfying (i) such that $\mathcal{N}_{nj} = \tilde{G}_n \circ \pi_j$ a.s.;
\item[(4)] $\tilde{r}_{n,3} = \sup_{t \in \mathcal{T}} \left|\tilde{G}_n(t) - \tilde{G}_n \circ \pi_j(t) \right| = o_p(1)$.
\end{itemize}

\noindent Given (1)-(4), the existence of a sequence of Gaussian processes $(\tilde{G}_n)_{n \geq 1}$ that satisfy both (i) and (ii) follows directly from the triangle inequality, i.e., $\sup_{t \in \mathcal{T}} \left|\tilde{G}_n(t) - \mathbb{G}_n(t) \right| \leq \tilde{r}_{n,1} + \tilde{r}_{n,2} + \tilde{r}_{n,3} = o_p(1)$.

We now prove relations (1)-(4).

\textit{Proof of step (1):} We have shown in \eqref{eq:68} in Theorem 3, which in turn depends on the asymptotic tightness of the process $\mathbb{G}_n(\cdot)$ shown in Lemma \ref{lemma8}, that for any sequence $\delta_n \downarrow 0$, 
\begin{align} \label{eq:75}
\sup_{|t-s| \leq \delta_n} \left|\mathbb{G}_n(t) - \mathbb{G}_n(s) \right| = o_p(1). \\[-5pt]
\end{align}

\noindent We note that for the empirical processes $\mathbb{G}_n$ defined in \eqref{eq:68}, the covariates $(X_i)_{i=1}^n$ are assumed to be i.i.d. random vectors in $\mathbb{R}^d$, but \eqref{eq:68} and its proof also hold for any sequence of non-stochastic vectors $(X_i)_{i=1}^n$ in $\mathbb{R}^d$ such that $\lVert X_i \rVert$ is bounded a.s. for each $i$ (Assumption (A1)). Therefore, for any given sequence of such non-stochastic vectors $(X_i)_{i=1}^n$, \eqref{eq:75} still holds and leads to
\begin{align}
\tilde{r}_{n,1} = \; \sup_{t \in \mathcal{T}} \left|\mathbb{G}_n(t) - \mathbb{G}_n \circ \pi_j(t) \right|  \leq  \sup_{|t-s| \leq 2^{-j_n}} \left|\mathbb{G}_n(t) - \mathbb{G}_n(s) \right| = o_p(1), \;\; \text{for any} \;\; j_n \rightarrow \infty. \\[-5pt]
\end{align}

\textit{Proof of step (2):} We use Yurinskii's coupling to show relation (2). For completeness, we cite Yurinskii's coupling from \citet{belloni2019conditional}. Let $V_1, \dots, V_n$ be independent zero-mean $p$-vectors such that $\kappa \coloneqq \sum_{i=1}^n \mathbb{E}\left[\lVert V_i \rVert^3 \right]$ is finite. Let $S = V_1 + \dots + V_n$. Then for each $\delta > 0$, there exists a random vector $W$ with a $\text{MVN}(0, \text{Cov}(S))$ distribution such that
\begin{align} \label{eq:76}
P\left(\lVert S - W \rVert > 3\delta \right) \leq C_0 B \: \left(1 + \frac{|\log(1/B)|}{p} \right) \quad \text{where} \;\; B \coloneqq \kappa p \delta^{-3}, \\[-5pt]
\end{align}
for some universal constant $C_0$. 

Now apply the coupling to the zero-mean $2^j$-vectors $V_i \; (i=1, \dots, n)$ such that the $l$-th component of $V_i$ is $V_{i,\: l} = \bm{a}' J_{\tau}(t_l)^{-1} X_i\left(\mathbf{1}\{Y_i(t_l) \leq X_i'\bm{\beta}_{\tau}(t_l)\} - \tau\right)$, where $t_l = l/2^j$ and $l=1, \dots, 2^j$. By definition of $V_i$, we have $\mathbb{G}_n \circ \pi_j = \sum_{i=1}^n V_i/\sqrt{n}$. Then
\begin{align}
\lVert V_i \rVert^2 = \sum_{l=1}^{2^j} V_{i,\: l}^2 \leq \sum_{l=1}^{2^j} \left|\bm{a}' J_{\tau}(t_l)^{-1} X_i \right|^2 \leq \sum_{l=1}^{2^j} \lambda_{\max} \left(J_{\tau}(t_l)^{-1} \right)^2 \xi^2 \leq \left(\frac{1}{\inf_{t\in \mathcal{T}}\lambda_{\min}(J_{\tau}(t))}\right)^{2} 2^j \xi^2. \\[-5pt]
\end{align}

\noindent Therefore, $\lVert V_i \rVert^3 = \left( \lVert V_i \rVert^2 \right)^{3/2} \lesssim 2^{3j/2}$, and $\sum_{i=1}^n \mathbb{E}\left[\lVert V_i \rVert^3 \right] \lesssim n \: 2^{3j/2}$. Here we use $\lesssim$ in $a_n \lesssim b_n$ to denote that $a_n \leq C b_n$ holds for all $n$ with a constant $C$ that does not depend on $n$. 

Now choose $j=j_n$ such that $2^{j_n} = n^{\tilde{\epsilon}}$ for some $\tilde{\epsilon} > 0$. By \eqref{eq:76}, there exists $\mathcal{N}_{nj} \overset{d}{=} \text{MVN}\left(0, \text{Cov}\left[\mathbb{G}_n \circ \pi_j \right]\right)$ such that 
\begin{align} \label{eq:77}
P\left(\left\lVert \frac{\sum_{i=1}^n V_i}{\sqrt{n}} - \mathcal{N}_{nj} \right \rVert \geq 3 \delta \right) \lesssim \frac{2^{5j/2}}{\delta^3 n^{1/2}} \left(1+ \frac{\left|\log \frac{\delta^3 n^{1/2}}{2^{5j/2}} \right|}{2^j} \right). \\[-5pt]
\end{align}

\noindent Setting $\delta_n = \left(2^{5j} \log n / n \right)^{1/6}$, the second term in the r.h.s. of \eqref{eq:77} goes to 0, so
\begin{align}
P\left(\left\lVert \frac{\sum_{i=1}^n V_i}{\sqrt{n}} - \mathcal{N}_{nj} \right \rVert \geq 3 \delta_n \right) \lesssim \frac{2^{5j/2}}{\delta^3 n^{1/2}} = \frac{1}{\left(\log n \right)^{1/2}} \downarrow 0. \\[-5pt]
\end{align}

\noindent The proof of step (2) is completed if we have $\delta_n \downarrow 0$. To achieve this, we can choose $\tilde{\epsilon}$ such that $n^{5 \tilde{\epsilon}} \log n = o(n)$. 

\textit{Proof of step (3):} The existence of a Gaussian process $\tilde{G}_n$ with properties stated in (i) such that $\mathcal{N}_{nj} = \tilde{G}_n \circ \pi_j$ a.s. can be established using Lemma 17 in \citet{belloni2019conditional}.

\textit{Proof of step (4):} We first show that for a sequence of zero-mean Gaussian processes $(\tilde{G}_n)_{n \geq 1}$ that satisfy properties stated in (i), for any $\gamma \in (0, 1)$, 
\begin{align} \label{eq:78}
\sup_{|t-s| \leq \gamma} \left| \tilde{G}_n(t) - \tilde{G}_n(s) \right| = O_p \left(\sqrt{\gamma \log(1/\gamma)} \right). \\[-5pt]
\end{align}
% Again, $\left(X_i\right)_{i=1}^{n}$ are assumed to be non-stochastic and bounded in \eqref{eq:78}.
To show this, for each $n$, we define the following zero-mean Gaussian process $Z_n: \mathcal{T} \times \mathcal{T} \rightarrow \mathbb{R}$:
\begin{align} \label{eq:79}
Z_{n,u} = \tilde{G}_n (t) - \tilde{G}_n (s), \quad  u=(t, s) \in \mathcal{U}, \\[-5pt]
\end{align}  
where $\mathcal{U} \coloneqq \{(t,s): \: t, s \in \mathcal{T}, \: |t-s| \leq \gamma \}$. We have $\sup_{u \in \mathcal{U}} Z_{n,u} = \sup_{|t-s| \leq \gamma} \left| \tilde{G}_n(t) - \tilde{G}_n(s) \right|$. For any $u \in \mathcal{U}$, 
\begin{align} 
\text{Var}\left[Z_{n,u}\right] = & \; \text{Var}\left[\tilde{G}_n(t) - \tilde{G}_n(s) \right] \\
= & \; \mathbb{E}\left[\tilde{G}_n(t)^2 \right] + \mathbb{E}\left[\tilde{G}_n(s)^2 \right] - 2 \: \mathbb{E}\left[\tilde{G}_n(t) \tilde{G}_n(s) \right] \\
= & \; \mathbb{E}\left[\mathbb{G}_n(t)^2 \right] + \mathbb{E}\left[\mathbb{G}_n(s)^2 \right] - 2 \: \mathbb{E}\left[\mathbb{G}_n(t) \mathbb{G}_n(s) \right] \\
= & \; \text{Var}\left[\mathbb{G}_n(t) - \mathbb{G}_n(s) \right] = \mathbb{E}\left[\left|\mathbb{G}_n(t) - \mathbb{G}_n(s) \right|^2 \right] \leq C_0 \left|t-s\right|,  \label{eq:84} \\[-5pt]
\end{align}
for some universal constant $C_0$ that does not depend on $t, s$ or $n$, based on some intermediate results in Lemma \ref{lemma8}. Therefore, $\sigma(Z_n) \coloneqq \sup_{u \in \mathcal{U}} \sigma(Z_{n,u}) \leq \left(C_0 \gamma \right)^{1/2}$.

Similarly, we can show that
\begin{align} \label{eq:81}
\rho_n(u, u') \coloneqq \sigma(Z_{n,u} - Z_{n,u'}) \leq \left(2\:C_0\: \lVert u - u' \rVert_1 \right)^{1/2}, \\[-5pt]
\end{align}
which suggests that 
\begin{align} \label{eq:82}
N(\epsilon,\: \mathcal{U}, \: \rho_n) \leq \left(\frac{L}{\epsilon} \right)^V, \quad \text{for all} \;\; 0 < \epsilon < \epsilon_0, \\[-5pt]
\end{align}
holds for $\epsilon_0 = \sigma(Z_n)$, $L = \gamma^{1/4}$, and $V = 4$. In \eqref{eq:82}, $N(\epsilon,\: \mathcal{U}, \: \rho_n)$ denotes the covering number of $\mathcal{U}$ by $\epsilon$-balls with respect to the metric $\rho_n(u, u')$ in \eqref{eq:81}. Invoking Proposition A.2.7 in \citet{van1996weak}, we have that for any large enough constant $C$, 
\begin{align} \label{eq:83}
\begin{split}
P\left(\sup_{u \in \mathcal{U}} Z_{n,u} > C \lambda_0 \right) \leq & \; \left(\frac{D L C \lambda_0}{\sqrt{V}\:\sigma^2(Z_n)} \right)^V \overline{\Phi}\left(\frac{C \lambda_0}{\sigma(Z_n)}\right), \\
\lesssim & \; \left(\frac{L C \lambda_0}{\sigma^2(Z_n)} \right)^4 \frac{\sigma(Z_n)}{C \lambda_0} \exp \left[-\frac{1}{2} \frac{C^2 \lambda_0^2}{\sigma^2(Z_n)} \right],  \; (*) \\[-5pt]
\end{split}
\end{align}
where $D$ is some universal constant, $\lambda_0 = \sqrt{C_0 \gamma \log(1/\gamma) }$, and $\overline{\Phi}$ denotes the right tail probability of a standard normal variable.  

To obtain $\sup_{u \in \mathcal{U}} Z_{n,u} = O_p(\lambda_0)$, we need to show that the r.h.s in the second line in \eqref{eq:83}, which we denote by $(*)$, goes to 0 for large enough $C$. To show this, let $\zeta = \lambda_0 / \sigma(Z_n)$, we have $(*) \propto C^3 \zeta^7 \exp \left[-\frac{1}{2} C^2 \zeta^2 \right]$. For any fixed $C$, the aforementioned expression is maximized at $\zeta = \frac{\sqrt{7}}{C}$. Substituting $\zeta = \frac{\sqrt{7}}{C}$ leads to $(*) \propto C^{-4}$, which is arbitrarily small for sufficiently large $C$. Therefore, \eqref{eq:78} follows.  

Now take $\gamma = \gamma_n = 2^{-j_n}$ for any $j_n \rightarrow \infty$ in \eqref{eq:78}, we have
\begin{align}
\tilde{r}_{n,3} \leq \sup_{|t-s| \leq 2^{-j_n}} \left|\tilde{G}_n(t) - \tilde{G}_n(s) \right| = O_p \left(\sqrt{2^{-j_n} \log(2^{j_n})} \right) = o_p(1). \\[-5pt]
\end{align}
This completes the proof.
\end{proof}

\bigskip

\begin{proof}[Proof of Theorem 5]
Given Theorem 4 and Proposition \ref{prop1},  Theorem 5 follows from the same arguments used in the proof of Theorem 15 in \citet{belloni2019conditional}.
\end{proof}

\bigskip

\begin{proof}[Proof of Theorem 6]
We first show that for any estimator $\check{\mu}$ of the coefficient function $\mu$
\begin{align} \label{eq:37} 
\lim_{c \rightarrow 0} \limsup\limits_{n \rightarrow \infty} \sup\limits_{\mu \in \mathcal{W}_2^{r}} P\left(\lVert \check{\mu} - \mu \rVert_{\mathcal{L}_2}^2 > c \:\! n^{-1} \right) = 1. \\[-5pt]
\end{align}
Given \eqref{eq:37},  we can follow the argument used in the proof of Theorem 2.1 in \citet{cai2011optimal} to conclude (4.15).  To show \eqref{eq:37}, it suffices to show that the minimax lower bound for estimating $\mu(t)$ at any given location $t$ is $O(n^{-1})$. To see this, consider the special case that $\mu(t)$ is constant over $\mathcal{T}$.  In this case, we have $\lVert \check{\mu} - \mu \rVert_{\mathcal{L}_2}^2 \geq \lVert \bar{\mu} - \mu \rVert_{\mathcal{L}_2}^2$ for $\bar{\mu} \equiv \int \check{\mu}(t) dt$, so we can always replace a given estimator $\check{\mu}(t)$ with $\bar{\mu} \equiv \int \check{\mu}(t) dt$ to reduce the L2 error. 

In the rest of the proof, we focus on estimation of $\mu(t)$ at a given $t$ by a measurable function $\check{\mu}(t)$ of the observed data $\left\{(X_i, \; Y_i(t))\right\}_{i=1}^{n}$, and omit the index $t$ in $Y(t)$, $\check{\mu}(t)$ and $\mu(t)$ for ease of notation. Let $\mathcal{P}$ denote the collection of joint distributions of $(X, Y)$ such that assumptions (A1)-(A3) are satisfied. We now use the Le Cam method to show
\begin{align} \label{eq:38}
\inf_{\check{\mu}} \sup_{P \in \mathcal{P}} \mathbb{E}_P \left[ \left(\check{\mu} - \mu(P) \right)^2 \right] \geq O\left(n^{-1} \right). \\[-5pt]
\end{align} 

We consider two distributions $P_1, P_2 \in \mathcal{P}$. Assume that the densities of $P_1$ and $P_2$ are respectively $p_1(x, y) = p_0(x) \: p_1(y | x)$ and $p_2(x, y) = p_0(x) \: p_2(y | x)$,  where $p_1(y | x) = N(x' \bm{\beta}_1, 1)$ for some $\bm{\beta}_1 \in \mathbb{R}^d$,  $p_2(y | x) = N(x' \bm{\beta}_2, 1)$ for some $\bm{\beta}_2 \in \mathbb{R}^d$,  and $p_0(x)$ is the uniform distribution in its compact domain $\mathcal{X}$.  Apparently,  $P_1, P_2 \in \mathcal{P}$.  By Le Cam method, 
\begin{align} \label{eq:39}
\inf_{\check{\mu}} \sup_{P \in \mathcal{P}} \mathbb{E}_P \left[ \left(\check{\mu} - \mu(P) \right)^2 \right] \geq \frac{\Delta}{8} \exp\left(-n \: \text{KL}(P_1 \lVert P_2) \right), \\[-5pt]
\end{align}
where $\Delta = \left(\mu(P_1) - \mu(P_2) \right)^2 = \left(\bm{a}' \bm{\beta}_1 - \bm{a}'\bm{\beta}_2 \right)^2$, and $\text{KL}(P_1 \lVert P_2)$ denotes the Kullback-Leibler divergence from $P_1$ to $P_2$. Simplifications of $\text{KL}(P_1 \lVert P_2)$ show that
\begin{align} \label{eq:40} 
\text{KL}(P_1 \lVert P_2) = \int_{\mathcal{X}} \int_{-\infty}^{\infty} p_1(x, y) \log \left(\frac{p_1(x, y)}{p_2(x, y)} \right) dy \: dx  \leq \frac{C_0}{2} \lVert \bm{\beta}_2 - \bm{\beta}_1 \rVert^2, \\[-5pt]
\end{align} 
where $C_0 = \int_{\mathcal{X}} p_0(x) \lVert x \rVert^2 d x  > 0$.  Substituting \eqref{eq:40} into the r.h.s. of \eqref{eq:39} gives
\begin{align} \label{eq:41}
\inf_{\check{\mu}} \sup_{P \in \mathcal{P}} \mathbb{E}_P \left[ \left(\check{\mu} - \mu(P) \right)^2 \right] \geq \frac{1}{8} \left(\bm{a}' \bm{\beta}_1 - \bm{a}'\bm{\beta}_2 \right)^2 \exp\left(-\frac{C_0}{2} n \lVert \bm{\beta}_2 - \bm{\beta}_1 \rVert^2 \right). \\[-5pt]
\end{align}
Setting $\bm{\beta}_2 = \bm{\beta}_1 + \frac{\bm{a}}{\sqrt{n}}$ gives
\begin{align} \label{eq:42}
\inf_{\check{\mu}} \sup_{P \in \mathcal{P}} \mathbb{E}_P \left[ \left(\check{\mu} - \mu(P) \right)^2 \right] \geq \frac{1}{8n} \exp\left(-\frac{C_0}{2} \right), \\[-5pt]
\end{align}
which concludes \eqref{eq:38} and thus Theorem 6.
\end{proof}

\bigskip

\begin{proof}[Proof of Theorem 7]
Note there exists a constant $C$ such that $\underset{t \in \mathcal{T}}{\sup} \: \text{Var} \left[\tilde{G}_n(t) \!\mid\! (X_i)_{i=1}^n \right]$ $\leq C$ a.s, by Lemma \ref{lemma9}. We first prove Theorem 7 conditional on $(X_i)_{i=1}^n$ such that
$\underset{t \in \mathcal{T}}{\sup} \: \text{Var} \left[\tilde{G}_n(t) \!\mid\! (X_i)_{i=1}^n \right] \leq C$. 

We first use the maximal inequality for Gaussian processes to bound $\sup_{t \in \mathcal{T}} \tilde{G}^2_n(t)$ by $O_p(1)$. More specifically, for the Gaussian process $\tilde{G}_n \!\mid\! (X_i)_{i=1}^n $, let $\sigma(\tilde{G}_n) \coloneqq \sup_{t \in \mathcal{T}} \sigma_n(t)$, and denote its standard deviation metric by $\rho_n(t,s) = \sigma(\tilde{G}_n(t)-\tilde{G}_n(s))$. We have shown in \eqref{eq:84} in the proof of Theorem 4 that $\rho_n(t,s) \leq c_0 \left|t-s \right|^{1/2}$ for some constant $c_0$,  suggesting
\begin{align} \label{eq:85}
N(\epsilon,\: \mathcal{T}, \: \rho_n) \leq \left(\frac{L}{\epsilon} \right)^V, \quad \text{for all} \; 0 < \epsilon < \epsilon_0, \\[-5pt]
\end{align}
holds for $\epsilon_0 = \sigma(\tilde{G}_n)$, $L > \sigma(\tilde{G}_n)$, and $V = 2$. We then invoke Proposition A.2.7 in \citet{van1996weak} to obtain that, for all $\lambda \geq (1+\sqrt{2})\sigma^2(\tilde{G}_n)/\epsilon_0$, 
\begin{align}\label{eq:86}
P\left(\sup_{t \in \mathcal{T}} \tilde{G}_n(t) \geq \lambda \right) \leq \frac{D\lambda}{\sigma^3(\tilde{G}_n)} \exp\left[-\frac{1}{2} \frac{\lambda^2}{\sigma^2(\tilde{G}_n)}\right], \\[-5pt]
\end{align}
where $D$ is some universal constant. Because $\tilde{G}_n(t) \!\mid\! (X_i)_{i=1}^n $ is a zero-mean Gaussian process, \eqref{eq:86} gives $\sup_{t \in \mathcal{T}} \tilde{G}^2_n(t)  = O_p(1)$.

We next proceed to calculate the rate of convergence for $\hat{\mu}_n(t)^{r-SI}$. Let $\eta_l \coloneqq \hat{\mu}_n(t_l)^{r-SI} - \mu(t_l) = \hat{\mu}_n(t_l) - \mu(t_l)$ for each $1 \leq l \leq T$, and let $h$ be the linear interpolation of $\{\left(t_l, \eta_l\right): \; 1 \leq l \leq T\}$, i.e.,
\begin{align}
h(t) \coloneqq \frac{t_{l+1}- t}{t_{l+1}-t_l} \eta_l + \frac{t - t_l}{t_{l+1}-t_l} \eta_{l+1}, \quad\quad \text{for} \; t \in \left[t_l, t_{l+1}\right]. \\[-5pt]
\end{align}

\noindent Then $\hat{\mu}_n(t)^{r-SI} = Q_r(\mu(t)+h(t))$, where $Q_r$ is the operator associated with the $r$-th order spline interpolation, i.e., for a general function $f$, $Q_r(f)$ is the solution to
\begin{align}
\min_{g \in \mathcal{W}_2^{r}} \int_{\mathcal{T}}\left[g^{(r)}(t)\right]^2\: dt, \quad \text{subject to} \;\; g(t_l) = f(t_l), \quad l=1, \dots, T. \\[-5pt]
\end{align} 

\noindent Since $Q_r$ is a linear operator ~\citep{devore1993constructive}, $\hat{\mu}_n^{r-SI} = Q_r(\mu+h) = Q_r(\mu) + Q_r(h)$, and we have
\begin{align}\label{eq:87}
\lVert \hat{\mu}_n^{r-SI} - \mu \rVert_{\mathcal{L}_2} \leq \lVert Q_r(\mu) - \mu \rVert_{\mathcal{L}_2} + \lVert Q_r(h) \rVert_{\mathcal{L}_2}, \\[-5pt]
\end{align}
by triangle inequality. If $\mu \in \mathcal{W}_2^{r}$, the first term on the r.h.s. in \eqref{eq:87}, which is the approximation error caused by $r$-th spline interpolation for $\mu$, can be bounded by (DeVore and Lorentz, 1993) \\[-5pt]
\begin{align}\label{eq:88}
\lVert Q_r(\mu) - \mu \rVert_{\mathcal{L}_2}^2 \lesssim T^{-2r}. \\[-5pt]
\end{align}

\noindent Following the arguments used in the proof of Theorem 2.2 in \citet{cai2011optimal}, we can bound the second term on the r.h.s. in \eqref{eq:87} by
\begin{align}
\lVert Q_r(h) \rVert_{\mathcal{L}_2}^2 \lesssim & \; \lVert h \rVert _{\mathcal{L}_2}^2 \lesssim  T^{-1} \sum_{l=1}^T \eta_l^2 =  T^{-1} \sum_{l=1}^T \left(\hat{\mu}_n(t_l) - \mu(t_l) \right)^2 \nonumber \\
= & \; T^{-1} \sum_{l=1}^T \left(\frac{1}{\sqrt{n}} \tilde{G}_n(t_l) + \frac{1}{\sqrt{n}} \tilde{r}_n(t_l) \right)^2 \quad (\text{by Theorem 2}) \nonumber \\
\lesssim & \; \frac{1}{n\:\!T} \sum_{l=1}^T \left(\tilde{G}^2_n(t_l) + \tilde{r}^2_n(t_l) \right) \leq \frac{1}{n} \sup_{t \in \mathcal{T}} \tilde{G}^2_n(t) + \frac{1}{n} \sup_{t \in \mathcal{T}} \tilde{r}^2_n(t) = O_p(1/n).  \label{eq:89} \\[-5pt]
\end{align}

\noindent Since $\underset{t \in \mathcal{T}}{\sup} \: \text{Var} \left[\tilde{G}_n(t) \!\mid\! (X_i)_{i=1}^n \right] \leq C$ a.s., the asserted claim of Theorem 7 follows from \eqref{eq:88} and \eqref{eq:89}.  
\end{proof}

\section{Proofs of Propositions and Lemmas} \label{lemma_proof}
In this section, we provide proofs of Proposition \ref{prop1}, and Lemma \ref{lemma1} -- \ref{lemma10} presented in Section \ref{lemmas}. Throughout the following proofs, $C, C_1, C_2, c_1, c_2$,  etc. will denote constants that do not depend on $n$ but may have different values in different parts of the proofs.

\begin{proof}[Proof of Proposition \ref{prop1}]
The proof consists of three steps.

\textit{Step 1:} Under the conditions stated in Theorem 1, if we additionally assume that the distribution for weights $\left(\omega_i \right)_{i=1}^n$ satisfies $\omega > 0, \; \mathbb{E} \left[\omega \right] = 1, \; \mathbb{E} \left[\omega^2 \right] \lesssim 1, \; \max_{1 \leq i \leq n} \omega_i \lesssim_{P} \log n$, then we can adapt the proof of Theorem 1 to show \eqref{eq:438}, a Bahadur representation uniform in $t \in \bm{t}$ for the weighted data $\left\{(\omega_i X_i, \; \omega_i Y_i(\bm{t}))\right\}_{i=1}^{n}$. 
\begin{align}\label{eq:438}
\hat{\bm{\beta}}^{b}_{\tau}(t)-\bm{\beta}_{\tau}(t) = -\frac{1}{n} J_{\tau}(t)^{-1} \sum_{i=1}^{n} \omega_i \:\! \psi(Y_i(t), X_i; \bm{\beta}_{\tau}(t), \tau) + r_{n}(t, \tau), \\[-5pt]
\end{align}
for $t \in \bm{t}$, where $\hat{\bm{\beta}}^{b}_{\tau}(t)$ is defined by \eqref{eq:30} and $\sup_{t \in \bm{t}}\lVert r_{n}(t, \tau) \rVert = o_p(n^{-1/2})$. We just need to replace $\left(X_i, Y_i(\bm{t})\right)$ in the proof of Theorem 1 with $\left(\omega_i X_i, \omega_i Y_i(\bm{t})\right)$ for $i=1, \dots, n$, and replace the condition $\lVert X \rVert \leq \xi$ in Assumption (A1) by $\lVert \omega X \rVert \lesssim_{P} \xi \log n$. 

Combining \eqref{eq:438} with (2.1) in Theorem 1 gives
\begin{align}\label{eq:439}
\hat{\bm{\beta}}^{b}_{\tau}(t) - \hat{\bm{\beta}}_{\tau}(t) = -\frac{1}{n} J_{\tau}(t)^{-1} \sum_{i=1}^{n} (\omega_i -1) \:\! \psi(Y_i(t), X_i; \bm{\beta}_{\tau}(t), \tau) + r^{b}_{n}(t, \tau), \\[-5pt]
\end{align}
for $t \in \bm{t}$, where $\sup_{t \in \bm{t}}\lVert r^{b}_{n}(t, \tau) \rVert = o_p(n^{-1/2})$. 

\bigskip 

\textit{Step 2:} For a given linear combination $\bm{a} \in \mathcal{S}^{d-1}$ and any $t \in \mathcal{T}$, let
\begin{align}\label{eq:440}
\mathbb{G}^{b}_n(t) \coloneqq \frac{1}{\sqrt{n}} \bm{a}' J_{\tau}(t)^{-1} \sum_{i=1}^{n} (\omega_i -1) X_i\left(\mathbf{1}\{Y_i(t) \leq X_i'\bm{\beta}_{\tau}(t)\} - \tau\right). \\[-5pt]
\end{align}
Given \eqref{eq:439}, for $\hat{\mu}_n(t)^{b-LI}$ defined by \eqref{eq:31}, we can follow the proof of Theorem 3 to show 
\begin{align}\label{eq:441}
\sup_{t \in \mathcal{T}} \left|\hat{\mu}_n(t)^{b-LI} - \hat{\mu}_n(t)^{LI} + \frac{1}{\sqrt{n}} \mathbb{G}^{b}_n(t) \right| = o_p(\frac{1}{\sqrt{n}}), \\[-5pt]
\end{align}
provided we can show the asymptotic tightness of the process $\mathbb{G}^{b}_n(\cdot)$ given by \eqref{eq:442}
\begin{align} \label{eq:442}
\sup_{v,\: s \in \mathcal{T}, \: |v-s| \leq \delta_{T(n)}} \left|\mathbb{G}^{b}_n(v) - \mathbb{G}^{b}_n(s) \right| = o_p(1), \\[-5pt]
\end{align}
where $\delta_{T(n)} = o(1)$. Assuming that $\mathbb{E} \left[\omega^4 \right] \lesssim 1$, for any $v, s \in \mathcal{T}$, we have
\begin{align}
\mathbb{E}\left[ \left|\mathbb{G}^{b}_n(v) - \mathbb{G}^{b}_n(s) \right|^4 \right] \lesssim  \frac{1}{n}|v-s| + |v-s|^2, \\[-5pt]
\end{align}
and can adapt the proof of step (ii) in Lemma \ref{lemma8} to show \eqref{eq:442}. 

\bigskip 

\textit{Step 3:} Following the proof of Theorem 4, we can show that there exists a sequence of zero-mean Gaussian processes $(\tilde{G}^{b}_n)_{n \geq 1}$ such that 
\begin{itemize}
\item[(i)] the sample paths of $\tilde{G}^{b}_n$ are a.s. continuous and the covariance functions of $\tilde{G}^{b}_n$ coincide with those of $\mathbb{G}^{b}_n$ for each $n$, i.e., $\mathbb{E}\left[\tilde{G}^{b}_n(t) \tilde{G}^{b}_n(s) \mid (X_i)_{i=1}^n \right] = \mathbb{E}\left[\vphantom{\tilde{G}^{b}_n(t)} \mathbb{G}^{b}_n(t) \mathbb{G}^{b}_n(s) \mid (X_i)_{i=1}^n \right]$ for all $t,\: s \in \mathcal{T}$; 
\item[(ii)] $\tilde{G}^{b}_n$ closely approximates $\mathbb{G}^{b}_n$ in that $\sup_{t \in \mathcal{T}} \left|\tilde{G}^{b}_n(t) - \mathbb{G}^{b}_n(t) \right| = o_p(1)$.
\end{itemize}
Given this and \eqref{eq:441}, for \eqref{eq:32} in Proposition \ref{prop1} to hold, we only need $\mathbb{E}\left[\mathbb{G}^{b}_n(t) \mathbb{G}^{b}_n(s) \mid (X_i)_{i=1}^n \right]= \mathbb{E}\left[\mathbb{G}_n(t) \mathbb{G}_n(s) \mid (X_i)_{i=1}^n \right]$ for each $n$, which is satisfied when $\mathbb{E} \left[\omega \right] = 1$ and $\mathbb{E} \left[\omega^2 \right] =2$.
\end{proof}

\bigskip

\begin{proof}[Proof of Lemma \ref{lemma1}]
We follow the proof of Lemma S.1.1 in \citet{chao2017quantile}. Recalling that $\vartheta(\bm{\beta}_{\tau}(t); t,\tau) = \mathbb{E}\left[\psi(Y(t), X; \bm{\beta}_{\tau}(t), \tau)\right]$, we have
$\partial_{\bm{\beta}} \vartheta(\bm{\beta}_{\tau}(t); t, \tau) = \mathbb{E}[XX'f_{Y(t)|X}(X'\bm{\beta}_{\tau}(t)|X)] = J_{\tau}(t)$, where $\partial_{\bm{\beta}} \vartheta(\bm{\beta}_{\tau}(t); t, \tau) = \frac{\partial}{\partial \bm{\beta}} \vartheta(\bm{\beta}; t,\tau)|_{\bm{\beta}=\bm{\beta}_{\tau}(t)}$. We have
\begin{align}
\vartheta(\bm{\beta}; t, \tau) = \vartheta(\bm{\beta}_{\tau}(t); t, \tau) + \partial_{\bm{\beta}} \vartheta(\bar{\bm{\beta}}_{\tau}(t); t, \tau) (\bm{\beta} - \bm{\beta}_{\tau}(t)), \\[-5pt]
\end{align}
where $\bar{\bm{\beta}}_{\tau}(t) = \bm{\beta} + \theta_{\bm{\beta},t,\tau}(\bm{\beta}_{\tau}(t) - \bm{\beta})$ for some $\theta_{\bm{\beta},t,\tau} \in [0,1]$. 

For any $t \in \bm{t}$, $\beta \in \mathbb{R}^d$,
\begin{align}
&\; \lVert \vartheta(\bm{\beta}; t, \tau) - \vartheta(\bm{\beta}_{\tau}(t); t, \tau) - \partial_{\bm{\beta}} \vartheta(\bm{\beta}_{\tau}(t); t, \tau) (\bm{\beta} - \bm{\beta}_{\tau}(t)) \rVert \\
= &\; \sup_{\bm{u} \in \mathcal{S}^{d-1}} \left| \bm{u}'\left[ \vartheta(\bm{\beta}; t, \tau) - \vartheta(\bm{\beta}_{\tau}(t); t, \tau) - \partial_{\bm{\beta}} \vartheta(\bm{\beta}_{\tau}(t); t, \tau) (\bm{\beta} - \bm{\beta}_{\tau}(t)) \right] \right| \\
= &\; \sup_{\bm{u} \in \mathcal{S}^{d-1}} \left| \bm{u}'\left[ \left(\partial_{\bm{\beta}} \vartheta(\bar{\bm{\beta}}_{\tau}(t); t, \tau) - \partial_{\bm{\beta}} \vartheta(\bm{\beta}_{\tau}(t); t, \tau)\right) (\bm{\beta} - \bm{\beta}_{\tau}(t)) \right] \right| \\
= &\; \sup_{\bm{u} \in \mathcal{S}^{d-1}} \left| \bm{u}'\left(\mathbb{E}[XX'f_{Y(t)|X}(X'\bar{\bm{\beta}}_{\tau}(t)|X)] - \mathbb{E}[XX'f_{Y(t)|X}(X'\bm{\beta}_{\tau}(t)|X)] \right) (\bm{\beta} - \bm{\beta}_{\tau}(t)) \right| \\
= &\; \sup_{\bm{u} \in \mathcal{S}^{d-1}} \left| \mathbb{E}\left[(\bm{u}'X) X'(\bm{\beta} - \bm{\beta}_{\tau}(t)) (f_{Y(t)|X}(X'\bar{\bm{\beta}}_{\tau}(t)|X) - f_{Y(t)|X}(X'\bm{\beta}_{\tau}(t)|X))\right] \right| \\
\leq & \;\; \overline{f'} \sup_{\bm{u} \in \mathcal{S}^{d-1}} \mathbb{E} \left[\left| \bm{u}'X \right| \left| X'(\bm{\beta}-\bm{\beta}_{\tau}(t)) \right| \left| X'(\bar{\bm{\beta}}_{\tau}(t)-\bm{\beta}_{\tau}(t)) \right| \right] \quad (\text{by Assumption (A2)}) \\
\leq & \;\; \overline{f'} \xi \; \mathbb{E} \left[\left| X'(\bm{\beta}-\bm{\beta}_{\tau}(t)) \right| \left| X'(\bar{\bm{\beta}}_{\tau}(t)-\bm{\beta}_{\tau}(t)) \right| \right] \quad (\text{by Assumption (A1)}) \\
\leq & \;\; \overline{f'} \xi \; \lVert \bm{\beta}-\bm{\beta}_{\tau}(t) \rVert^2 \; \sup_{\bm{u} \in \mathcal{S}^{d-1}} \mathbb{E}[\bm{u}'XX'\bm{u}] \quad (\text{by Cauchy–Schwarz inequality}) \\
= & \;\; \overline{f'} \xi \; \lVert \bm{\beta}-\bm{\beta}_{\tau}(t) \rVert^2 \; \lambda_{\max}(\mathbb{E}[XX']). \\[-5pt]
\end{align}
This gives Lemma \ref{lemma1} by taking the supremum over $\lVert \bm{\beta}-\bm{\beta}_{\tau}(t) \rVert \leq \delta$ and $t \in \bm{t}$.
\end{proof}

\bigskip

\begin{proof}[Proof of Lemma \ref{lemma2}]
We follow the proof of Lemma S.1.2 in \citet{chao2017quantile}. We first show that for any $a > 0$, 
\begin{align}\label{eq:401}
\!\!\!\!\!\! \left\{\sup_{t \in \bm{t}} \lVert \bm{\hat{\beta}}_{\tau}(t) - \bm{\beta}_{\tau}(t) \rVert \leq a s_{n,1} \right\} \supseteq \left\{\inf_{t \in \bm{t}} \inf_{\bm{\delta} =1} \bm{\delta}'\mathbb{P}_n\left[\psi(Y_i(t),X_i; \bm{\beta}_{\tau}(t)+ a s_{n,1} \bm{\delta},\tau) \right] > 0 \right\}. \\[-5pt]
\end{align}

\noindent To see this, define $\bm{\delta} \coloneqq (\bm{\hat{\beta}}_{\tau}(t) - \bm{\beta}_{\tau}(t)) / \lVert \bm{\hat{\beta}}_{\tau}(t) - \bm{\beta}_{\tau}(t) \rVert$. Observe that $f: \bm{\beta} \mapsto \mathbb{P}_n \rho_{\tau}(Y_i(t) - X_i'\bm{\beta})$ is convex for any $t$, and $\mathbb{P}_n \left[\psi(Y_i(t),X_i; \bm{\beta},\tau) \right]$ is a subgradient of $f$ at the point $\bm{\beta}$. By definition of the subgradient, we have for any $t \in \bm{t}$, $\zeta_n > 0$,
\begin{align}
\begin{split}
& \; \mathbb{P}_n \left[\rho_{\tau}(Y_i(t) - X_i'\bm{\hat{\beta}}_{\tau}(t))\right] \\ 
\geq & \; \mathbb{P}_n \left[\rho_{\tau}(Y_i(t) - X_i'(\bm{\beta}_{\tau}(t)+\zeta_n \bm{\delta}))\right] + (\bm{\hat{\beta}}_{\tau}(t) - \bm{\beta}_{\tau}(t) - \zeta_n \bm{\delta})'\mathbb{P}_n \left[\psi(Y_i(t),X_i; \bm{\beta}_{\tau}(t)+\zeta_n \bm{\delta},\tau) \right] \\
= & \; \mathbb{P}_n \left[\rho_{\tau}(Y_i(t) - X_i'(\bm{\beta}_{\tau}(t) + \zeta_n \bm{\delta}))\right] + (\lVert \bm{\hat{\beta}}_{\tau}(t) - \bm{\beta}_{\tau}(t) \rVert - \zeta_n) \bm{\delta}' \mathbb{P}_n \left[\psi(Y_i(t),X_i; \bm{\beta}_{\tau}(t)+\zeta_n \bm{\delta},\tau) \right]. \\[-5pt]
\end{split}
\end{align}

\noindent Recalling that $\bm{\hat{\beta}}_{\tau}(t)$ is a minimizer of $\mathbb{P}_n\left[\rho_{\tau}(Y_i(t) - X_i'\bm{\beta})\right]$, the inequality above leads to
\begin{align}
& \; (\lVert \bm{\hat{\beta}}_{\tau}(t) - \bm{\beta}_{\tau}(t) \rVert - \zeta_n) \bm{\delta}' \mathbb{P}_n \left[\psi(Y_i(t),X_i; \bm{\beta}_{\tau}(t)+\zeta_n \bm{\delta},\tau) \right] \\
\leq & \; \mathbb{P}_n \left[\rho_{\tau}(Y_i(t) - X_i'\bm{\hat{\beta}}_{\tau}(t))\right] - \mathbb{P}_n \left[\rho_{\tau}(Y_i(t) - X_i'(\bm{\beta}_{\tau}(t)+\zeta_n \bm{\delta}))\right] \leq 0. \\[-5pt] 
\end{align}

\noindent Setting $\zeta_n = a s_{n,1}$, we have that $\inf_{t \in \bm{t}} \inf_{\bm{\delta} =1} \bm{\delta}' \mathbb{P}_n \left[\psi(Y_i(t),X_i; \bm{\beta}_{\tau}(t)+ a s_{n,1} \bm{\delta},\tau) \right] > 0$, so $\sup_{t \in \bm{t}} \lVert \bm{\hat{\beta}}_{\tau}(t) - \bm{\beta}_{\tau}(t) \rVert \leq a s_{n,1}$, 
which yields \eqref{eq:401}.

Under Assumptions (A1)-(A3), by Lemma \ref{lemma1}, we also have that for any $t$,
\begin{align}
&\sup_{\bm{\delta}\in \mathcal{S}^{d-1}} \left| \mathbb{E}\left[\bm{\delta}'\left\{\psi(Y(t),X; \bm{\beta}, \tau) \! - \! \psi(Y(t),X; \bm{\beta}_{\tau}(t), \tau) \! - \! XX'f_{Y(t)|X}(X'\bm{\beta}_{\tau}(t)|X) (\bm{\beta} \! - \! \bm{\beta}_{\tau}(t)) \right\} \right] \right| \nonumber \\
= & \sup_{\bm{\delta}\in \mathcal{S}^{d-1}} \left|\bm{\delta}'\left\{\vartheta(\bm{\beta}; t,\tau) - \vartheta(\bm{\beta}_{\tau}(t); t, \tau) - J_{\tau}(t)(\bm{\beta} - \bm{\beta}_{\tau}(t)) \right\} \right| \nonumber \\
= & \left\lVert \vartheta(\bm{\beta}; t,\tau) - \vartheta(\bm{\beta}_{\tau}(t); t, \tau) - J_{\tau}(t)(\bm{\beta} - \bm{\beta}_{\tau}(t)) \right\rVert \leq \lambda_{\max}(\mathbb{E}[XX']) \overline{f'} \xi \; \lVert \bm{\beta}-\bm{\beta}_{\tau}(t) \rVert ^2. \label{eq:402} \\[-5pt]
\end{align}

\noindent Given $\mathbb{E}[\psi(Y(t),X; \bm{\beta}_{\tau}(t), \tau)] = \bm{0}$, for any $t \in \bm{t}$,  $\bm{\delta} \in \mathcal{S}^{d-1}$, \eqref{eq:402} leads to
\begin{align}
& \; - \bm{\delta}'\; \mathbb{E}\left[\psi(Y(t),X; \bm{\beta}_{\tau}(t) + a s_{n,1}\bm{\delta}, \tau) - XX'f_{Y(t)|X}(X'\bm{\beta}_{\tau}(t)|X)\; a s_{n,1}\bm{\delta} \right] \\
\leq & \; \lambda_{\max}(\mathbb{E}[XX']) \overline{f'} \xi \; a^2 s_{n,1}^2, \\[-5pt]
\end{align}
so we have
\begin{align}\label{eq:403}
\bm{\delta}' \mathbb{E}[\psi(Y(t),X; \bm{\beta}_{\tau}(t) + a s_{n,1}\bm{\delta}, \tau)] \geq a s_{n,1} \bm{\delta}'J_{\tau}(t)\bm{\delta} \! - \! \lambda_{\max}(\mathbb{E}[XX']) \overline{f'} \xi a^2 s_{n,1}^2. \\[-5pt]
\end{align}

\noindent Therefore, for any $t \in \bm{t}$, arbitrary $\bm{\delta} \in \mathcal{S}^{d-1}$,
\begin{align}\label{eq:404}
& \; \bm{\delta}'\; \mathbb{P}_n\left[\psi(Y_i(t), X_i; \bm{\beta}_{\tau}(t) + a s_{n,1}\bm{\delta}, \tau) \right] \\
\geq & \; - s_{n,1} + \bm{\delta}'\; \mathbb{E}\left[\psi(Y(t),X; \bm{\beta}_{\tau}(t) + a s_{n,1}\bm{\delta}, \tau) \right]  \quad (\text{by definition of} \;\;s_{n,1}) \\
\geq & \; - s_{n,1} + a s_{n,1}\;\bm{\delta}'J_{\tau}(t)\bm{\delta} -  \lambda_{\max}(\mathbb{E}[XX']) \overline{f'} \xi \; a^2 s_{n,1}^2  \quad (\text{by} \;\; \eqref{eq:403}) \\
\geq & \; - s_{n,1} + a s_{n,1} \inf_{t \in \bm{t}}\lambda_{\min}(J_{\tau}(t)) -  \lambda_{\max}(\mathbb{E}[XX']) \overline{f'} \xi \; a^2 s_{n,1}^2. \\[-5pt]
\end{align}

\noindent Setting $a = 2\upsilon / (\inf_{t \in \bm{t}} \lambda_{\min}(J_{\tau}(t)))$ for some $\upsilon > 0$, the expression in the last line of \eqref{eq:404} is positive when 
\begin{align}
s_{n,1} < \frac{(2\upsilon-1) \inf_{t \in \bm{t}}\lambda_{\min}^2(J_{\tau}(t))} {4\upsilon^2 \xi \overline{f'} \lambda_{\max}(\mathbb{E}[XX'])}. \\[-5pt]
\end{align}

\noindent For $\upsilon > 1$, $(2\upsilon - 1)/\upsilon^2 > 1/\upsilon$ holds. Substitute $a = 2\upsilon / (\inf_{t \in \bm{t}} \lambda_{\min}(J_{\tau}(t)))$ for $\upsilon > 1$,  we have 
\begin{align}
\begin{split}
& \; \left\{ s_{n,1} < \frac{\inf_{t \in \bm{t}}\lambda_{\min}^2(J_{\tau}(t))} {4\upsilon \xi \overline{f'} \lambda_{max}(\mathbb{E}[XX'])} \right\}  \\
\subseteq & \; \left\{\inf_{t \in \bm{t}} \inf_{\bm{\delta} =1} \bm{\delta}'\mathbb{P}_n [\psi(Y_i(t),X_i; \; \bm{\beta}_{\tau}(t) +  \frac{2\upsilon}{\inf_{t \in \bm{t}} \lambda_{\min}(J_{\tau}(t))} s_{n,1} \bm{\delta},\tau) ] > 0 \right\}  \quad (\text{by} \; \eqref{eq:404}) \\
\subseteq & \; \left\{\sup_{t \in \bm{t}} \lVert \bm{\hat{\beta}}_{\tau}(t) - \bm{\beta}_{\tau}(t) \rVert \leq \frac{2\upsilon}{\inf_{t \in \bm{t}} \lambda_{\min}(J_{\tau}(t))} s_{n,1} \right\}. \quad (\text{by} \; \eqref{eq:401})   \\[-5pt]
\end{split}
\end{align}
This completes the proof.
\end{proof}

\bigskip

\begin{proof}[Proof of Lemma \ref{lemma3}]
Let $Z=(X, \bm{Y})$ denote the functional data taking values in $\mathcal{Z}$, where $\bm{Y}$ is the $T \times 1$ vector with the $t^{th}$ element equal to $Y(t)$, as defined in Section \ref{lemmas}. Let $Z_t=(X,Y(t))$ denote the data at location $t$, which take values in $\mathcal{Z}_t$. Define the following classes of functions: 
\begin{align} \label{eq:405}
\mathcal{W} \coloneqq \left\{ (X, \bm{Y}) \mapsto \bm{u}'X \; | \;\bm{u} \in \mathcal{S}^{d-1} \right\}. \\[-5pt]
\end{align}
\begin{align} \label{eq:406}
\mathcal{F} \coloneqq \left\{ (X, \bm{Y}) \mapsto \bm{1}\{Y(t) \leq X'\bm{\beta}\} \; | \;\bm{\beta} \in \mathbb{R}^{d},\;  1 \leq t \leq T \right\}. \\[-5pt]
\end{align}

Our first step is to bound the VC index of $\mathcal{F}$. For any $t \in \{1,\dots, T\}$, define
\begin{align} 
\mathcal{F}_t \coloneqq \left\{ (X,Y(t)) \mapsto \bm{1}\{Y(t)\leq X'\bm{\beta}\} \; | \;\bm{\beta} \in \mathbb{R}^{d} \right\}. \\[-5pt]
\end{align}

\noindent Note that for any $t$, the class of functions $\mathcal{V}_t \coloneqq \left\{ (X,Y(t)) \mapsto X'\bm{\beta} - Y(t) | \;\bm{\beta} \in \mathbb{R}^{d} \right\}$ has a VC index bounded by $d+2$, by Lemma 2.6.15 in \citet{van1996weak}. Applying Lemma 2.6.18 (iii) in \citet{van1996weak}, we can show that $V(\mathcal{F}_t)$ is also bounded by $d+2$, and $V(\mathcal{F}_t)$ is constant cross $t$. Denote the VC index of $V(\mathcal{F}_t)$ by $\nu$. We next prove that $V(\mathcal{F})\leq C\:\nu \log T$ for some constant $C$ that does not depend on $T$ or $\nu$, which is equivalent to show that the subgraphs of functions $f: \mathcal{Z} \mapsto \mathbb{R}$ in $\mathcal{F}$ cannot shatter any collection of $n \geq C\:\nu \log T$ points: $\left((x_1, y_1), \eta_1\right), \dots, \left((x_n, y_n), \eta_n\right)$ in $\mathcal{Z} \times \mathbb{R}$; see Section 2.6.2 in \citet{van1996weak} for the definition of subgraphs of a function.

We first show that, for any given subset $\mathcal{I} \subset \{1,\dots, n\}$, (i) the subgraphs of functions $f \in \mathcal{F}$ can pick out $\left\{\left((x_{i},y_{i}), \eta_{i}\right): \; i\in \mathcal{I} \right\}$  is equivalent to (ii) for some $t \in \{1,\dots, T\}$, the subgraphs of functions $f_t \in \mathcal{F}_t$ can pick out $\left\{\left((x_{i},y_{i}(t)), \eta_{i}\right): \; i\in \mathcal{I} \right\}$.  To see this, note if (i) holds, i.e., there exists $f \in \mathcal{F}$ whose subgraph can pick out $\left\{\left((x_{i},y_{i}), \eta_{i}\right): \; i\in \mathcal{I} \right\}$, we have
\begin{align} 
\eta_{i} < &\; f\left((x_{i},y_{i})\right) \quad\quad\quad \text{for} \quad i \in \mathcal{I}; \nonumber \\
\eta_{i} \geq &\; f\left((x_{i},y_{i})\right) \quad\quad\quad \text{for} \quad i \notin \mathcal{I}. \label{eq:407} \\[-5pt]
\end{align}

\noindent But $f \in \mathcal{F}$ indicates that $f\left((x,y)\right) = f_t\left((x,y(t))\right) $ for some $t \in \{1,\dots,T\}$, $f_t \in \mathcal{F}_t$, so \eqref{eq:407} leads to
\begin{align} 
\eta_{i} < &\; f_t\left((x_{i},y_{i}(t))\right) \quad\quad\quad \text{for} \quad i \in \mathcal{I}; \\
\eta_{i} \geq &\; f_t\left((x_{i},y_{i}(t))\right) \quad\quad\quad \text{for} \quad i \notin \mathcal{I}, \\[-5pt]
\end{align}
which shows that (ii) holds. In a similar manner, it is easy to show that (ii) also leads to (i). Therefore, (i) and (ii) are equivalent.

By Sauer's Lemma on page 86 in \citet{van1996weak}, given $V(\mathcal{F}_t) = \nu$ for each $t$, for any $n \geq \nu$, the subgraphs of $f_t \in \mathcal{F}_t$ can pick out at most $\left(\frac{n e}{\nu - 1}\right)^{\nu-1}$ subsets from any collection of $n$ points $\left((x_1, y_1(t)), \eta_1\right), \dots, \left((x_n, y_n(t)), \eta_n\right)$ in $\mathcal{Z}_t \times \mathbb{R}$, and equivalently, can pick out at most $\left(\frac{n e}{\nu - 1}\right)^{\nu-1}$ subsets from any collection of $n$ points $\left((x_1, y_1), \eta_1\right), \dots, \left((x_n, y_n), \eta_n\right)$ in $\mathcal{Z} \times \mathbb{R}$. By symmetry, the subgraphs of $f \in \mathcal{F}$ can pick out at most $ T \left(\frac{n e}{\nu - 1}\right)^{\nu-1}$ subsets from these $n$ points in $\mathcal{Z} \times \mathbb{R}$. For some constant $C$ that is independent of $\nu$ or $T$ and satisfy $C \log 2 > 1$,  we have that for $n = C \nu \log T$, 
\begin{align} 
& \; \log T + (\nu-1) (\log n + \log\frac{e}{\nu-1}) \\
= & \; \log T + (\nu-1) \left(\log C + \log(\nu) + \log(\log T) + \log(\frac{e}{\nu-1}) \right) \\
= & \; \log T + (\nu-1) \left(\log C + \log(\log T) + 1 + \log(\frac{\nu}{\nu-1}) \right) \\
< & \; C \nu\;\log 2 \; \log T = n \log 2,  \quad \text{for all sufficiently large } T.   \\[-5pt]
\end{align}
Therefore, \\[-5pt]
\begin{align}
T \left(\frac{n e}{\nu - 1}\right)^{\nu-1} < 2^n, \quad \text{for all sufficiently large } T.  \\[-5pt] 
\end{align}

\noindent This means that for $n = C \nu \log T$, the number of subsets that can be picked out by the subgraphs of $f \in \mathcal{F}$ is strictly smaller than $2^n$, suggesting that the subgraphs of $f \in \mathcal{F}$ cannot shatter these $n$ points. Given that the $n$ points are arbitrarily chosen in $\mathcal{Z} \times \mathbb{R}$, we have proved that $V(\mathcal{F}) \leq C \nu \log T$. 

Note that any $g \in \mathcal{G}_1$ can be written as $g = w \cdot (f- \upsilon)$ for $w \in \mathcal{W}$, $f \in \mathcal{F}$, and $\upsilon \coloneqq \left\{ (X, \bm{Y}) \mapsto \tau \right\}$ where $\tau$ is a constant. We have $V(\upsilon) = O(1)$. Given this and $V(\mathcal{W}) \leq d+1$ (by Lemma 2.6.15 in \citet{van1996weak}) and $V(\mathcal{F})\leq C (d+2) \log T$ for some constant $C$, the first claim of Lemma \ref{lemma3} follows from Lemma 24 in \citet{belloni2019conditional} and Theorem 2.6.7 in \citet{van1996weak}. 

Note that any $g \in \mathcal{G}_2(\delta)$ can be written as $g = w \cdot (f_1 - f_2)$ for $w \in \mathcal{W}$, $f_1 \in \mathcal{F}$ and $f_2 \in \mathcal{F}$, so the second claim of Lemma \ref{lemma3} also follows from Lemma 24 in \citet{belloni2019conditional} and Theorem 2.6.7 in \citet{van1996weak}.
\end{proof} 

\bigskip

\begin{proof}[Proof of Lemma \ref{lemma4}]
We need the following results on empirical processes theory, which were stated in S.2.1 in \citet{chao2017quantile}. 

(i). Denote by $\mathcal{G}$ a class of functions that satisfy $|f(x)| \leq F(x) \leq U$ 
for every $f \in \mathcal{G}$ and let $\sigma^2 \geq \sup_{f\in \mathcal{G}} P f^2$. Additionally, let for some $A > 0, V > 0$ and all $\epsilon > 0$,
\begin{align}
N(\epsilon; \; \mathcal{G}, L_2(\mathbb{P}_n)) \leq \left(\frac{A \lVert F \rVert_{L^2(\mathbb{P}_n)}}{\epsilon} \right)^V. \\[-5pt]
\end{align} 

\noindent Note that if $\mathcal{G}$ is a VC-class, then $V$ is the VC-index of the set of subgraphs of functions in $\mathcal{G}$. This yields 
\begin{align}\label{eq:411}
\mathbb{E}\lVert \mathbb{P}_n - P \rVert_{\mathcal{G}} \leq c_0 \left[\sigma \left(\frac{V}{n} \log \frac{A \lVert F \rVert_{L^2(P)}}{\sigma} \right)^{1/2} + \frac{V U}{n} \log \frac{A \lVert F \rVert_{L^2(P)}}{\sigma} \right], \\[-5pt]
\end{align}
for a universal constant $c_0 > 0$ provided that $1 \geq \sigma^2 \geq \text{const} \times n^{-1}$. 

(ii). The second inequality (a refined version of Talagrand's concentration inequality) states that for any countable class of measurable functions $\mathcal{F}$ with elements mapping into $[-M, M]$,
\begin{align}\label{eq:412}
\!\!\!\! P\left(\lVert \mathbb{P}_n \!-\! P \rVert_{\mathcal{F}} \geq 2\mathbb{E}\lVert \mathbb{P}_n \!-\! P \rVert_{\mathcal{F}} + c_1 n^{-1/2} \left(\sup_{f \in \mathcal{F}} P f^2 \right)^{1/2} \! \sqrt{\upsilon} + n^{-1} c_2 M \upsilon \right) \leq e^{-\upsilon}, \\[-5pt] 
\end{align}
for all $\upsilon > 0$ and universal constants $c_1, c_2 > 0$.

\bigskip

By the first part of Lemma \ref{lemma3}, 
\begin{align}
N(\epsilon;\; \mathcal{G}_1, L^2(\mathbb{P}_n)) \leq \left(\frac{A \lVert F_1 \rVert_{L^2(\mathbb{P}_n)}}{\epsilon}\right)^{\nu_1(T)}, \\[-5pt]
\end{align}
where $A$ is some constant, $F_1$ is an envelope funciton of $\mathcal{G}_1$, and $\nu_1(T) = O(\log T)$. For each $f \in \mathcal{G}_1$, we also have $\mathbb{E}[f^2] \leq \sup_{\bm{u} \in \mathcal{S}^{d-1}} \bm{u}' \mathbb{E}[XX'] \bm{u} = \lambda_{\max}(\mathbb{E}[XX'])$.

In \eqref{eq:411}, let $\mathcal{G} = \mathcal{G}_1, \; V = \nu_1(T), \; F = F_1 \equiv \xi, \; \sigma^2 = \lambda_{\max}(\mathbb{E}[XX'])$ (or a multiple of $\lambda_{\max}(\mathbb{E}[XX'])$ such that $1 \geq \sigma^2 \geq \text{const} \times n^{-1}$). Applying (i) gives
\begin{align}\label{eq:413}
\mathbb{E}\lVert \mathbb{P}_n - P \rVert_{\mathcal{G}_1} \leq c_0 \left[\sigma \left(\frac{\nu_1(T)}{n} \log \frac{A \xi}{\sigma} \right)^{1/2} + \frac{\nu_1(T)\:\xi}{n} \log \frac{A \xi}{\sigma} \right]. \\[-5pt]
\end{align}

In \eqref{eq:412}, let $\mathcal{F} = \mathcal{G}_1, \; M = \xi, \; \upsilon = \kappa_n$, then we apply (ii) to obtain
\begin{align}\label{eq:414}
\!\!\!\! P\left(\lVert \mathbb{P}_n \!-\! P \rVert_{\mathcal{G}_1} \geq 2\mathbb{E}\lVert \mathbb{P}_n \!-\! P \rVert_{\mathcal{G}_1} \!+\! c_1 \sqrt{\sup_{f \in \mathcal{G}_1} P f^2} \left(\frac{\kappa_n}{n}\right)^{1/2} \!+\! c_2 \xi \frac{\kappa_n}{n}\right) \leq e^{-\kappa_n}. \\[-5pt]
\end{align}

Combining \eqref{eq:413} and \eqref{eq:414}, for some constant $C$, we have
\begin{align}
P\left(\lVert \mathbb{P}_n - P \rVert_{\mathcal{G}_1} \geq C\left[\left(\frac{\log T}{n}\right)^{1/2} + \frac{\log T}{n} + \left(\frac{\kappa_n}{n}\right)^{1/2} + \frac{\kappa_n}{n}\right] \right) \leq e^{-\kappa_n}, \\[-5pt]
\end{align}
which gives the first inequality in Lemma \ref{lemma4}.

By the second part of Lemma \ref{lemma3}, 
\begin{align}
N(\epsilon;\; \mathcal{G}_2(\delta_n), L^2(\mathbb{P}_n)) \leq \left(\frac{A \lVert F_2 \rVert_{L^2(\mathbb{P}_n)}}{\epsilon}\right)^{\nu_2(T)}, \\[-5pt]
\end{align}
where $A$ is some constant, $F_2$ is an envelope funciton of $\mathcal{G}_2(\delta_n)$, and $\nu_2(T) = O(\log T)$. For each $f \in \mathcal{G}_2(\delta_n)$, by Assumption (A2), we also have
\begin{align}
\mathbb{E}[f^2]  \leq & \; \sup_{\bm{u} \in \mathcal{S}^{d-1}} \sup_{t \in \bm{t}} \sup_{\lVert \bm{\beta}_1(t)-\bm{\beta}_2(t)\rVert \leq \delta_n} \mathbb{E}\left[(\bm{u}'X)^2 \mathbf{1}\{|Y(t)-X'\bm{\beta}_1(t)| \leq |X'(\bm{\beta}_1(t)- \bm{\beta}_2(t))|\} \right] \\
\leq & \; \sup_{\bm{u} \in \mathcal{S}^{d-1}} \sup_{t \in \bm{t}} \sup_{\bm{\beta}_1(t) \in \mathbb{R}^{d}} \mathbb{E} \left[(\bm{u}'X)^2 \mathbf{1}\{|Y(t)-X'\bm{\beta}_1(t)| \leq \xi \delta_n\} \right] \\
= & \; \sup_{\bm{u} \in \mathcal{S}^{d-1}} \sup_{t \in \bm{t}} \sup_{\bm{\beta}_1(t) \in \mathbb{R}^{d}} \mathbb{E} \left[\mathbb{E}\left[(\bm{u}'X)^2 \mathbf{1}\{|Y(t)-X'\bm{\beta}_1(t)| \leq \xi \delta_n\}|X \right] \right] \\
= & \; \sup_{\bm{u} \in \mathcal{S}^{d-1}} \sup_{t \in \bm{t}} \sup_{\bm{\beta}_1(t) \in \mathbb{R}^{d}} \mathbb{E} \left[(\bm{u}'X)^2 \mathbb{E}\left[\mathbf{1}\{|Y(t)-X'\bm{\beta}_1(t)| \leq \xi \delta_n\}|X \right] \right] \\
\leq & \; \sup_{\bm{u} \in \mathcal{S}^{d-1}} \mathbb{E} \left[(\bm{u}'X)^2 2 \overline{f} \xi \delta_n \right] \leq 2\lambda_{\max}(\mathbb{E}[XX']) \overline{f} \xi \delta_n = c_3 \delta_n. \\[-5pt]
\end{align}

For $\delta_n \downarrow 0 $ such that $\delta_n \gg n^{-1}$, let $\sigma^2 = c_3 \delta_n$ (or a multiple of $c_3 \delta_n$ such that $1 \geq \sigma^2 \geq \text{const} \times n^{-1}$) in \eqref{eq:411}. Let $\mathcal{G} = \mathcal{G}_2(\delta_n), \; V = \nu_2(T)$ and $F = F_2 \equiv \xi$.  Applying (i) gives \\[-5pt]
\begin{align}\label{eq:415}
\mathbb{E}\lVert \mathbb{P}_n - P \rVert_{\mathcal{G}_2(\delta_n)}  \leq &\;  c_0 \left[\sigma \left(\frac{\nu_2(T)}{n} \log \frac{A\:\xi}{\sigma} \right)^{1/2} + \frac{\nu_2(T)\:\xi}{n} \log \frac{A\:\xi}{\sigma} \right] \\
= &\; c_0 \left[c_3^{1/2} \delta_n^{1/2} \left(\frac{\nu_2(T)}{2n} \log (\delta_n^{-1})\right)^{1/2} + \frac{\nu_2(T)\:\xi}{2n} \log (\delta_n^{-1}) \right] \\
= &\; c_4\: \delta_n^{1/2} \left(\frac{\log T}{n} \log (\delta_n^{-1})\right)^{1/2} + c_5\: \frac{\log T}{n} \log (\delta_n^{-1}) \\
\leq &\; c_4\: \delta_n^{1/2} \left(\frac{\log T}{n} \log n\right)^{1/2} + c_5\: \left(\frac{\log T}{n} \log n \right).   \\[-5pt]
\end{align}

In \eqref{eq:412}, let $\mathcal{F} = \mathcal{G}_2(\delta_n), \; M = \xi, \; \upsilon = \kappa_n$, then we apply (ii) to obtain
\begin{align}\label{eq:416}
\!\!\!\! P\left(\lVert \mathbb{P}_n \!-\! P\rVert_{\mathcal{G}_2(\delta_n)} \geq 2\mathbb{E}\lVert \mathbb{P}_n \!-\! P\rVert_{\mathcal{G}_2(\delta_n)} \!+\! c_1\sqrt{\sup_{f\in \mathcal{G}_2(\delta_n)} P f^2} \left(\frac{\kappa_n}{n}\right)^{1/2} \!+\! c_2\xi \frac{\kappa_n}{n}\right) \leq e^{-\kappa_n}. \\[-5pt]
\end{align}

Combining \eqref{eq:415} and \eqref{eq:416}, for some constant $C$, we have
\begin{align}
P\left(\lVert \mathbb{P}_n \!-\! P \rVert_{\mathcal{G}_2(\delta_n)} \geq C\left[\delta_n^{1/2} \left(\frac{\log T}{n} \log n \right)^{1/2} \!\! + \! \frac{\log T}{n} \log n + \delta_n^{1/2} \left(\frac{\kappa_n}{n}\right)^{1/2} + \frac{\kappa_n}{n}\right] \right) \leq e^{-\kappa_n}, \\[-5pt]
\end{align}
which gives the second inequality in Lemma \ref{lemma4}.
\end{proof}

\bigskip

\begin{proof}[Proof of Lemma \ref{lemma6}] 
Recall that $J_{\tau}(t) = \mathbb{E}[XX'f_{Y(t)|X}(X'\bm{\beta}_{\tau}(t)|X)]$ for each $t$. By Assumptions (A2), (A4), (A5),
\begin{align}
C'_1 = & \; \sup_{x\in \mathcal{X}, t\in \mathcal{T}} \left|\frac{\partial}{\partial t}f_{Y(t)|X}(x'\bm{\beta}_{\tau}(t)|x)\right| \\
= & \; \sup_{x\in \mathcal{X}, t\in \mathcal{T}} \left|\frac{\partial}{\partial y}f_{Y(t)|X}(y|x)|_{y=x'\bm{\beta}_{\tau}(t)} \: \frac{d}{dt} x'\bm{\beta}_{\tau}(t) \: + \frac{\partial}{\partial t}f_{Y(t)|X}(y|x)|_{y=x'\bm{\beta}_{\tau}(t)} \right| \\
\leq & \; \sup_{x\in \mathcal{X}, y\in \mathbb{R}, t\in \mathcal{T}} \left|\frac{\partial}{\partial y}f_{Y(t)|X}(y|x)\right| \sup_{x\in \mathcal{X}, t\in \mathcal{T}} \left|x' \frac{d}{dt} \bm{\beta}_{\tau}(t) \right| + \! \sup_{x\in \mathcal{X}, y\in \mathbb{R}, t\in \mathcal{T}} \left|\frac{\partial}{\partial t}f_{Y(t)|X}(y|x) \right| < \infty. \\[-5pt]
\end{align}

\noindent By a Taylor expansion we have that uniformly over $t, s \in \mathcal{T}$, $x \in \mathcal{X}$,
\begin{align} 
\left|f_{Y(t)|X}(x'\bm{\beta}_{\tau}(t)|x) - f_{Y(s)|X}(x'\bm{\beta}_{\tau}(s)|x) \right| \leq C'_1 \left|t-s\right|. \\[-5pt]
\end{align}

\noindent Therefore, 
\begin{align}
J_{\tau}(t) - J_{\tau}(s) \preccurlyeq C'_1 \left|t-s\right| \mathbb{E}[XX'], \;\; \text{and} \;\; J_{\tau}(s) - J_{\tau}(t) \preccurlyeq C'_1 \left|t-s\right| \mathbb{E}[XX'], \\[-5pt]
\end{align}
where the inequalities $\preccurlyeq$ are in the semi-definite positive sense. Using the matrix identity $A^{-1} - B^{-1} = B^{-1}(B-A)A^{-1}$,
\begin{align}
& \; \lVert J_{\tau}(t)^{-1} - J_{\tau}(s)^{-1} \rVert \\
= & \; \lVert J_{\tau}(s)^{-1} \left(J_{\tau}(s) - J_{\tau}(t) \right) J_{\tau}(t)^{-1} \rVert \leq \lVert J_{\tau}(s)^{-1} \rVert \lVert J_{\tau}(s) - J_{\tau}(t) \rVert \lVert J_{\tau}(t)^{-1} \rVert \\
\leq & \; \lambda_{\max}(J_{\tau}(s)^{-1}) \: C'_1 |t-s| \lambda_{\max}(\mathbb{E}[XX']) \: \lambda_{\max}(J_{\tau}(t)^{-1}) \\
\leq & \; C'_1 \left(\frac{1}{\inf_{t\in \mathcal{T}}\lambda_{\min}(J_{\tau}(t))}\right)^{2} \: \lambda_{\max}(\mathbb{E}[XX']) \: |t-s|, \quad (\text{by Assumption (A1), (A3)})  \\[-5pt]
\end{align}
which completes the proof.
\end{proof}

\bigskip

\begin{proof}[Proof of Lemma \ref{lemma8}] 
To prove the theorem, we adapt the proof of Theorem 2.1 in \citet{chao2017quantile} and Theorem 4.2 in \citet{volgushev2019distributed}. Let $V_i(t) \coloneqq \bm{a}' J_{\tau}(t)^{-1} X_i\left(\mathbf{1}\{Y_i(t) \leq X_i'\bm{\beta}_{\tau}(t)\} - \tau\right)$, so we have that $\mathbb{G}_n(t) = n^{-1/2} \sum_{i=1}^n V_i(t)$. The proof consists of the following two steps, 

(i). Finite-dimensional convergence. By Cram\'er-Wold theorem, it suffices to show that for an arbitrary, finite set of $\{t_1, \dots, t_L \}$ and $\{\zeta_1, \dots, \zeta_L\} \in \mathbb{R}^{L}$,
\begin{align}\label{eq:448}
\sum_{l=1}^{L} \zeta_l \mathbb{G}_n(t_l) \xrightarrow[]{d} \sum_{l=1}^{L} \zeta_l \mathbb{G}(t_l). \\[-5pt]
\end{align}

(ii). Asymptotic tightness of the process $\mathbb{G}_n(t)$ in $l^{\infty}(\mathcal{T})$,  i.e.,  for any $c > 0$,
\begin{align}\label{eq:449}
\lim_{\delta \downarrow 0} \limsup_{n \rightarrow \infty} P\left(\sup_{t,s \in \mathcal{T}, \; |t-s| \leq \delta} \left| \mathbb{G}_n(t) - \mathbb{G}_n(s) \right| > c \right) = 0. \\[-5pt]
\end{align}

\noindent With \eqref{eq:449}, the existence of an almost surely continuous sample path for $\mathbb{G}$ follows from Addendum 1.5.8 in \citet{van1996weak}.

\bigskip

\textit{Proof of step (i):} First note that for any $t \in \mathcal{T}, \mathbb{E}\left[V_i(t)\right] = 0$, and by Assumptions (A1)-(A3), $\text{Var}\left[V_i(t)\right] = \tau(1-\tau)\bm{a}' J_{\tau}(t)^{-1}\mathbb{E}[XX']J_{\tau}(t)^{-1}\bm{a} < \infty$. Therefore,
\begin{align}\label{eq:450}
\text{Var}\left[\sum_{l=1}^{L} \zeta_l \: V_i(t_l)\right] <  \infty,   \\[-5pt]
\end{align}
for any finite set of $\{t_1, \dots, t_L \}$ and $\{\zeta_1, \dots, \zeta_L\} \in \mathbb{R}^{L}$.  

Also, we have $\mathbb{E}\left[\sum_{l=1}^{L} \zeta_l \: V_i(t_l)\right] =  \sum_{l=1}^L \zeta_l \: \mathbb{E}\left[\vphantom{\left(\lambda\right)^2}V_i(t_l)\right] = 0$, and
\begin{align}
& \; \text{Var}\left[\sum_{l=1}^{L} \zeta_l \: V_i(t_l)\right] = \sum_{l=1}^L \sum_{l'=1}^L \zeta_l \: \zeta_{l'} \: \text{Cov}\left[V_i(t_l), \vphantom{\left(\lambda\right)^2} \: V_i(t_{l'})\right] \nonumber \\
= & \; \sum_{l=1}^L \sum_{l'=1}^L \zeta_l \: \zeta_{l'} \: \text{Cov}\left[\bm{a}' J_{\tau}(t_l)^{-1} \psi(Y(t_l), X; \: \bm{\beta}_{\tau}(t_l), \tau), \vphantom{\left(\lambda\right)^2} \: \bm{a}' J_{\tau}(t_{l'})^{-1} \psi(Y(t_{l'}), X; \: \bm{\beta}_{\tau}(t_{l'}), \tau) \right] \nonumber\\
= & \; \sum_{l=1}^L \sum_{l'=1}^L \zeta_l \: \zeta_{l'} \: \bm{a}' J_{\tau}(t_l)^{-1} \mathbb{E}\left[\psi(Y(t_l), X; \: \bm{\beta}_{\tau}(t_l), \tau) \cdot \psi(Y(t_{l'}), X; \: \bm{\beta}_{\tau}(t_{l'}), \tau)'\right] J_{\tau}(t_{l'})^{-1} \bm{a} \nonumber\\
= & \; \sum_{l=1}^L \sum_{l'=1}^L \zeta_l \: \zeta_{l'} \: H_{\tau}(t_l, t_{l'}; \: \bm{a}) = \text{Var}\left[\sum_{l=1}^{L} \zeta_l \: \mathbb{G}(t_l)\right]. \label{eq:451} \\[-5pt]
\end{align}

Given \eqref{eq:450} and \eqref{eq:451}, \eqref{eq:448} directly follows from the central limit theorem.

\bigskip

\textit{Proof of step (ii):} Consider the decomposition 
\begin{align}\label{eq:436}
&  \mathbb{G}_n(t) - \mathbb{G}_n(s) \\
= & \frac{1}{\sqrt{n}} \bm{a}' \! \left(J_{\tau}(t)^{-1} \! - \!J_{\tau}(s)^{-1}\right) \! \sum_{i=1}^{n} X_i\left(\mathbf{1}\{Y_i(t) \leq X_i'\bm{\beta}_{\tau}(t)\} \!-\! \tau\right) \! + \! \frac{1}{\sqrt{n}} \bm{a}' J_{\tau}(s)^{-1} \sum_{i=1}^{n} \! X_i \Delta_i(t,s), \nonumber\\[-5pt]
\end{align}
\begin{align}
\text{where} \: \Delta_i(t,s) = \mathbf{1}\{Y_i(t) \leq X_i'\bm{\beta}_{\tau}(t)\} \!-\! \mathbf{1}\{Y_i(s) \leq X_i'\bm{\beta}_{\tau}(s)\} = \mathbf{1}\{\eta_i(t) \leq 0\} \!-\! \mathbf{1}\{\eta_i(s) \leq 0\}. \\[-5pt]
\end{align}

By Lemma \ref{lemma6}, $\lVert J_{\tau}(t)^{-1} - J_{\tau}(s)^{-1} \rVert \leq C' \left|t-s \right|$ for some constant $C'$, so for $\forall \: L \geq 2$,
\begin{align}\label{eq:452}
& \; \mathbb{E}\left[\left|\bm{a}' \left(J_{\tau}(t)^{-1} - J_{\tau}(s)^{-1}\right) X_i\left(\mathbf{1}\{Y_i(t) \leq X_i'\bm{\beta}_{\tau}(t)\} - \tau\right) \right|^L \right] \\
\lesssim & \; \xi^{L-2} \: \mathbb{E} \left[\left|\bm{a}' \left(J_{\tau}(t)^{-1} - J_{\tau}(s)^{-1}\right) X_i \right|^2 \right] \\
= & \; \xi^{L-2} \: \bm{a}' \left(J_{\tau}(t)^{-1} - J_{\tau}(s)^{-1}\right) \mathbb{E}\left[X_i X_i'\right] \left(J_{\tau}(t)^{-1} - J_{\tau}(s)^{-1}\right) \bm{a} \\
\leq & \; \xi^{L-2} \: \lVert \left(J_{\tau}(t)^{-1} - J_{\tau}(s)^{-1}\right) \: \mathbb{E}\left[X_i X_i'\right] \: \left(J_{\tau}(t)^{-1} - J_{\tau}(s)^{-1}\right) \rVert \lesssim \xi^{L-2} \: |t-s|^2. \\[-5pt]
\end{align}

\noindent Also, observe that  $\mathbb{E}\left[\Delta_i(t,s)^2\:|X_i\right] =  \mathbb{E}\left[\mathbf{1}\{\eta_i(s) \cdot \eta_i(t) < 0\}\:|X_i \right] \leq c_0 |t-s|$ for $\forall \: X_i \in \mathcal{X}$ by Assumption (A6), so for $\forall \: L \geq 2$,
\begin{align}\label{eq:453}
\begin{split}
& \; \mathbb{E}\left[\left|\bm{a}' J_{\tau}(s)^{-1} X_i \Delta_i(t,s) \right|^L \right] \\
\lesssim & \; \xi^{L-2} \: \mathbb{E}\left[\left|\bm{a}' J_{\tau}(s)^{-1} X_i \Delta_i(t,s) \right|^2 \right] \\
= & \; \xi^{L-2} \: \bm{a}' J_{\tau}(s)^{-1} \mathbb{E}\left[X_i X_i'\: \Delta_i(t,s)^2\right] J_{\tau}(s)^{-1} \bm{a} \\
= & \; \xi^{L-2} \: \bm{a}' J_{\tau}(s)^{-1} \mathbb{E}\left[X_i X_i'\: \mathbb{E}\left[\Delta_i(t,s)^2\:|\:X_i\right] \right] J_{\tau}(s)^{-1} \bm{a} \lesssim \xi^{L-2} \: |t-s|. \\[-5pt]
\end{split}
\end{align}

\noindent Given \eqref{eq:452} and \eqref{eq:453}, $\forall \: t, s \in \mathcal{T}$, similar calculations as on page 3307 in \citet{chao2017quantile} yield
\begin{align}
\mathbb{E}\left[ \left|\mathbb{G}_n(t) - \mathbb{G}_n(s) \right|^4 \right] \lesssim  \frac{1}{n}|t-s| + |t-s|^2, \\[-5pt]
\end{align}
so for $|t-s| \geq \frac{1}{n^3}$, or equivalently $|t-s|^{1/3} \geq \frac{1}{n}$, 
\begin{align}
\mathbb{E}\left[ \left|\mathbb{G}_n(t) - \mathbb{G}_n(s) \right|^4 \right] \lesssim & \; |t-s|^{4/3}. \\[-5pt]
\end{align}

Now apply Lemma \ref{lemma7} with $\mathcal{T}=\mathcal{T},\; d(s,t) = |s-t|^{1/3}, \; \bar{\omega}_n = 2/n $ and $\Psi(x) = x^4$, for any $\delta > 0, \; \omega \geq \bar{\omega}_n$, we have
\begin{align}\label{eq:455}
\sup_{|s-t|^{1/3} \leq \delta^{1/3}} \left| \mathbb{G}_t - \mathbb{G}_s \right| \; \leq \; S_1(\omega) \: + 2 \sup_{|s-t|^{1/3} \leq \bar{\omega}_n, \: t \in \tilde{T}} \left| \mathbb{G}_t - \mathbb{G}_s \right|, \\[-10pt]
\end{align}
where the set $\tilde{T}$ contains at most $D(\bar{\omega}_n,d) = O(n^3)$ points, and $S_1(\omega)$ satisfies
\begin{align}\label{eq:456}
\lVert S_1(\omega) \rVert_4 \lesssim \int_{\omega_n/2}^{\omega} \epsilon^{-3/4} \text{d}\epsilon \; + (\delta^{1/3} + 2\bar{\omega}_n) \: \omega^{-3/2}. \\[-5pt]
\end{align}

\noindent For any $c > 0, \; \omega \geq \bar{\omega}_n$, it follows from \eqref{eq:456} and Markov's inequality that
\begin{align}
& \; \lim_{\delta \downarrow 0} \limsup_{n \rightarrow \infty} P(S_1(\omega) \geq c/2) \nonumber \\
\lesssim & \; \lim_{\delta \downarrow 0} \limsup_{n \rightarrow \infty} \frac{16}{c^4} \: \left[\int_0^{\omega} \epsilon^{-3/4} \text{d}\epsilon \; + (\delta^{1/3} + 4/n) \: \omega^{-3/2} \right]^4 = \frac{16}{c^4} \: \left[\int_0^{\omega} \epsilon^{-3/4} \text{d}\epsilon \right]^4. \label{eq:457} \\[-5pt]
\end{align}

\noindent We can make the r.h.s. of \eqref{eq:457} arbitrarily small by choosing small $\omega$. If we can show 
\begin{align}\label{eq:458}
\sup_{|t-s|^{1/3} \leq \bar{\omega}_n, \: t \in \tilde{T}} \left| \mathbb{G}_n(t) - \mathbb{G}_n(s) \right| = o_p(1), \\[-5pt]
\end{align}
then by \eqref{eq:455}, 
\begin{align}
\begin{split}
& \; \lim_{\delta \downarrow 0} \limsup_{n \rightarrow \infty} P\left(\sup_{t,s \in \mathcal{T}, \; |t-s| \leq \delta} \left| \mathbb{G}_n(t) - \mathbb{G}_n(s) \right| > c \right)  \\
\leq & \; \lim_{\delta \downarrow 0} \limsup_{n \rightarrow \infty} P(S_1(\omega) \geq c/2) + \limsup_{n \rightarrow \infty} P\left(2 \sup_{|t-s|^{1/3} \leq \bar{\omega}_n, \: t \in \tilde{T}} \left| \mathbb{G}_n(t) - \mathbb{G}_n(s) \right| > c/2 \right) = 0, \\[-5pt]
\end{split}
\end{align}
which is the asymptotic equicontinuity condition that we want to prove in \eqref{eq:449}. 

Therefore, it remains to prove \eqref{eq:458}. Recall the decomposition in \eqref{eq:436}, and $\lVert J_{\tau}(t)^{-1} - J_{\tau}(s)^{-1} \rVert \leq C' \left|t-s \right|$ for some constant $C'$ by Lemma \ref{lemma6}. Let $\epsilon_n \coloneqq \bar{\omega}_n^3$, we have that for any $t,\:s \in \mathcal{T}$, $|t-s| \leq \epsilon_n$,
\begin{align}
& \; \frac{1}{\sqrt{n}} \left| \bm{a}' \left(J_{\tau}(t)^{-1} - J_{\tau}(s)^{-1}\right) \sum_{i=1}^{n} X_i\left(\mathbf{1}\{Y_i(t) \leq X_i'\bm{\beta}_{\tau}(t)\} - \tau\right) \right| \nonumber \\
\leq & \; \frac{1}{\sqrt{n}} \: n \: C' \: \xi \: \left| t-s \right| = O(n^{-5/2}), \quad \text{a.s.} \label{eq:459} \\[-5pt]
\end{align}

Next, observe that $\forall \; t \in \tilde{T}$, we have a.s. for a constant $C_1$ independent of $t,\: s$ and $n$ such that
\begin{align}\label{eq:460}
\sup_{|t-s| \leq \epsilon_n} \frac{1}{\sqrt{n}} \left| \bm{a}' J_{\tau}(s)^{-1} \sum_{i=1}^{n} X_i \Delta_i(t,s) \right| \leq \frac{C_1}{\sqrt{n}} B_n(t,\: \epsilon_n), \\[-5pt]
\end{align}
where $B_n(t,\: \epsilon_n) \coloneqq \sum_{i=1}^n \mathbf{1}\{\text{At least one crossing with} \; y=0 \; \text{occurs in} \; \eta_i(s): \: |s-t| \leq \epsilon_n \}$.
By Assumption (A6), $\forall \: t \in \mathcal{T}, \: C > 0$, we have $P\left(B_n(t,\: \epsilon_n) \geq C \right) \leq P\left(\tilde{B}_n(\epsilon_n) \geq C \right)$, where $\tilde{B}_n(\epsilon_n) \sim \text{Bin} (n, c_1 \epsilon_n)$ for some constant $c_1$ that does not depend on $t$ or $n$. By applying the multiplicative Chernoff bound for the binomial random variable $\tilde{B}_n(\epsilon_n)$, for any $\nu > 1$, 
\begin{align}
P\left(B_n(t,\: \epsilon_n) \geq (1+\nu) n c_1 \epsilon_n \right) \leq P\left(\tilde{B}_n(\epsilon_n) \geq (1+\nu) n c_1 \epsilon_n \right)  \leq \exp\left(-\frac{1}{3} \nu n c_1 \epsilon_n \right). \\[-8pt]
\end{align}

\noindent Substitute in $\epsilon_n = \bar{\omega}_n^3 = 2^3 / n^3 $ and $\nu = 3A n^2 \log n / 2^3 c_1$ for some positive integer $A$, the inequality above becomes 
\begin{align}\label{eq:461}
P\left(B_n(t,\: \epsilon_n) \geq \frac{2^3 c_1}{n^2} + 3A \log n \right) \leq \exp\left(- A \log n \right) = n^{-A}. \\[-5pt]
\end{align}

\noindent By \eqref{eq:460} and \eqref{eq:461}, for any $t \in \tilde{T}$, 
\begin{align}
& \; P\left(\sup_{|t-s| \leq \epsilon_n} \frac{1}{\sqrt{n}} \left|\bm{a}' J_{\tau}(s)^{-1} \sum_{i=1}^{n} X_i \Delta_i(t,s) \right| \geq \frac{C_1}{\sqrt{n}} \left(\frac{2^3 c_1}{n^2} + 3A \log n \right) \right) \\
\leq & \; P\left(B_n(t,\: \epsilon_n) \geq \frac{2^3 c_1}{n^2} + 3A \log n \right) \leq n^{-A}. \\[-5pt]
\end{align}

\noindent Now recall that the set $\tilde{T}$ contains at most $O(n^3)$ points, so we have
\begin{align}
& \;P\left(\sup_{t\in \tilde{T}} \sup_{|t-s|\leq \epsilon_n} \frac{1}{\sqrt{n}} \left|\bm{a}' J_{\tau}(s)^{-1} \sum_{i=1}^{n} X_i \Delta_i(t,s) \right| \geq \frac{C_1}{\sqrt{n}} \left(\frac{2^3 c_1}{n^2} + 3A \log n \right) \right) \\
\leq & \; |\tilde{T}| \; P\left(\sup_{|t-s| \leq \epsilon_n} \frac{1}{\sqrt{n}} \left|\bm{a}' J_{\tau}(s)^{-1} \sum_{i=1}^{n} X_i \Delta_i(t,s) \right| \geq \frac{C_1}{\sqrt{n}} \left(\frac{2^3 c_1}{n^2} + 3A \log n \right) \right)
\leq O(n^{3-A}). \\[-5pt]
\end{align}

\noindent Choose $A=4$, there exists a constant $C_2$ such that
\begin{align}\label{eq:462}
\!\!\!\! P\left(\sup_{t\in \tilde{T}} \sup_{|t-s|\leq \epsilon_n} \frac{1}{\sqrt{n}} \left|\bm{a}' J_{\tau}(s)^{-1} \sum_{i=1}^{n} X_i \Delta_i(t,s) \right| \geq \frac{C_2}{\sqrt{n}} \log n \right) \leq O(n^{-1}), \; \text{a.s.} \\[-5pt]
\end{align}

Combining \eqref{eq:436}, \eqref{eq:459} and \eqref{eq:462}, we arrive at \eqref{eq:458}, which completes the proof.
\end{proof}

\medskip

\begin{proof}[Proof of Lemma \ref{lemma9}] 
We have
\begin{align} \label{eq:423}
& \; \mathbb{E}\left[\tilde{G}^2_n(t) \:|\: (X_i)_{i=1}^n \right] = \mathbb{E}\left[\vphantom{\tilde{G}^2_n(t)} \mathbb{G}^2_n(t) \:|\: (X_i)_{i=1}^n \right] \\
= & \; \frac{1}{n} \sum_{i=1}^{n} \left(\bm{a}'J_{\tau}(t)^{-1} X_i \right)^2 \mathbb{E}\left[\left(\mathbf{1}\{Y_i(t) \leq X_i'\bm{\beta}_{\tau}(t)\} - \tau\right)^2 \:|\: (X_i)_{i=1}^n \right] \\
= & \; \tau (1-\tau) \bm{a}'J_{\tau}(t)^{-1} \:\! \mathbb{E}_n \left[X_i X_i' \right]  \:\!  J_{\tau}(t)^{-1} \bm{a}. \\[-5pt]
\end{align}

\noindent In \eqref{eq:423}, we have that for $\forall \; t \in \mathcal{T}$,
\begin{align}
\begin{split}
& \; \bm{a}'J_{\tau}(t)^{-1} \:\! \mathbb{E}_n \left[X_i X_i' \right]  \:\!  J_{\tau}(t)^{-1} \bm{a} \\
\geq & \; \lVert J_{\tau}(t)^{-1} \bm{a} \rVert_2^2 \; \lambda_{\min}(\mathbb{E}_n \left[X_i X_i' \right]) \geq \frac{1}{\overline{f}^2 \lambda^2_{\max}(\mathbb{E}\left[XX'\right])} \; \lambda_{\min}(\mathbb{E}_n \left[X_i X_i' \right]), \\
& \; \bm{a}'J_{\tau}(t)^{-1} \:\! \mathbb{E}_n \left[X_i X_i' \right]  \:\!  J_{\tau}(t)^{-1} \bm{a} \\
\leq & \; \lVert J_{\tau}(t)^{-1} \bm{a} \rVert_2^2 \; \lambda_{\max}(\mathbb{E}_n \left[X_i X_i' \right]) \leq \frac{1}{f^2_{\min} \lambda^2_{\min}(\mathbb{E}\left[XX'\right])} \; \lambda_{\max}(\mathbb{E}_n \left[X_i X_i' \right]). \\[-5pt]
\end{split}
\end{align}

It then remains to prove that there exist constants $c_2 > c_1 > 0$ such that 
\begin{align}\label{eq:426}
\lambda_{\min}(\mathbb{E}_n \left[X_i X_i' \right]) \geq c_1, \quad \text{with probability approaching one}, \\[-5pt]
\end{align}
and 
\begin{align}\label{eq:427}
\lambda_{\max}(\mathbb{E}_n \left[X_i X_i' \right]) \leq c_2, \quad \text{a.s.} \\[-5pt]
\end{align}

To show \eqref{eq:426}, first note that $\lambda_{\min}(\mathbb{E}[XX'])$ is strictly positive by Assmuption (A1). For any given $\epsilon_2$ such that $0 < \epsilon_2 < \lambda_{\min}(\mathbb{E}[XX'])$, 
\begin{align}\label{eq:428}
\mathbb{E}[XX'] - \epsilon_2 I \prec \mathbb{E}_n \left[X_i X_i' \right], \quad \text{with probability approaching one}, \\[-5pt]
\end{align}
due to the element-wise convergence of $\mathbb{E}_n \left[X_i X_i' \right]$ to $\mathbb{E}[XX']$. \eqref{eq:428} leads to
\begin{align}
\lambda_{\min}(\mathbb{E}[XX']) - \epsilon_2 \leq \lambda_{\min}(\mathbb{E}_n \left[X_i X_i' \right]). \\[-5pt]
\end{align}

\noindent Taking $\epsilon_2 = \lambda_{\min}(\mathbb{E}[XX'])/2$, we then have $\lambda_{\min}(\mathbb{E}[XX'])/2 \leq \lambda_{\min}(\mathbb{E}_n \left[X_i X_i' \right])$ with probability approaching one, proving \eqref{eq:426}.

Given that there exists $\xi > 0$ such that $\lVert X\rVert \leq \xi$ a.s. by Assumption (A1), \eqref{eq:427} is straightforward using $\lambda_{\max}(\mathbb{E}_n \left[X_i X_i' \right]) = \sup_{\lVert \bm{u} \rVert =1} \bm{u}' \mathbb{E}_n \left[X_i X_i' \right] \bm{u} = \sup_{\lVert \bm{u} \rVert =1} \mathbb{E}_n \left[(\bm{u}' X_i)^2 \right]$. This completes the proof.
\end{proof}

\bigskip

\begin{proof}[Proof of Lemma \ref{lemma10}]
We basically follow the proof to Lemma 30 and Lemma 35 in \citet{belloni2019conditional}, and just sketch the key steps here.

First note that for any $r_n = o(1)$, we have
\begin{align}\label{eq:429}
\left\{\sup_{t \in \bm{t}} \lVert \bm{\hat{\beta}}_{\tau}(t) - \bm{\beta}_{\tau}(t) \rVert \leq r_n \right\} \subseteq \left\{\sup_{t \in \bm{t}} \lVert \hat{J}_{\tau}(t) - J_{\tau}(t) \rVert \leq \epsilon_1(n, T) + \epsilon_2(n, T) \right\}, \\[-5pt]
\end{align}
where $\hat{J}_{\tau}(t)$ is the Powell's estimator on $t \in \bm{t}$, and
\begin{align} \label{eq:430}
\epsilon_1(n, T) = & \;  \frac{1}{2\sqrt{n} h_n} \sup_{t \in \bm{t}, \; \lVert \bm{\beta} - \bm{\beta}_{\tau}(t) \rVert \leq r_n, \; \bm{u} \in S^{d-1} } \left| \mathbb{G}_n \left(\mathbf{1}\{ | Y_i(t) \leq X_i' \bm{\beta} | \leq h_n \} (\bm{u}' X_i)^2 \right) \right|, \\
\epsilon_2(n, T) = & \; \sup_{t \in \bm{t}, \; \lVert \bm{\beta} - \bm{\beta}_{\tau}(t) \rVert \leq r_n, \; \bm{u} \in S^{d-1} } \left| \frac{1}{2 h_n} \mathbb{E} \left[\mathbf{1}\{ | Y_i(t) \leq X_i' \bm{\beta} | \leq h_n \} (\bm{u}' X_i)^2 \right] - \bm{u}' J_{\tau}(t) \bm{u} \right|. \\[-5pt]
\end{align}
To bound $\epsilon_1(n, T)$, we define the class of functions
\begin{align} \label{eq:431}
\mathcal{G}_3(h) \coloneqq  \left\{ (X, \bm{Y}) \mapsto (\bm{u}'X)^2 \: \mathbf{1}\{ | Y(t) \leq X' \bm{\beta} | \leq h \} \; | \; \bm{u} \in \mathcal{S}^{d-1}, \; \bm{\beta} \in \mathbb{R}^{d}, \;  t \in \bm{t} \right\}. \\[-5pt]
\end{align}
If $h_n = o(1)$ and $\log T \log n = o(n h_n)$, similar arguments as used in the proof to Lemma 35 in \citet{belloni2019conditional} yield 
\begin{align} \label{eq:432}
\epsilon_1(n, T) \; = & \; \frac{1}{2\sqrt{n} h_n} \sup_{f \in \mathcal{G}_3(h_n)} \left| \mathbb{G}_n(f) \right| \lesssim_{P} \frac{1}{2\sqrt{n} h_n} (\log T \log n)^{1/2} \left( \sup_{f \in \mathcal{G}_3(h_n)} \mathbb{E} \left[f^2 \right] + \frac{\log T \log n}{n} \right)^{1/2} \\
\simeq & \; \frac{1}{2\sqrt{n} h_n} (\log T \log n)^{1/2} \left( h_n + \frac{\log T \log n}{n} \right)^{1/2} \simeq \left( \frac{\log T \log n}{n h_n} \right)^{1/2},  \\[-5pt] 
\end{align}
and 
\begin{align} \label{eq:433}
\epsilon_2(n, T) \lesssim r_n + h_n. \\[-5pt]
\end{align}
In the proof to Theorem 1, we have shown that 
\begin{align} \label{eq:434}
\sup_{t \in \bm{t}} \lVert \bm{\hat{\beta}}_{\tau}(t) - \bm{\beta}_{\tau}(t) \rVert \lesssim_{P} \left(\frac{\log T}{n} \right)^{1/2}. \\[-5pt]
\end{align}
Take $r_n = \left(\frac{\log T}{n} \right)^{1/2}$ and combine \eqref{eq:429} and \eqref{eq:434}, we have
\begin{align} \label{eq:435}
\sup_{t \in \bm{t}} \lVert \hat{J}_{\tau}(t) - J_{\tau}(t) \rVert \lesssim_{P} \left(\frac{\log T \log n}{n h_n} \right)^{1/2} + h_n. \\[-5pt]
\end{align}
And observe that for $\forall \; t \in \left[t_l, t_{l+1} \right]$,
\begin{align} \label{eq:437}
& \; \lVert \hat{J}_{\tau}(t) - J_{\tau}(t) \rVert  = \lVert \frac{t_{l+1}- t}{t_{l+1}-t_l} \hat{J}_{\tau}(t_l) + \frac{t - t_l}{t_{l+1}-t_l} \hat{J}_{\tau}(t_{l+1}) - J_{\tau}(t) \rVert \\
\leq & \;  \frac{t_{l+1}- t}{t_{l+1}-t_l} \lVert \hat{J}_{\tau}(t_l) - J_{\tau}(t_l) \rVert + \frac{t - t_l}{t_{l+1}-t_l} \lVert \hat{J}_{\tau}(t_{l+1}) - J_{\tau}(t_{l+1}) \rVert + \frac{t_{l+1}- t}{t_{l+1}-t_l} \lVert J_{\tau}(t_l) - J_{\tau}(t) \rVert + \\
& \; \frac{t - t_l}{t_{l+1}-t_l} \lVert J_{\tau}(t_{l+1}) - J_{\tau}(t) \rVert \lesssim_{P}  \left(\frac{\log T \log n}{n h_n} \right)^{1/2} + h_n + \delta_T, \\[-5pt]
\end{align}
where the last line follows from intermediate results in the proof to Lemma \ref{lemma6}.
\end{proof}

\bigskip

\bibliographystyle{apalike}
\bibliography{references}